\documentclass[reqno,a4paper,11pt]{amsart}
\usepackage[utf8]{inputenc}

\calclayout

\usepackage[utf8]{inputenc}
\usepackage{verbatim}
\usepackage[english]{babel}
\usepackage{amsmath}
\usepackage{amsfonts}
\usepackage{mathtools}
\usepackage{upgreek}

\usepackage{cite}

\usepackage{amsthm}
\usepackage{bm}
\usepackage{float}
\usepackage{color}
\usepackage{tikz}
\usetikzlibrary{decorations.pathreplacing,decorations.markings,calc}
\usetikzlibrary{arrows}
\tikzset{
  on each segment/.style={
    decorate,
    decoration={
      show path construction,
      moveto code={},
      lineto code={
        \path [#1]
        (\tikzinputsegmentfirst) -- (\tikzinputsegmentlast);
      },
      curveto code={
        \path [#1] (\tikzinputsegmentfirst)
        .. controls
        (\tikzinputsegmentsupporta) and (\tikzinputsegmentsupportb)
        ..
        (\tikzinputsegmentlast);
      },
      closepath code={
        \path [#1]
        (\tikzinputsegmentfirst) -- (\tikzinputsegmentlast);
      },
    },
  },
  mid arrow/.style={postaction={decorate,decoration={
        markings,
        mark=at position .5 with {\arrow[#1]{stealth}}
      }}},
}

\tikzset{
	on each segment/.style={
		decorate,
		decoration={
			show path construction,
			moveto code={},
			lineto code={
				\path [#1]
				(\tikzinputsegmentfirst) -- (\tikzinputsegmentlast);
			},
			curveto code={
				\path [#1] (\tikzinputsegmentfirst)
				.. controls
				(\tikzinputsegmentsupporta) and (\tikzinputsegmentsupportb)
				..
				(\tikzinputsegmentlast);
			},
			closepath code={
				\path [#1]
				(\tikzinputsegmentfirst) -- (\tikzinputsegmentlast);
			},
		},
	},
	mid arrow/.style={postaction={decorate,decoration={
				markings,
				mark=at position .7 with {\arrow[#1]{stealth}}
	}}},
	rmid arrow/.style={postaction={decorate,decoration={
				markings,
				mark=at position .3 with {\arrowreversed[#1]{stealth}}
	}}},
	vmid arrow/.style 2 args={postaction={decorate,decoration={
				markings,
				mark=at position #2 with {\arrow[#1]{stealth}}
	}}},
	vrmid arrow/.style 2 args={postaction={decorate,decoration={
				markings,
				mark=at position #2 with {\arrowreversed[#1]{stealth}}
	}}},
}

\makeatletter
\def\grd@save@target#1{%
  \def\grd@target{#1}}
\def\grd@save@start#1{%
  \def\grd@start{#1}}
\tikzset{
  grid with coordinates/.style={
    to path={%
      \pgfextra{%
        \edef\grd@@target{(\tikztotarget)}%
        \tikz@scan@one@point\grd@save@target\grd@@target\relax
        \edef\grd@@start{(\tikztostart)}%
        \tikz@scan@one@point\grd@save@start\grd@@start\relax
        \draw[minor help lines] (\tikztostart) grid (\tikztotarget);
        \draw[major help lines] (\tikztostart) grid (\tikztotarget);
        \grd@start
        \pgfmathsetmacro{\grd@xa}{\the\pgf@x/1cm}
        \pgfmathsetmacro{\grd@ya}{\the\pgf@y/1cm}
        \grd@target
        \pgfmathsetmacro{\grd@xb}{\the\pgf@x/1cm}
        \pgfmathsetmacro{\grd@yb}{\the\pgf@y/1cm}
        \pgfmathsetmacro{\grd@xc}{\grd@xa + \pgfkeysvalueof{/tikz/grid with coordinates/major step}}
        \pgfmathsetmacro{\grd@yc}{\grd@ya + \pgfkeysvalueof{/tikz/grid with coordinates/major step}}
        \foreach \x in {\grd@xa,\grd@xc,...,\grd@xb}
        \node[anchor=north] at (\x,\grd@ya) {\pgfmathprintnumber{\x}};
        \foreach \y in {\grd@ya,\grd@yc,...,\grd@yb}
        \node[anchor=east] at (\grd@xa,\y) {\pgfmathprintnumber{\y}};
      }
    }
  },
  minor help lines/.style={
    help lines,
    step=\pgfkeysvalueof{/tikz/grid with coordinates/minor step}
  },
  major help lines/.style={
    help lines,
    line width=\pgfkeysvalueof{/tikz/grid with coordinates/major line width},
    step=\pgfkeysvalueof{/tikz/grid with coordinates/major step}
  },
  grid with coordinates/.cd,
  minor step/.initial=.2,
  major step/.initial=1,
  major line width/.initial=2pt,
}

\usepackage{graphicx}
\usepackage{subcaption}
\usepackage{enumerate}
\usepackage{bm}

\usepackage[pdftex,colorlinks]{hyperref}
\hypersetup{
	colorlinks=true,
	linkcolor=blue,
	citecolor=red
}

\makeatletter
\def\l@subsection{\@tocline{2}{0pt}{2.5pc}{5pc}{}}

\DeclareMathOperator{\ai}{Ai}
\DeclareMathOperator{\re}{Re}
\DeclareMathOperator{\im}{Im}
\DeclareMathOperator{\ee}{\rm e}
\DeclareMathOperator{\supp}{supp}

\newcommand{\C}{\mathbb{C}}
\newcommand{\R}{\mathbb{R}}
\newcommand{\Z}{\mathbb{Z}}

\newcommand{\E}{\mathbb{E}}
\newcommand{\boh}{\mathit{o}}
\newcommand{\Boh}{\mathcal{O}}

\newcommand{\ii}{\mathrm{i}}
\newcommand{\dd}{\mathrm{d}}
\newcommand*{\deff}{\mathrel{\vcenter{\baselineskip0.5ex \lineskiplimit0pt
                     \hbox{\scriptsize.}\hbox{\scriptsize.}}}%
                     =}
\DeclareMathOperator{\Li}{Li}
\newcommand\mb[1]{\mathbb{#1}}
\renewcommand{\bm}{\mathbf}
\newcommand{\mcal}{\mathcal}

\newcommand{\msf}{\mathsf}

\newcommand{\wh}{\widehat}
\newcommand{\wt}{\widetilde}
\newcommand{\tp}{\mathrm T}
\newcommand{\ds}{\displaystyle}
\renewcommand{\sp}{\boldsymbol \sigma}

\newcommand\restr[1]{\raisebox{-.5ex}{$\big|$}_{#1}}

\newcommand{\hkpz}{\msf h^{\mathrm{\scriptscriptstyle (KPZ)}}}
\newcommand{\Phikpz}{{\bm \Phi}^{\mathrm{\scriptscriptstyle (KPZ)}}}

\newcommand{\Dkpz}{{\bm\Delta}^{\mathrm{\scriptscriptstyle (KPZ)}}}
\newcommand{\Tkpz}{T}
\newcommand{\skpz}{s}
\newcommand{\msfga}{\upgamma}
\newcommand{\fgue}{F_{\mathrm{GUE}}}

\newcommand{\hcc}{\msf h_0}
\newcommand{\Phicc}{{\bm \Phi}_0}
\newcommand{\Phiccp}{{\bm \Phi}_{0,+}}
\newcommand{\Phiccm}{{\bm \Phi}_{0,-}}
\newcommand{\Phiccpm}{{\bm \Phi}_{0,\pm}}
\newcommand{\lcc}{\lambda_0}
\newcommand{\Dcc}{{\bm\Delta}_0}
\newcommand{\Dccp}{{\bm\Delta}_{0,+}}

\newcommand{\Qcc}{\msf L^{\kern -0.15em \ai}}
\newcommand{\Psicc}{{\bm \Psi}_0}
\newcommand{\Psiccp}{{\bm \Psi}_{0,+}}
\newcommand{\Psiccm}{{\bm \Psi}_{0,-}}
\newcommand{\Psiccpm}{{\bm \Psi}_{0,\pm}}

\newcommand{\Phiai}{{\bm \Phi}_{\rm Ai}}

\newcommand{\cSig}{\msf \Sigma}
\newcommand{\tauad}{\tau}
\newcommand{\Sad}{\msf S}
\newcommand{\Tad}{\msf T}
\newcommand{\tad}{\msf t}
\newcommand{\sad}{\msf s}

\newcommand{\phig}{\upphi}
\newcommand{\phiad}{\Phi}

\usepackage{todonotes}

\newtheorem{theorem}{Theorem}[section]
\newtheorem{prop}[theorem]{Proposition}
\newtheorem{lemma}[theorem]{Lemma}
\newtheorem{corollary}[theorem]{Corollary}

\theoremstyle{definition}
\newtheorem{definition}[theorem]{Definition} 

\newtheorem{assumption}[theorem]{Assumptions} 

\theoremstyle{remark}
\newtheorem{remark}[theorem]{Remark} 
\numberwithin{equation}{section}

\begin{document}

\title[Universality of Hermitian Random matrices and the Int-Diff PII]{Universality for multiplicative statistics of Hermitian random matrices and the integro-differential Painlevé II equation}

\author[P.~Ghosal]{Promit Ghosal}
\address[PG]{Department of Mathematics, Massachusetts Institute of Technology, USA.}
\email{promit@mit.edu}

\author[G.~Silva]{Guilherme L.~F.~Silva}
\address[GS]{Instituto de Ciências Matemáticas e de Computação, Universidade de S\~ao Paulo (ICMC - USP), São Carlos, São Paulo, Brazil.}
\email{silvag@usp.br}

\date{}


\begin{abstract}
We study multiplicative statistics for the eigenvalues of unitarily-invariant Hermitian random matrix models. We consider one-cut regular polynomial potentials and a large class of multiplicative statistics. We show that in the large matrix limit several associated quantities converge to limits which are universal in both the potential and the family of multiplicative statistics considered. In turn, such universal limits are described by the integro-differential Painlevé II equation, and in particular they connect the random matrix models considered with the narrow wedge solution to the KPZ equation at any finite time.
\end{abstract}


\vspace*{-1.6cm}

\maketitle

\setcounter{tocdepth}{1}
\tableofcontents

\section{Introduction}

Random matrix theory has proven over time to be a powerful modern tool in mathematics and physics. With widespread applications in different areas such as engineering, statistical mechanics, probability, number theory, to mention only a few, its theory is rich and has been under intense development in the past thirty or so years. In a sense, much of the success of random matrix theory has been due to its exact solvability, or integrability, turning them into touchstones for predicting and confirming complex phenomenon in nature.

One of the most celebrated results in random matrix theory is the convergence of fluctuations of the largest eigenvalue towards the Tracy-Widom law $\fgue$. This result was first obtained by Tracy and Widom \cite{TW94} for matrices from the Gaussian Unitary Ensemble (GUE), who also showed that $\fgue$ is expressible in terms of a particular solution to the Painlevé II equation (shortly PII). Their findings sparked numerous advances in mathematics and physics, which began from the extension to several other matrix models but shortly afterwards widespread beyond the realm of random matrices.

 Starting with the celebrated Baik-Deift-Johansson Theorem \cite{baik_deift_johansson}, the distribution $\fgue$ has been identified as the limiting one-point distribution for the fluctuations of a wide range of different probabilistic models. One of the most ubiquitous of such models is the KPZ equation, introduced in the 1980s by Kardar, Parisi and Zhang. Despite numerous developments surrounding it, exactly solving it remained an outstanding open problem until the early 2010s, where four different groups of researchers \cite{AmirCorwinQuastel2011, SasamotoSpohn2010, Dotsenko2010, CalabreseLeDoussalRosso2010} independently found an exact solution for its so-called narrow wedge solution. Amongst these works, Amir, Corwin and Quastel \cite{AmirCorwinQuastel2011} found the one-point distribution for the height function of the KPZ solution, showing that it relates to a distribution found a little earlier by Johansson in a grand canonical Gaussian-type matrix model \cite{Johansson2007}, and further characterizing it in terms of the integro-differential Painlevé II equation. The latter is an extension of the PII differential equation, and almost as an immediate consequence the authors of \cite{AmirCorwinQuastel2011} also obtained that this one-point distribution, in the large time limit, converges to $\fgue$ itself.

In much inspired by \cite{baik_deift_johansson,johansson_2000} and later by \cite{AmirCorwinQuastel2011, SasamotoSpohn2010, Dotsenko2010, CalabreseLeDoussalRosso2010}, it has been realized that several stochastic growth models share an inherent connection with statistics of integrable point processes, in what is formally established as an identity between a transformation of the growth model and statistics for the point process. To our knowledge, the very first instance of such relation appears in the work of Borodin \cite{Borodin_2018} which connects higher spin vertex model with Macdonald measures. By taking appropriate limit of such connection, it was later found that the KPZ equation is connected to the Airy$_2$ point process \cite{BorodinGorin2016}, ASEP is related to the discrete Laguerre Ensemble, stochastic six vertex model is connected in the same way to the Meixner ensemble, or yet the Krawtchouk ensemble \cite{BorodinOlshanski2017}.

As a common feature to these connections, the underlying correspondences establish that the so-called $q$-Laplace transform of the associated height function coincides with some multiplicative statistics of the point process. The latter, in turn, admits exact solvability, and it is widely believed that a long list of new insights on the growth models can be obtained by studying the corresponding multiplicative statistics of the point processes. 

This program has already been taken to a great start for the KPZ equation, and exploring its connection with the Airy$_2$ point process Corwin and Ghosal \cite{CorwinGhosal2020} were able to obtain bounds for the lower tail of the KPZ equation. Shortly afterwards, such bounds were improved by Cafasso and Claeys \cite{CafassoClaeys2021} with Riemann-Hilbert methods common in random matrix theory.

Our major goal is to take on the program of understanding multiplicative statistics for random particle systems, and carry out its detailed asymptotic analysis for one of the most inspiring models, namely eigenvalues of random matrices.
 While fluctuations of linear statistics of the eigenvalues of random matrices has been extensively studied in the past, the study of multiplicative statistics has only been carried over in particular instances, remarkably when dealing with products and quotients of characteristic polynomials, see \cite{DesrosiersLiu2014, BaikDeiftStrahov2003, AkemannStrahovWurfel2020} and the references therein for a complete account. 

In this work, we consider the Hermitian matrix model with an arbitrary one-cut regular potential $V$, and associate to it a general family of multiplicative statistics on its eigenvalues, indexed by a function $Q$ satisfying certain regularity conditions. Our findings show that when the number of eigenvalues is large such multiplicative statistics become universal: in the large matrix limit they converge to a multiplicative statistics of the Airy$_2$ point process which is independent of $V$ and $Q$. This limiting statistics admits a characterization in terms of a particular solution to the integro-differential Painlevé II equation, and it  is the same quantity that connects the KPZ equation and the Airy$_2$ point process. So, in turn, we find that random matrix theory can recast the narrow wedge solution to KPZ equation for finite time in an universal way.

The random matrix statistics that we study are associated to a deformed orthogonal polynomial ensemble, also indexed by $Q$, which we analyze. As we learn from earlier work of Borodin and Rains \cite{Borodin_Rains_2005} which was recently rediscovered and greatly extended by Claeys and Glesner \cite{ClaeysGlesner2021} (and which we also briefly explain later on), this deformed ensemble is a conditional ensemble of a marked process associated to the original random matrix model. We show that the correlation kernel for this point process converges to a kernel constructed out of the same solution to the integro-differential PII equation that appeared in \cite{AmirCorwinQuastel2011}. This kernel is again universal in both $V$ and $Q$, and turns out to be the kernel of the induced conditional process on the marked Airy$_2$ point process. Naturally, there are orthogonal polynomials and their norming constants and recurrence coefficients associated to this deformed ensemble. With our approach we also obtain similar universality results for such quantities, showing that they are indeed universal in $V$ and $Q$ and also connect to the integro-differential PII in a neat way.

Beyond the concrete results, with this work we also hope to shed light into the rich structure underlying multiplicative statistics for eigenvalues of random matrices. Much of the recent relevance of Painlevé equations is due to its appearance in random matrix theory, see \cite{duits_painleve_kernels} for an overview of several of these connections. There has been a growing recent interest in integro-differential Painlevé-type equations \cite{BothnerCafassoTarricone2021,CafassoClaeysRuzza2021, Bothner2021, Krajenbrink2020, CharlierClaeysRuzza2021,LiechtyWang2020}, and our results place the integro-differential PII as a central universal object in random matrix theory as well.

We scale the multiplicative statistics to produce a critical behavior at a soft edge of the matrix model, and consequently the core of our asymptotic analysis lies within the construction of a local approximation to all the quantities near this critical point. Our main technical tool is the application of the Deift-Zhou nonlinear steepest descent method to the associated Riemann-Hilbert problem (shortly RHP), and the mentioned local approximation is the so-called construction of a local parametrix. In our case, a novel feat is that this local parametrix construction is performed in a two-step way, first with the construction of a model problem varying with large parameter, and second with the asymptotic analysis of this model problem. In the latter, a RHP recently studied by Cafasso and Claeys \cite{CafassoClaeys2021} (see also the subsequent works \cite{CafassoClaeysRuzza2021, CharlierClaeysRuzza2021}) which is related to the lower tail of the KPZ equation shows up, and it is this RHP that ultimately connects all of our considered quantities to the integro-differential PII.

The choice of scaling of our multiplicative statistics is natural, illustrative but not exhaustive. As we point out later, with our approach it becomes clear that other scalings could also be analyzed, say for instance scaling around a bulk point, or yet soft/hard edge points with critical potentials, and indicate that other integrable systems extending the integro-differential PII may emerge.

%
%

\section{Statement of main results}\label{sec:StatementResults}

Let $\Lambda^{(n)} \deff(\lambda_1<\ldots <\lambda_n)$ be a $n$-particle system with 
distribution
\begin{equation}\label{deff:classicalUE}
\frac{1}{Z_n}\prod_{1\leq j<k\leq n}(\lambda_k-\lambda_j)^2\prod_{j=1}^n \ee^{-nV(x)}\dd x,
\end{equation}
where $Z_n$ is the partition function
\begin{equation}\label{deff:nondeformedpartfunction}
Z_n\deff\int_{\R^n} \prod_{1\leq j<k\leq n}(\lambda_k-\lambda_j)^2\prod_{j=1}^n  \ee^{-n V(\lambda_j)}\dd \lambda_1\cdots \dd\lambda_n.
\end{equation}
%
The distribution \eqref{deff:classicalUE} is the eigenvalue distribution of the unitarily-invariant random matrix model with potential function $V$ \cite{deift_book,mehta_book}.

We associate to $\Lambda^{(n)}$ the multiplicative statistics
\begin{align}
\msf L_n^Q(\sad) \deff& \; \E\left(\prod_{j=1}^n \frac{1}{1+ \ee^{-\sad -n^{2/3} Q(\lambda_j)}} \right) \nonumber \\
 =& \; \frac{1}{Z_n}\int_{\R^n}  \prod_{1\leq j<k\leq n}(\lambda_k-\lambda_j)^2\prod_{j=1}^n  \sigma_{n}(\lambda_j)\ee^{-n V(\lambda_j)}\dd \lambda_1\cdots \dd\lambda_n = \frac{\msf Z^Q_n(\sad)}{Z_n},\quad \sad>0, \label{deff:Ln}
\end{align}
where
$\msf Z_n^Q(\sad)$ is the partition function for the deformed model
\begin{equation}\label{deff:DeformedPartFunction}
\msf Z_n^Q(\sad)\deff \int_{\R^n}  \prod_{1\leq j<k\leq n}(\lambda_k-\lambda_j)^2\prod_{j=1}^n  \sigma_{n}(\lambda_j)\ee^{-n V(\lambda_j)}\dd \lambda_1\cdots \dd\lambda_n
\end{equation}
and we denoted
\begin{equation*}
\sigma_{n}(z)=\sigma_{n}(z\mid \sad)\deff \left(1+ \ee^{-\sad-n^{2/3} Q(z) } \right)^{-1}.
\end{equation*}
When $Q$ is linear, with a straightforward change of parameters $\msf L^Q_n$ reduces to
\begin{equation}\label{eq:qlaplacegeneral}
\E\left(\prod_{x\in \mcal X}\frac{1}{1+\zeta q^{x}}\right),
\end{equation}
 where the expectation is over the set $\mcal X$ of configurations of points (that is, for us $\mcal X=\Lambda^{(n)}$) and $q\in (0,1)$ should be viewed as a parameter of the model and, in general, $\zeta\in \C$ is a free parameter. The expression \eqref{eq:qlaplacegeneral} may be viewed as a transformation of the point process, where $\zeta\in \C$ becomes the spectral variable of this transformation, and the matching $\zeta=\ee^{-\sad}$ motivates the distinguished role of $\sad$ in \eqref{deff:Ln}. In the context of random particle systems, this particular multiplicative statistics is associated to the notion of a $q$-Laplace transform \cite{BorodinCorwinMacProc,BorodinGorin2016,BorodinOlshanski2017,BorodinCorwinSasamoto2014} that we already mentioned in the Introduction, and it has been one of the key quantities in several outstanding recent progresses in asymptotics for random particle systems \cite{ImamuraSasamoto2019,CorwinGhosal2020,CafassoClaeys2021}.

 We work under the following assumptions.

\begin{assumption}\label{MainAssumption}\hfill 
\begin{enumerate}[(i)]
\item The potential $V$ is a nonconstant real polynomial of even degree and positive leading coefficient, and its equilibrium measure $\mu_V$ is one-cut regular, we refer to Section~\ref{sec:equilibriummeasure} below for the precise definitions. Performing a shift on the variable, without loss of generality we assume that the right-most endpoint of $\supp\mu_V$ is at the origin, so that
$$
\supp\mu_V=[-a,0],
$$
for some $a>0$.

\item The function $Q$ is real-valued over the real line, and analytic on a neighborhood of the real axis. We also assume that it changes sign at the right-most endpoint of $\supp\mu_V$, with 
\begin{equation}\label{eq:signchangeQ}
Q(x)>0 \text{ on }(-\infty,0), \quad Q(x)<0 \text{ on } (0,\infty),
\end{equation}
with $x=0$ being a simple zero of $Q$. A particular role is played by the negative value $Q'(0)$, so we set
\begin{equation}\label{deff:tQprime}
\msf t\deff -Q'(0)>0.
\end{equation}
%
%
\end{enumerate}
\end{assumption}

Although $\tad$ in Assumption~\ref{MainAssumption}--(ii) will have the interpretation of time, we stress that in this paper it will be kept fixed within a compact of $(0,+\infty)$ rather than being made large or small. 

For our results and throughout the whole work, we also talk about uniformity of several error terms with respect to $\tad$ in the sense that we now explain. Because $Q$ is analytic on a neighborhood of the real axis, analytic continuation shows that it is completely determined by its derivatives $Q^{(k)}(0)$, $k\geq 0$. When we say that some error is uniform in $\tad$ within a certain range, we mean uniform when we vary $Q$ as a function of $-Q'(0)=\tad$ while keeping all other derivatives $Q^{(k)}(0)$, $k\geq 2$, fixed.

The condition in Assumption~\ref{MainAssumption}--(i) is standard in random matrix theory and they are known to hold when, say, $V$ is a convex function \cite{Saff_book}. The one-cut assumption is made just for ease of presentation, as it simplifies the Riemann-Hilbert analysis at the technical level considerably. On the other hand, the regularity condition is used substantially in our arguments, but is standard in Random Matrix Theory literature and holds true generically \cite{KuijlaarsMcLaughlin2000}. Most of our results are of local nature near the right-most endpoint of $\supp\mu_V$ and could be shown to hold true for multi-cut potentials near regular endpoints as well, with appropriate but non-essential modifications.

Assumption~\ref{MainAssumption}--(ii) should be seen as specifying enough regularity on the multiplicative statistics, here indexed by this factor $Q$.  Because of condition (ii), we have the pointwise convergence
\begin{equation}\label{eq:limitingpert}
\sigma_n(x)\stackrel{n\to\infty}{\to}
\begin{cases}
0, & x>0, \\
1, & x<0,
\end{cases}
\end{equation}
which means that the introduction of the factor $\sigma_n$ in the original weight $\ee^{-nV}$ has the effect of producing an interpolation between this original weight and its cut-off version $\chi_{(-\infty,0)}\ee^{-nV}$, where from here onward $\chi_J$ is the characteristic version of a set $J$. Comparing the Euler-Lagrange conditions on the equilibrium problem induced by the weights $\ee^{-nV}$ and $\chi_{(-\infty,0)}\ee^{-nV}$, the observation we just made heuristically indicates that the factor $\sigma_n$ does not change the global behavior of eigenvalues. This may also be rigorously confirmed as an immediate consequence of our analysis, but we do not elaborate on this end.

On the other hand, introducing a local coordinate $u=-z/n^{2/3}$ near the origin, the approximation
$$
\ee^{-\sad-n^{2/3}Q(z)}\approx\ee^{-( \tad u+\sad)}
$$
goes through, and we see that there is a competition between the term $\sad$ and $Q(z)$ that affects the local behavior of the weight at the scale $\Boh(n^{-2/3})$ near the origin, which is the same scale for nontrivial fluctuations of eigenvalues around the same point. The main results that we are about to state concern obtaining the asymptotic behavior as $n\to \infty$ of several quantities of the model, and in particular they showcase how this term $Q$ affects the local scaling regime of the eigenvalues near the origin and leads to connections with the integro-differential Painlevé II equation as already mentioned.

A central object in this paper is the multiplicative statistics
\begin{equation}\label{deff:MultStatsAiry}
\Qcc(s,T)\deff\E \left(\prod_{j=1}^\infty \frac{1}{1+\ee^{T^{1/3}(s+\mathfrak{a}_j)}}\right),
\end{equation}
where the expectation is over the Airy$_2$ point process with random configuration of points $\{\mathfrak{a}_j\}$\cite{SpohnPrahofer2002}. The quantity $\Qcc$ admits two remarkable characterizations, which are also of particular interest to us. The first is the formulation via a Fredholm determinant, namely 
$$
\Qcc(s,T)=\det\left( \mb I-\mb K^{{\rm Ai}}_T \right)\restr{L^2(-s,\infty)},
$$
where $\mb K^{{\rm Ai}}_T$ is the integral operator on $L^2(-s,\infty)$ acting with the finite temperature (or fermi-type) deformation of the Airy kernel $\msf K_T^{{\rm Ai}}$, defined by
$$
\msf K^{\rm{Ai}}_T(u,v)\deff \int_{-\infty}^\infty \frac{\ee^{T^{1/3}\zeta}}{1+\ee^{T^{1/3}\zeta}}\ai(u+\zeta)\ai(\zeta+v)\dd \zeta, \quad u,v\in \R.
$$
The term `temperature' stems from the connection between the KPZ equation and the random polymer models. Despite the name finite temperature, the parameter $T$ here corresponds to the time in the KPZ equation, see \eqref{eq:KPZAiry} below. The Fredholm determinant $\det\left( \mb I-\mb K^{\rm{Ai}}_T \right)$ appeared for the first time in the work of Johansson \cite{Johansson2007} as the limiting process of a grand canonical (that is, when the number of particles/size of matrix is also random) version of a Gaussian random matrix model, and interpolates between the classical Airy kernel when $T\to +\infty$ and the Gumbel distribution when $T\to 0^+$ with $s$ scaled appropriately. In \cite[Remark~1.13]{Johansson2007} Johansson already raises the question on whether a related classical (that is, not grand canonical) matrix model has limiting local statistics that interpolate between Gumbel and Tracy-Widom, as a similar feature to $\det\left( \mb I-\mb K^{\rm{Ai}}_T \right)$. Since then, other works have found $\det\left( \mb I-\mb K^{\rm{Ai}}_T \right)$ to be the limiting distribution for fluctuations around the largest particle of a point process \cite{DeanLeDoussalMajumdarSchehr2015,CundenMezzadriOConnell2018,LiechtyWang2020,BeteaBouttier2019}. In common, these works consider specific models rather than obtaining $\det\left( \mb I-\mb K^{\rm{Ai}}_T \right)$ as the universal limit for a whole family of particle systems.

Another characterization of $\Qcc$ is via a Tracy-Widom type formula that relates it to the integro-differential PII. It reads
\begin{equation}\label{eq:intdiffRepLAiry}
\log \Qcc(-\Sad\Tad^{1/3},\Tad^{-2})=-\frac{1}{\Tad}\int_{\Sad}^{\infty}(v-\Sad) \left(\int_{-\infty}^\infty \phiad(r\mid v,\Tad)^2 \frac{\ee^{-r}}{(1+\ee^{-r})^2}\dd r - \frac{v}{2}\right) \dd v,
\end{equation}
where $\phiad$ solves the integro-differential Painlevé II equation
\begin{equation}\label{eq:intdiffPII}
\partial^2_\Sad \phiad(\xi\mid \Sad,\Tad)=\left(\xi+\frac{\Sad}{\Tad}+\frac{2}{\Tad}\int_{-\infty}^\infty \phiad(r\mid \Sad,\Tad)^2\frac{\ee^{-r}}{(1+\ee^{-r})^2}\dd r \right)\phiad(\xi\mid \Sad,\Tad)
\end{equation}
with boundary value
$$
\phiad(\xi\mid \Sad,\Tad)\sim \Tad^{1/6}\ai(\Tad^{2/3}\xi+\Sad\Tad^{-1/3}),\quad \text{as } \xi\to\infty \text{ with } |\arg \xi|<\pi-\delta,
$$
for any $\delta>0$. This characterization has been obtained in the already mentioned work by Amir, Corwin and Quastel \cite{AmirCorwinQuastel2011}, in connection with the narrow wedge solution to the KPZ equation, and following the work \cite{BorodinGorin2016} by Borodin and Gorin has the interpretation that we now describe. For $\mcal H(X,\Tkpz)$ being the Hopf-Cole solution to the KPZ equation with narrow wedge initial data at the space-time point $(X,T)$, introduce the rescaled random variable
$$
\Upsilon_\Tkpz=\frac{\mcal H(0,2\Tkpz)+\frac{\Tkpz}{12}}{\Tkpz^{1/3}}.
$$
Based on the previous works \cite{AmirCorwinQuastel2011,CalabreseLeDoussalRosso2010,SasamotoSpohn2010,Dotsenko2010}, in \cite{BorodinGorin2016} the identity
\begin{equation}\label{eq:KPZAiry}
\E_{\mathrm{KPZ}}\left(\ee^{-\ee^{\Tkpz^{1/3}(\Upsilon_\Tkpz+\skpz)}}\right)=\Qcc(s,T)
\end{equation}
between the height function of the KPZ equation and the multiplicative statistics $\Qcc(s,T)$ is identified. This is an instance of matching formulas relating growth processes with determinantal point processes that we already mentioned at the Introduction. One of the key aspects of this representation is that the Airy$_2$ point process is determinantal, and consequently its statistics can be studied using techniques from exactly solvable/integrable models. Indeed Equation~\ref{eq:KPZAiry} is the starting point taken by Cafasso and Claeys \cite{CafassoClaeys2021}, who then connected $\Qcc(\skpz,\Tkpz)$ to a RHP that will also play a major role for us. Recently, Cafasso, Claeys and Ruzza \cite{CafassoClaeysRuzza2021} also obtained an independent proof of the representation \eqref{eq:intdiffRepLAiry}, extending it to more general multiplicative statistics of the Airy$_2$ point process. Other proofs and extensions of this integro-differential equation have also been recently found in related contexts \cite{BothnerCafassoTarricone2021, Krajenbrink2020, Bothner2021}. Also, by exploring \eqref{eq:KPZAiry} the tail behavior of the KPZ equation has become rigorously accessible in various asymptotic regimes \cite{CorwinGhosal2020,CafassoClaeys2021,CafassoClaeysRuzza2021,CharlierClaeysRuzza2021}.

As a first result, we prove that the multiplicative statistics $\Qcc$ is the universal limit of $\msf L_n^Q(\sad)$.
\begin{theorem}\label{thm:asymptoticsqLaplacetransf}
Suppose that $V$ and $Q$ satisfy Assumptions~\ref{MainAssumption} and fix $\sad_0>0$ and $\tad_0\in (0,1)$. For a constant $\msf c_V>0$ that depends solely on $V$, and any $\nu\in (0,2/3)$, the asymptotic estimate
\begin{equation}\label{eq:AsymptFormLQLAi}
\log \msf L^Q_n\left(\sad\right)=\log\Qcc\left(-\frac{\sad\msf c_V}{\tad},\frac{\tad^3}{\msf c_V^3}\right)+\Boh(n^{-\nu}),\quad n\to\infty
\end{equation}%
holds true uniformly for $\sad\geq -\sad_0$ and $\tad_0\leq \tad \leq 1/\tad_0$.
\end{theorem}

Findings on random matrix theory surrounding the Tracy-Widom distribution  have inspired an enormous development in the KPZ universality theory. For the KPZ equation one of the major developments can be phrased by saying that the fluctuations of the height function for the narrow wedge solution coincide, in the large time limit, with the $\beta=2$ Tracy-Widom law from random matrix theory. Theorem~\ref{thm:asymptoticsqLaplacetransf} is saying that the connection between random matrix theory and the KPZ equation can be recast already at any finite time, and not only for Gaussian models but also universally in $V$ and $Q$. Similar connection exists \cite{BBCW18} between the solution of the KPZ equation in half-space under the Robin boundary condition and Airy$_1$ point process which, in turn, in the large time limit relate this KPZ solution to GOE matrices.


The constant $\msf c_V$ is determined from \eqref{eq:spectralcurvehV} and \eqref{eq:derivativepsicpsi} below. It is the first derivative of a conformal map near the origin, which is constructed out of the equilibrium measure for $V$. Ultimately, we make a conformal change of variables of the form $\zeta\approx\msf c_V z/n^{2/3}$, which in turn identifies
$$
\frac{1}{1+\ee^{-\sad-n^{2/3}Q(z)}}\approx \frac{1}{1+\ee^{-\sad +\tad \zeta/\msf c_V}}.
$$
In light of \eqref{deff:MultStatsAiry}, this explains the evaluation $s=-\sad \msf c_V/\tad$ and $T=\tad^3/\msf c_V^3$ on the right-hand side of \eqref{eq:AsymptFormLQLAi}.

We emphasize that the error term in \eqref{eq:AsymptFormLQLAi} is not in sharp form. In Section~\ref{sec:issues} we explain how this term arises from our techniques. We do not have indications regarding whether the true optimal error would be $\Boh(n^{-2/3})$ (or of any polynomial order) or if it should involve, say, logarithmic corrections. 

Our next results concern limiting asymptotic formulas for the matrix model underlying $\msf L_n^Q$, starting with the partition function $\msf Z_n^Q$ from \eqref{deff:nondeformedpartfunction}. For a polynomial $V$, its deformation
\begin{equation}\label{eq:deformedGaussianpotential}
V_t(z)\deff \left(1-\frac{1}{t}\right)z^2+V(z/t),\quad 1\leq t\leq +\infty
\end{equation}
has the property that $V_1=V$ and $V_\infty=z^2$ is Gaussian. Under the assumption that $V_t$ is one-cut regular for every $t\geq 1$, Bleher and Its \cite{BleherIts2005} proved that an expansion of the form
\begin{equation}\label{eq:asymptPartBleherIts}
Z_n=Z_n^{\mathrm{GUE}}\ee^{n^2\bm e_0^V+\bm e_1^V}\left(1+\Boh(n^{-2})\right),\quad n\to\infty,
\end{equation}%
holds, where $Z_n^{\mathrm{GUE}}$ is the GUE partition function for $V_\infty=z^2$ and $\bm e_0^V$ and $\bm e_1^V$ are functions analytic on the coefficients of $V$. In fact, their result ensures a full asymptotic expansion in inverse powers of $n^2$, see also \cite{bessis_itzykson_zuber,ercolani_mclaughlin_partition_function,BleherEynard2003} for important earlier work obtaining similar results under different conditions, and also the more recent contributions \cite{bleher_deano_partition_function,GuionnetBorot2013}.

As an immediate corollary to Theorem~\ref{thm:asymptoticsqLaplacetransf} we obtain some terms in the asymptotic expansion of the deformed partition function \eqref{deff:DeformedPartFunction}.

\begin{corollary}
Suppose $V$ and $Q$ satisfy Assumptions~\ref{MainAssumption} and fix $\sad_0>0$ and $\tad_0>0$. In addition, assume that the potential $V_t$ in \eqref{eq:deformedGaussianpotential} is one-cut regular for any $t\geq 1$. For any $\nu\in (0,2/3)$, the partition function admits an expansion of the form
\begin{equation}\label{eq:DefPartFctionAsympt}
\msf Z_n^Q(\sad)=Z_n^{\mathrm{GUE}}\ee^{n^2 \bm e_1^V +\bm e_0^V}\Qcc\left(-\frac{\sad\msf c_V}{\tad},\frac{\tad^3}{\msf c_V^3}\right)\left(1+\Boh(n^{-\nu})\right),\quad n\to\infty
\end{equation}
which is valid uniformly for $\sad\geq -\sad_0$ and $\tad_0\leq \tad\leq 1/\tad_0$, and the coefficients $\bm e_0^V$ and $\bm e_1^V$ are as in \eqref{eq:asymptPartBleherIts}.
\end{corollary}

\begin{proof}
Follows from the expansion \eqref{eq:asymptPartBleherIts}, Theorem~\ref{thm:asymptoticsqLaplacetransf} and \eqref{deff:Ln}.
\end{proof}

The order of error \eqref{eq:DefPartFctionAsympt} is not $n^{-2}$ as in \eqref{eq:asymptPartBleherIts} but weaker and not sharp. This phenomenon can be traced back to the fact that $\sigma_n$ has infinitely many poles accumulating on the real axis as $n\to \infty$, see the discussion in Section~\ref{sec:issues} below. A similar error order was obtained in \cite[Theorem~9.1 and Equation~(9.68)]{BleherIts2005}, in a transitional regime from a one-cut to two-cut potential, and where the role played here by $\Qcc$ is replaced by the GUE Tracy-Widom distribution itself.

From the general theory of unitarily invariant random matrix models, it is known that the density appearing in \eqref{deff:Ln} admits a determinantal form. Setting
\begin{equation}\label{def:perturbedweight}
\omega_{n}(z)=\omega^Q_n(z\mid \sad)\deff\sigma_{n}(z)\ee^{-n V(z)},\quad \text{where we recall}\quad \sigma_{n}(z)=\left(1+ \ee^{-\sad- n^{2/3}Q(z) } \right)^{-1},
\end{equation}
this means that the identity
$$
\frac{1}{\msf Z_n(\sad)}
\prod_{1\leq j<k\leq n}(\lambda_k-\lambda_j)^2\prod_{j=1}^n  \sigma_{n}(\lambda_j)\ee^{-n V(\lambda_j)}
=
\frac{1}{n!}\det\left(\omega_n(\lambda_j)^{1/2}\msf K^Q_n(\lambda_j,\lambda_k)\omega_n(\lambda_k)^{1/2}\right)_{j,k=1}^n
$$
holds true for a function of two variables $\msf K^Q_n(x,y)$ satisfying certain properties, known as the correlation kernel of the eigenvalue density on the left-hand side. The correlation kernel is not unique, but in the present setup it may be taken to be the Christoffel-Darboux kernel for the orthogonal polynomials for the weight $\omega_n$, as we introduce in detail in \eqref{deff:Kndefweight}, and whenever we talk about $\msf K_n^Q$ we mean this Christoffel-Darboux kernel. In particular, $\msf K^Q_n=\msf K^Q_n(\cdot\mid \sad)$ does depend on both $Q$ and $\sad$.

Our second result proves universality of the kernel $\msf K_n^Q$, showing that its limit depends solely on $\sad$ and $\tad=-Q'(0)$, but not on other aspects on $Q$, and relates to the integro-differential PII. For its statement, it is convenient to introduce the new set of variables
\begin{equation}\label{eq:corr_tad_Tad}
\Tad=\tad^{-3/2}\quad \text{and}\quad \Sad=\sad\tad^{-3/2}.
\end{equation}
With $\phiad(\xi)=\phiad(\xi\mid \Sad,\Tad)$ being the solution to the integro-differential Painlevé II equation in \eqref{eq:intdiffPII} and the variables $\sad,\tad$ and $\Sad,\Tad$ related by \eqref{eq:corr_tad_Tad},  we set
$$
\phig_1(\zeta\mid \sad,\tad)=\phiad(\xi(\zeta)\mid \Sad,\Tad),\quad \phig_2(\zeta\mid \sad,\tad)=(\partial_\Sad \phiad)(\xi(\zeta)\mid \Sad,\Tad),\quad \quad \xi(\zeta)\deff -\sad+\tad\zeta,
$$
and introduce the kernel
$$
\msf K_\infty(u,v\mid \sad,\tad)\deff  \frac{\phig_1(v\mid \sad,\tad)\phig_2(u\mid \sad,\tad)-\phig_1(u\mid \sad,\tad)\phig_2(v\mid \sad,\tad)}{u-v},\quad u,v\in \R.
$$

\begin{theorem}\label{thm:limitingkernel}
Assume that $V$ and $Q$ satisfy Assumptions~\ref{MainAssumption} and fix $\sad_0>0$ and $\tad_0\in (0,1)$. With
\begin{equation}\label{deff:scalingunvn}
u_n\deff\frac{u}{\msf c_V n^{2/3}},\quad v_n\deff\frac{v}{\msf c_V n^{2/3}},
\end{equation}
the estimate
\begin{equation}\label{eq:LimitKnQFinalForm}
\frac{\ee^{-\frac{n}{2}(V(u_n)+V(v_n))}}{\msf c_V n^{2/3}} \msf K^Q_n\left(u_n,v_n\mid \sad\right)= \msf K_\infty(u,v\mid \sad,\tad/\msf c_V)+ \Boh(n^{-1/3}),\quad n\to\infty,
\end{equation}
holds true uniformly for $u,v$ in compacts of $\R$, and uniformly for $\sad\geq -\sad_0$ and $\tad_0\leq \tad\leq 1/\tad_0$. 
\end{theorem}

In the recent work \cite{ClaeysGlesner2021}, Claeys and Glesner developed a general framework for certain conditional point processes, which in particular yields a probabilistic interpretation of the kernel $\msf K_n^Q$ as we now explain. For a point process $\Lambda$, we add a mark $0$ to a point $\lambda\in \Lambda$ with probability $\sigma_n(\lambda)$ and a mark $1$ with complementary probability $1-\sigma_n(\lambda)$. This induces a decomposition of the point process $\Lambda=\Lambda_0\cup \Lambda_1$, where $\Lambda_j$ is the set of eigenvalues with mark $j$. We then consider the induced point process $\wh\Lambda$ obtained from $\Lambda$ upon conditioning that $\Lambda_1=\emptyset$, that is, that all points have mark $0$. 

When applied to the eigenvalue point process $\Lambda=\Lambda^{(n)}$ induced by the distribution \eqref{deff:classicalUE}, the theory developed in \cite{ClaeysGlesner2021} shows that $\wh\Lambda^{(n)}$ is a determinantal point process with correlation kernel proportional to $\omega_n(x)^{1/2}\omega_n(y)^{1/2}\msf K_n^Q(x,y)$ which, in turn, generates the same point process as the left-hand side of \eqref{eq:LimitKnQFinalForm}, see \cite[Sections~4 and 5]{ClaeysGlesner2021}. A comparison of the RHP that characterizes the kernel $\msf K_\infty$ (see Section~\ref{sec:KPZRHPKPZ} below, in particular \eqref{eq:transfPhicckernelphicc}) with the discussion in \cite[Section~5.2]{ClaeysGlesner2021} shows that $\msf K_\infty$ is a (renormalized) correlation kernel for the marked point process $\wh{\{\mathfrak a_k\}}$ of the Airy$_2$ point process $\{\mathfrak a_k\}$ with the marking function $(1+\ee^{-\sad+\tad \lambda})^{-1}$. So Theorem~\ref{thm:limitingkernel} assures that the conditional process on the marked eigenvalues converges, at the level of rescaled correlation kernels, to the conditional process on the marked Airy$_2$ point process.

We also obtain asymptotics for the norming constant $\msfga^{(n,Q)}_{n-1}(\sad)$ for the ${(n-1)}$-th monic orthogonal polynomial for the weight $\omega_n(x\mid \sad)$ (see \eqref{def:normingdeformed} for the definition), showing that its first correction term depends again solely on $\sad,\tad$, and also relates to the integro-differential Painlevé II equation.

\begin{theorem}\label{thm:asymptoticsnormingconstant}
Suppose that $V$ and $Q$ satisfy Assumptions~\ref{MainAssumption} and fix $\sad_0>0$ and $\tad_0\in (0,1)$. The norming constant has asymptotic behavior
\begin{equation}\label{eq:asymptformNormingctt}
\msfga^{(n,Q)}_{n-1}(\sad)^2=\frac{a}{4\pi }\ee^{-2n\ell_V}\left(\frac{1}{2}-\frac{1}{n^{1/3}}\frac{\msf c_V^{1/2}}{\tad^{1/2}}
\left(\msf p(\sad,\tad/\msf c_V)-\frac{\sad^2\msf c_V^{3/2}}{4\tad^{3/2}}\right)+\Boh(n^{-2/3})
\right), \quad n\to \infty,
\end{equation}
uniformly for $\sad\geq -\sad_0$ and $\tad_0\leq \tad\leq 1/\tad_0$, where $\msf p(\sad,\tad)=\msf P(\Sad,\Tad)$, and the function $\msf P$ relates to the solution $\phiad$ from \eqref{eq:intdiffPII} via
\begin{equation}\label{eq:intdiffPIIPhip}
\partial_\Sad \msf P(\Sad,\Tad)=\frac{\Sad}{2\Tad}+\frac{1}{\Tad}\int_{-\infty}^\infty \phiad(r\mid \Sad,\Tad)^2\frac{\ee^{-r}}{(1+\ee^{-r})^2}\dd r.
\end{equation}
\end{theorem}

Our approach also yields asymptotic formulas for the orthogonal polynomials and their recurrence coefficients, and relate them to the integro-differential Painlevé II equation as well, but for the sake of brevity we do not state them.

\section{About our approach: issues and extensions}

\subsection{Issues to be overcome}\label{sec:issues}\hfill 

Our main tool for obtaining all of our results is the Fokas-Its-Kitaev \cite{FokasItsKitaev92} Riemann-Hilbert Problem (RHP) for orthogonal polynomials (shortly OPs) that encodes the correlation kernel $\msf K_n^Q$, the norming constants $\msfga^{(n,Q)}_{n-1}(\sad)^2$ and ultimately also the multiplicative statistics $\msf L_n^Q$, and its asymptotic analysis via the Deift-Zhou nonlinear steepest descent method \cite{Deiftetal1999,DeiftZhou93mKdV}. The overall arch of this asymptotic analysis is the usual one, summarized in the diagram in Figure~\ref{fig:diagramRHP}, and we now comment on its major steps.
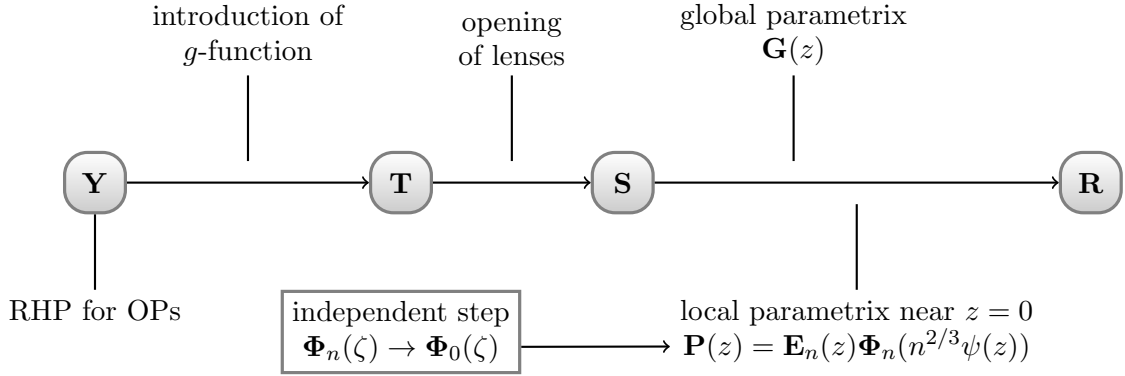
\begin{figure}[t]
\centering
		\begin{tikzpicture}[every text node part/.style={align=center}, node distance=10mm and 2mm,
                    terminal/.style={
                      rectangle,minimum size=8mm,rounded corners=3mm,
                      very thick,draw=black!50,
                      top color=white,bottom color=black!20},
                      interm/.style={
                      rectangle,minimum size=8mm,rounded corners=0mm,
                      very thick,draw=black!50,
                      top color=white,bottom color=white}
                      ]		
		
%
%
%
\node (Y) [terminal] {$\bm Y$};
\node (OP) [below=of Y] {RHP for OPs};
\node (g) [above right=of Y] {introduction of \\ $g$-function};
\node (YT) [below=of g]{};
\node (T) [terminal, below right=of g] {$\bm T$};
\node (l) [above right=of T] {opening \\ of lenses};
\node (TS) [below=of l] {};
\node (S) [terminal, below right=of l] {$\bm S$};
\node (lp) [below right=of S] {local parametrix near $z=0$\\ $\bm P(z)=\bm E_n(z)\bm \Phi_n(n^{2/3}\psi(z))$};
\node (SR) [above=of lp]{};
\node (gp) [above right=of S] {global parametrix\\ $\bm G(z)$};
\node (SRg) [below=of gp]{};
\node (SS) [interm, below=of T]{independent step\\ $\bm\Phi_n(\zeta)\to \Phicc(\zeta)$};
\node (R) [terminal, above right=of lp] {$\bm R$};
\path [thick] (Y)  edge[->] (T);
\path [thick] (g) edge[-] (YT.center);
\path [thick] (T)  edge[->] (S);
\path [thick] (l) edge[-] (TS.center);
\path [thick] (S)  edge[->] (R);
\path [thick] (lp) edge[-] (SR.center);
\path [thick] (gp) edge[-] (SRg.center);
\path [thick] (OP) edge[-] (Y);
\path [thick] (SS) edge[->,transform canvas={yshift=-2mm}] (lp);
%
\end{tikzpicture}
\caption{Schematic diagram for the steps in the asymptotic analysis of the RHP for orthogonal polynomials. There is also an Airy local parametrix used near $z=-a$ which we omit in this diagram.}\label{fig:diagramRHP}
\end{figure}

Starting with the RHP for OPs that we name $\bm Y$, in the first step we transform $\bm Y\mapsto \bm T$ with the introduction of the $g$-function (or, equivalently, the $\phi$-function), and this is done so with the help of the equilibrium measure for $V$ that accounts only for the part $\ee^{-nV}$ of the weight $\omega_n$. In the second step, we open lenses with a transformation $\bm T\mapsto \bm S$ as usual. 

The third step is the construction of the global parametrix $\bm G$. In our case, in this construction we also have to account for the perturbation $\sigma_n$ of the weight $\ee^{-nV}$, so a Szeg\"o function-type construction is used. 

The fourth step is the construction of local parametrices at the endpoints $z=-a,0$ of $\supp\mu_V$, with the goal of approximating all the quantities locally near these endpoints. This is accomplished by, first, considering a change of variables $z\mapsto \zeta=n^{2/3}\psi(z)$ after the conformal map $\psi$ chosen appropriately for each endpoint and, then, constructing the solution to a model RHP $\bm\Phi(\zeta)$ in the $\zeta$-plane. Following these steps, the local parametrix at the left edge $z=-a$ of $\supp\mu_V$ is standard and utilizes Airy functions.

The construction of the local parametrix at the right edge $z=0$ is, however, a lot more involved. As we mentioned earlier, the factor $\sigma_n$ affects asymptotics of local statistics near the origin. In fact, the weight $\sigma_n$ has singularities precisely at the points of the form
$$
-\frac{\sad}{n^{2/3}}+\frac{\pi \ii(2k+1)}{n^{2/3}},\quad k\in \Z.
$$
This means that for $\sad\ll n^{2/3}$ there are infinitely many poles of $\sigma_n$ accumulating near the real axis. As such, in this case for large $n$ the perturbed weight $\omega_n$ fails to be analytic in any fixed neighborhood of the origin. If we were to consider only $\sad\to +\infty$ fast enough, one could still push the standard RHP analysis further with the aid of Airy local parametrices, at the cost of a worse error estimate. However, when $\sad=\Boh(1)$ we have poles accumulating too fast to the real axis, and a different asymptotic analysis has to be accomplished, in particular a new local parametrix is needed.

When changing coordinates $z\mapsto \zeta$ near $z=0$, the model problem $\bm\Phi=\bm\Phi_n$ obtained is then $n$-dependent. This is so because the jump of the local parametrix involves $\sigma_n$, and consequently in the process of changing variables the resulting model problem has a jump that involves a transformation of $\sigma_n$ itself. This is in contrast with usual constructions with, say, Airy, Bessel or Painlevé-type parametrices, where the jumps can be turned into piecewise constant in the $z$-plane and, hence, also remain piecewise constant in the $\zeta$-plane. Another feature of the RHP for the model problem $\bm\Phi_n$ is that its jump is not analytic on the whole plane, and instead it is analytic only in a growing (with $n$) disk, and for a fixed $n$ we can only ensure that its jump matrix is $C^\infty$ in a neighborhood of the jump contour.

All in all, this means that carrying out the asymptotic analysis of $\bm \Phi_n(\zeta)$ as $n\to \infty$ is also needed. As we said, the jump for $\bm\Phi_n$ involves a transformation of $\sigma_n$, so ultimately also depends on the function $Q$ from \eqref{def:perturbedweight}. But it turns out that as $n\to\infty$, we have the convergence $\bm \Phi_n\to \bm\Phi_0$ in an appropriate sense, where $\bm \Phi_0$ is independent of $Q$. This limiting $\bm \Phi_0$ is the solution to a RHP that appeared recently in connection with the KPZ equation \cite{CafassoClaeys2021} and which was later shown to connect with the integro-differential PII in the recent work of Claeys, Cafasso and Ruzza \cite{CafassoClaeysRuzza2021,CharlierClaeysRuzza2021}. For this reason we term it the id-PII RHP.

With the construction of the global and local parametrices, the asymptotic analysis is concluded in the usual way, by patching them together and obtaining a new RHP for a matrix function $\bm R$. This matrix $\bm R$, in turn, solves a RHP whose jump is asymptotically close to the identity, and consequently $\bm R$ can be found perturbatively.

After concluding this asymptotic analysis, we undress the transformations $\bm R \mapsto \cdots \mapsto \bm Y$ and obtain asymptotic expressions for the wanted quantities. For the kernel $\msf K_n^Q$ and the norming constant $\msfga^{(n,Q)}_{n-1}(\sad)^2$, after this undressing Theorems~\ref{thm:limitingkernel} and \ref{thm:asymptoticsnormingconstant} follow in a standard manner.

However, to obtain \eqref{eq:AsymptFormLQLAi} quite some extra work is needed. When dealing with statistics of matrix models via OPs, one of the usual approaches is to extract the needed information via the partition function and its relation with the norming constants via a product formula, see for instance \eqref{eq:PartitionFctionNormingctt} below. Usually this is accomplished via some differential identity or with careful estimate of each term in the product formula, see for instance the works \cite{bleher_deano_partition_function, BleherIts2005, BaikBuckinghamDiFrancoIts2009,Krasovsky07} and their references for explorations along these lines. In virtue of the relation \eqref{deff:Ln} this was in fact our original attempt, but several technical issues arise. Instead, at the end we express $\msf L_n^Q$ directly as a weighted double integral of $\msf K_n^Q(x,x\mid \sad)$ in the variables in $x$ and $\sad$, this is done in Proposition~\ref{prop:intreprLn} below. The $x$-integral takes place over the whole real line, which means that when we undress $\bm R\mapsto \bm Y$ we obtain a formula for $\msf L_n^Q$ involving global and all local parametrices. The integral in $\sad$ extends to $+\infty$, which is one of the main reasons why in our main statements we also keep track of uniformity of errors when $\sad\to +\infty$. We then have to estimate the double integral, accounting for exponential decays of most of the terms but also exact cancellations of some other terms. Ultimately, the whole analysis leads to a leading contribution coming solely from a portion of the integral that arises from the model problem $\bm \Phi_n$. With a further asymptotic analysis of the later integral we obtain an integral solely of $\Phicc$ which then yields Theorem~\ref{thm:asymptoticsqLaplacetransf}.

The convergence $\bm \Phi_n\to \Phicc$ is treated as a separate issue, and to achieve it we need several information about this id-PII parametrix $\Phicc$. As a final outcome, we obtain that $\bm \Phi_n$ is close to $\Phicc$ with an error term of the form $\Boh(n^{-\nu})$, for any $\nu\in (0,2/3)$. But, in much due to the non-analyticity of the jump matrix for $\bm \Phi_n$, we are not able to achieve a sharp order $\Boh(n^{-2/3})$ unless further conditions were placed on $Q$. This non-optimal error explains the appearance of the same error order in \eqref{eq:AsymptFormLQLAi}. In the course of this asymptotic analysis we rely substantially in \cite{CafassoClaeysRuzza2021}. Among other needed info, we also borrow from the same work the connection of $\Phicc$ with the integro-differential PII. In the same work, the authors actually show that $\Phicc$ relates to particular solutions to the KdV equation that reduce to the integro-differential PII. As such, Theorems~\ref{thm:limitingkernel} and \ref{thm:asymptoticsnormingconstant} could be phrased in terms of a solution to the KdV rather than to the integro-differential PII. We opt to phrase them with the latter because this formulation encodes that all self-similarities have already been accounted for.

If we were to assume that the jump matrix for $\bm \Phi_n$ were piecewise analytic on the whole plane and not merely $C^\infty$, we could deform $\bm \Phi_n$ to a family of RHPs considered in \cite{CafassoClaeysRuzza2021}. With this in mind, the analysis of the convergence $\bm \Phi_n\to \Phicc$ is inspired by several aspects in this just mentioned work but, as we already said, here we are forced to work under different conditions on the jump matrix. In particular, one could adapt the methods in \cite{CafassoClaeysRuzza2021} to actually prove that $\bm \Phi_n$ does too relate to an $n$-dependent solution to the integro-differential PII. Consequently, with a careful inspection of our work one could show that Theorems~\ref{thm:asymptoticsqLaplacetransf}, \ref{thm:limitingkernel} and \ref{thm:asymptoticsnormingconstant} admit versions with $n$-dependent leading terms. For instance, relating the norming constant $\msfga^{(n,Q)}_{n-1}(\sad)^2$ with the model problem $\bm\Phi_n$ one could obtain an asymptotic formula of the form
$$
\msfga^{(n,Q)}_{n-1}(\sad)^2=\frac{a}{4\pi }\ee^{-2n\ell_V}\left(\frac{1}{2}-\frac{1}{n^{1/3}}\frac{\msf c_V^{1/2}}{\tad^{1/2}}
\left(\msf p_n(\sad,\tad/\msf c_V)-\frac{\sad^2\msf c_V^{3/2}}{4\tad^{3/2}}\right)+\Boh(n^{-2/3})
\right), \quad n\to \infty,
$$
where the $n$-dependent function $\msf p_n$ is obtained from $\bm \Phi_n$ and relates to a $n$-dependent solution $\phiad_n$ to the integro-differential PII. In fact, with standard arguments one could improve the formula above to a full asymptotic expansion in powers of $n^{-1/3}$, with bounded but $n$-dependent coefficients. Underlying our arguments there is the statement that $\msf p_n=\msf p+\Boh(n^{-\nu})$ for any $\nu\in (0,2/3)$, which then yields Theorem~\ref{thm:asymptoticsnormingconstant}. But as a drawback, although one could potentially improve \eqref{eq:asymptformNormingctt} and also obtain the term of order $n^{-2/3}$ explicitly, it is not possible to obtain the $\Boh(n^{-1})$ term in \eqref{eq:asymptformNormingctt} unless one improves the error $\Boh(n^{-\nu})$ in the convergence $\bm \Phi_n\to \Phicc$ to a sharp error $\Boh(n^{-2/3})$. 

\subsection{Possible extensions}\hfill

Most of our approach may be extended to potentials $V$ for which the equilibrium measure $\mu_V$ is critical, and also under different conditions on $Q$ as we now explain. 

Apart from technical adaptations in several steps of the RHP for OPs which are nowadays well understood, our analysis carries over to potentials $V$ for which the equilibrium measure $\mu_V$ is regular but multicut, with the same conditions on $Q$ when $\mu_V$ has the origin as its right-most endpoint. 

When, say, the density $\mu_V$ vanishes to a higher power at a soft edge and/or $Q$ changes sign with an arbitrary odd vanishing order at the same soft edge, we need to replace the power $n^{2/3}$ in $\sigma_n$ by another appropriate power to modify the local statistics near this point in a non-trivial critical manner. Once this is done, the asymptotic analysis of the RHP for OPs that we perform carries over mostly with minor modifications, and the only major issue to overcome is in the construction of a new local parametrix $\wt{\bm \Phi}_n$ near this soft edge point and its corresponding asymptotic analysis. In this case, we expect that $\wt{\bm \Phi}_n\to \wt{\bm \Phi}_0$ for a new function $\wt{\bm \Phi}_0$. It is relatively simple to write a RHP that should be satisfied by this $\wt{\bm \Phi}_0$, and we expect it to be related to the KdV hierarchy \cite{Claeys2012} but with nonstandard initial data. It would be interesting to see if the particular solutions obtained this way reduce to integro-differential hierarchies of Painlevé equations, in the same spirit of the recent works \cite{BothnerCafassoTarricone2021,Krajenbrink2020}.

One could also consider similar statistics to \eqref{deff:Ln} with a $Q$ that vanishes at a bulk point of $\supp\mu_V$. We do expect that most of our work carries through to this situation, at least when we impose $V$ to be again one-cut regular and $Q$ to vanish quadratically at a point inside $\supp\mu_V$. The main issue that should arise is again on the construction of the local parametrix near this point, and its corresponding asymptotic analysis. This model should lead to multiplicative statistics of the Sine kernel (and the higher order generalizations of it). Similar considerations go through to hard-edge models, leading to multiplicative statistics of the Bessel process. To our knowledge, such multiplicative statistics of Bessel and Sine have not been considered in the literature so far. However, finite temperature versions of the Sine and Bessel kernels do have appeared, see for instance \cite{Johansson2007,BeteaBouttier2019,BeteaOccelli2021,CundenMezzadriOConnell2018}.

\subsection{Organization of the paper}\hfill

The paper is organized in two parts. In the first part, we deal with a family of RHPs $\bm \Phi_\tauad$ that contains the model RHP $\bm \Phi_n$ needed in the asymptotic analysis of OPs. In Section \ref{sec:modelproblem} we introduce $\bm\Phi_\tauad$ formally. In Sections~\ref{sec:KPZRHP} and \ref{sec:boundsKPZRHP} we discuss the RHP $\Phicc$, which is a particular case of $\bm\Phi_\tauad$, and review several of its properties, translating results from \cite{CafassoClaeys2021,CafassoClaeysRuzza2021} to our notation and needs. In Section~\ref{sec:asympanalymodelprobl} we prove the convergence $\bm \Phi_\tauad\to \Phicc$ and of related quantities in the appropriate sense. The latter section contains all the needed results for the asymptotic analysis of the RHP for OPs, and concludes the first part of this paper.

The second part of the paper is focused on the asymptotic analysis of the RHP for OPs. In Section~\ref{sec:eqmeasuretau} we discuss several aspects that relate to the equilibrium measure. In Section~\ref{sec:rhpapproachops} we introduce the Christoffel-Darboux kernel $\msf K_n^Q$ and related quantities, and display how they relate to the RHP for OPs. In particular, in Proposition~\ref{prop:intreprLn} we write $\msf L_n^Q$ directly to as an integral of the kernel $\msf K_n^Q$, a result which may be of independent interest. In Section~\ref{sec:rhpanalysis} we perform the asymptotic analysis of the RHP for the OPs. In Section~\ref{sec:ProofMainResults} use the conclusions from Sections~\ref{sec:rhpanalysis} and \ref{sec:asympanalymodelprobl} and prove Theorems~\ref{thm:limitingkernel} and \ref{thm:asymptoticsnormingconstant}. Also from the results from Sections~\ref{sec:rhpanalysis} and \ref{sec:asympanalymodelprobl} and assuming additional technical estimates, the proof of Theorem~\ref{thm:asymptoticsqLaplacetransf} is given in Section~\ref{sec:ProofMainResults}. Such remaining technical estimates are also ultimately a consequence of the RHP analysis, but their proofs are rather cumbersome and postponed to Section~\ref{sec:technicalestimates}.

For the remainder of the paper it is convenient to denote 
\begin{equation}\label{deff:matrixnot1}
\bm e_1\deff\begin{pmatrix}
1 \\ 0
\end{pmatrix},\quad
\bm e_2\deff\begin{pmatrix}
0 \\ 1
\end{pmatrix},\quad 
\bm E_{jk}\deff\bm e_j\bm e_k^\tp,
\end{equation}
so $\bm E_{jk}$ is a $2\times 2$ matrix with the $(j,k)$-entry equals $1$ and all other entries zero. With this notation, the Pauli matrices, for instance, take the form
\begin{equation}\label{deff:matrixnot2}
\bm I\deff\bm E_{11}+\bm E_{22},\quad
\sp_1\deff \bm E_{12}+\bm E_{21}, \quad
\sp_2\deff -\ii \bm E_{12}+\ii \bm E_{21},\quad 
\sp_3\deff\bm E_{11}-\bm E_{22}.
\end{equation}
In particular, for any reasonably regular scalar function $f$, the spectral calculus yields
\begin{equation}\label{deff:matrixnot3}
f(z)^{\sp_3}=
\begin{pmatrix}
f(z) & 0 \\ 0 & 1/f(z)
\end{pmatrix}.
\end{equation}
These notations will be used extensively in the coming sections.

\subsection*{Acknowledgments}\hfill

P. G. wishes to thank Ivan Corwin and Alexandre Krajenbrink for many helpful conversations and Alexei Borodin for comments on the earlier version of this manuscript. G. S. is grateful to Jinho Baik, Tom Claeys, Mattia Cafasso, Lun Zhang and Alfredo Deaño for inspiring conversations, and Dan Betea for pointing us out to  relevant references. He also acknowledges his current support by São Paulo Research Foundation under grants \# 2019/16062-1 and \# 2020/02506-2, and by Brazilian National Council for Scientific and Technological Development (CNPq) under grant \# 315256/2020-6.
This work was partially developed while the authors participated in the program {\it Universality and Integrability in Random Matrix Theory and Interacting Particle Systems}, hosted
by the Mathematical Sciences Research Institute in Berkeley, California, during the Fall 2021 semester, supported by the National Science Foundation
under Grant No. DMS-1928930. We are grateful to the organizers for their effort in providing an excellent research atmosphere despite the uncertain times.

\section{A Model Riemann-Hilbert Problem}\label{sec:modelproblem}

In this section we discuss a model Riemann-Hilbert Problem that will be used in the construction of a local parametrix in the asymptotic analysis for the orthogonal polynomials. As such, this model problem plays a central role in obtaining all our major results.

\subsection{The model problem}\hfill

Set
\begin{equation}\label{def:contour_sigma}
\cSig_0\deff [0,+\infty), \; \cSig_1\deff [0,\ee^{2\pi \ii/3}),  \; \cSig_2\deff (-\infty,0], \; \cSig_3\deff [0,\ee^{-2\pi \ii/3}), \quad   \cSig\deff\bigcup_{j=0}^3 \cSig_j,
\end{equation}
orienting $\cSig_0$ from the origin to $\infty$, and the remaining arcs from $\infty$ to the origin, see Figure~\ref{fig:sigmaj}.

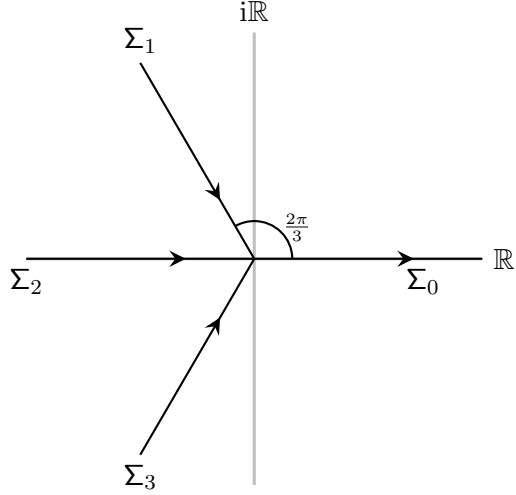
\begin{figure}[t]
\centering
		\begin{tikzpicture}[scale=1]
%
\draw [line width=0.4mm,lightgray] (-3,0)--(3,0) node [pos=1,right,black] {$\R$};
\draw [line width=0.4mm,lightgray] (0,-3)--(0,3) node [pos=1,above,black] {$\ii\R$};
\draw [thick,postaction={mid arrow={black,scale=1.5}}] (0,0) to (0:3) node [yshift=0pt,xshift=-12pt] (a) {};
\draw [thick,postaction={rmid arrow={black,scale=1.5}}] (0,0) to (120:3) node (b) {};	
\draw [thick,postaction={rmid arrow={black,scale=1.5}}] (0,0) to (180:3) node  (c) {};
\draw [thick,postaction={rmid arrow={black,scale=1.5}}] (0,0) to (-120:3) node (d) {};
\node at (a) [below left] {$\cSig_0$};
\node at (b) [above] {$\cSig_1$};
\node at (c) [below] {$\cSig_2$};
\node at (d) [below] {$\cSig_3$};
\node (0:60:0.5) [yshift=12pt,xshift=16pt] {$\scriptstyle \frac{2\pi}{3}$};
\draw[thick] (0.5,0) arc (0:120:0.5);
%
\end{tikzpicture}
\caption{The contours $\cSig_0,\cSig_1,\cSig_2$ and $\cSig_3$ in \eqref{def:contour_sigma} that constitute $\cSig$.}\label{fig:sigmaj}
\end{figure}

The model RHP we are about to introduce depends on a function $\msf h:\cSig\to \C$ used to describe its jump. For the moment we assume
$$\msf h\in C^\infty(\cSig), \quad \msf h(z)\in \R \text{ for }z\in \R,\quad \text{and}\quad \liminf_{\substack{z\to \infty \\ z\in \cSig}} \frac{\re \msf h(z)}{|z|}>0.
$$
These conditions are present only to ensure the RHP below is well posed and are far from optimal, but enough for our purposes. Later on we will impose more conditions on this function $\msf h$, these conditions will be tailored to our later needs regarding the asymptotic analysis of OPs.

The associated RHP asks for finding a $2\times 2$ matrix-valued function $\bm \Phi$ with the following properties.

\begin{enumerate}[\bf $\bm \Phi$-1.]
\item The matrix $\bm \Phi=\bm \Phi(\cdot \mid \msf h):\C\setminus \cSig\to \C^{2\times 2}$ is analytic.
\item Along the interior of the arcs of $\cSig$ the function $\bm \Phi$ admits continuous boundary values $\bm \Phi_\pm$
 related by $\bm \Phi_+(\zeta)=\bm \Phi_-(\zeta)\bm J_{\bm \Phi}(\zeta)$, $\zeta\in \cSig$, with
\begin{equation}\label{eq:jumpPhimodel}
\bm J_{\bm \Phi}(\zeta)\deff
\begin{dcases}
\bm I+\frac{1}{1+\ee^{-\msf h(\zeta)}} \bm E_{12}, & \zeta \in  \cSig_0 , \\
\bm I+(1+\ee^{-\msf h(\zeta)})\bm E_{21}, & \zeta \in \cSig_1\cup\cSig_3, \\
\frac{1}{1+\ee^{-\msf h(\zeta)}}\bm E_{12}-(1+\ee^{-\msf h(\zeta)})\bm E_{21}, & \zeta\in\cSig_2.
\end{dcases}
\end{equation}
\item As $\zeta\to\infty$,
\begin{equation}\label{eq:ModelRHPAsymp}
\bm \Phi(\zeta)=\left(\bm I+\frac{1}{\zeta}\bm \Phi^{(1)}+\Boh(1/\zeta^{2})\right)\zeta^{\sp_3/4}\bm U_0^{-1}
\ee^{-\frac{2}{3}\zeta^{3/2}\sp_3},
\end{equation}
where 
\begin{equation}\label{def:matrixU0modelprobl}
\bm U_0\deff \frac{1}{\sqrt{2}}
\begin{pmatrix}
1 & \ii \\ \ii  & 1
\end{pmatrix}
\end{equation}
and $\bm \Phi^{(1)}=\bm \Phi^{(1)}(\msf h)$ is a matrix that depends on the choice of function $\msf h$ but it is independent of $\zeta$.

\item The matrix $\bm \Phi$ remains bounded as $\zeta \to 0$.
\end{enumerate}

Given $\msf h$, it is not at all obvious that the RHP above has a solution and how to describe it. We study this model problem when $\msf h=\msf h_\tauad$ depends on an additional large parameter $\tauad$, in a way that appears naturally in the asymptotic analysis of the orthogonal polynomials mentioned earlier. For large values of $\tauad$, we then prove that the solution $\bm \Phi$ exists and is asymptotically close to a model RHP that appeared recently \cite{CafassoClaeys2021} and that we discuss in a moment.

\subsection{The model RHP with admissible data}\hfill 

For us, we need to consider the model problem $\bm \Phi=\bm \Phi(\cdot\mid \msf h)$ with functions $\msf h=\msf h_\tauad$ satisfying certain properties which are formally introduced in the next definition. 

\begin{definition}\label{deff:admissibleh}
We call a function $\msf h_\tauad:\Sigma\to \C$ {\it admissible} if it is of the form
$$
\msf h_\tauad(\zeta)=\msf h_\tauad(\zeta\mid \sad)=\sad+\tauad\msf H\left(\frac{\msf \zeta}{\tauad}\right),\quad \zeta\in \C,\quad \tauad>0,\;  \sad\in \R,
$$
where $\msf H$ is defined on a neighborhood $\mcal S$ of $\Sigma$ and satisfies the following properties.
\begin{enumerate}[(i)]
\item The function $\msf H$ is independent of $\tauad$ and $\sad$, of class $C^\infty$ on $\mcal S$ and real-valued along $\R$.
\item $\msf H$ is analytic on a disk $D_\delta(0)\subset\mcal S$ centered at the origin, and its unique zero on $D_\delta(0)$ is at $\zeta=0$, with 
$$
\tad\deff -\msf H'(0)>0.
$$
\item There exist constants $\eta,\wh\eta>0$ for which
$$
\re \msf H(w)>\eta|w| \quad \text{for } w\in \cSig_1\cup\cSig_2\cup\cSig_3,
$$
and
$$
-\wh\eta w^{3/2-\epsilon}<\msf H(w)<-\eta w  \quad \text{for } w\in \cSig_0,
$$
for some $\epsilon\in (0,1/2]$.
\end{enumerate}
\end{definition}

Conditions (i)--(ii), and also the bounds in (iii) involving $\eta$, are natural in our setup. The bound $\msf H(w)>-\wh\eta w^{3/2-\epsilon}$ is present for technical reasons, and it plays a role only for the proof of Lemma~\ref{lem:estJpsiadcSig0}, allowing us to write certain estimates in a cleaner matter. It could be removed, at the cost of slightly more complicated error terms in the mentioned Lemma. For our purposes, namely to use $\bm \Phi=\bm\Phi(\cdot\mid \msf h_\tauad)$ as a local parametrix with an appropriate $\msf h_{\tauad}$, this condition is satisfied anyway (this will be accomplished in Proposition~\ref{prop:ConformExtH}), so we include it in our definition here as well, as it simplifies our analysis.

In the course of the analysis for the RHP for the orthogonal polynomials discussed in Section~\ref{sec:rhpapproachops}, the function $\msf H$ will be a transformation of the function $Q$ appearing in the deformed weight \eqref{def:perturbedweight}, and the parameter $\tad$ that we defined here will play the same role as the one in the definition \eqref{deff:tQprime}.

Given an admissible $\msf h_\tauad$, we denote
\begin{equation}\label{deff:Phit}
\bm \Phi_\tauad(\zeta)\deff \bm \Phi(\zeta\mid \msf h=\msf h_\tauad(\cdot\mid \sad)).
\end{equation}
We are interested in the asymptotic analysis for $\bm \Phi_\tauad$ as $\tau\to +\infty$ and $\sad\geq -\sad_0$, for any $\sad_0>0$, and $\tad>0$ kept fixed within a compact of the positive axis.

We now explain in an ad hoc manner the appearance of a RHP for the integro-differential equation, which also relates to the KPZ equation.  Definition~\ref{deff:admissibleh}--(ii) gives that $\msf H$ has an expansion of the form
$$
\msf H(\zeta)=-\tad\zeta(1+\Boh(\zeta)),\quad |\zeta|\leq \delta,
$$
This means that any admissible function $\msf h_\tauad$ satisfies
$$
\msf h_\tauad(\zeta)=\sad-\tad\zeta\left(1+\Boh(\zeta \tauad^{-1})\right),\quad |\zeta|\leq \delta \tauad.
$$
In particular, the convergence
\begin{equation}\label{deff:h0}
\msf h_\tauad(\zeta)\to \msf h_0(\zeta)=\msf h_0(\zeta\mid \sad,\tad)\deff \sad-\tad\zeta,
\end{equation}
holds true uniformly in compacts as $\tauad\to\infty$. This indicates that the solution $\bm \Phi_\tauad$ should converge to the solution 
\begin{equation}\label{deff:Phicc}
\Phicc\deff \bm\Phi(\cdot\mid \msf h=\msf h_0)
\end{equation}
of the model problem obtained from $\msf h_0$. The RHP-{$\Phicc$} relates to the integro-differential PII and is a rescaled version of an RHP that appears in the description of the narrow wedge solution to the KPZ equation, as we discuss in the next section in detail.

\section{The RHP for the integro-differential RHP}\label{sec:KPZRHP}

For the choice
\begin{equation}\label{deff:hkpz}
\hkpz(\zeta)=\hkpz(\zeta\mid \skpz,\Tkpz)\deff -\Tkpz^{1/3}(\skpz+\zeta)
\end{equation}
the corresponding solution of the RHP--{$\bm\Phi$}
$$
\Phikpz(\zeta)=\Phikpz(\zeta\mid \skpz,\Tkpz)\deff \bm\Phi(\zeta\mid \msf h=\hkpz(\cdot \mid \skpz,\Tkpz))
$$
appeared for the first time in the work of Cafasso and Clayes \cite{CafassoClaeys2021} (this is the RHP-$\Psi$ in Section~2 therein) in connection with the narrow wedge solution to the KPZ equation as we explain in a moment, in Section~\ref{sec:KPZRHPKPZ}. To avoid confusion with the related quantities that we are about to introduce, we term it the KPZ RHP. In virtue of the identity
$$
\msf h_0(\zeta\mid \sad,\tad)=\hkpz(\zeta\mid \skpz=-\sad/\tad, \Tkpz=\tad^3)
$$
which follows from \eqref{deff:h0} and \eqref{deff:hkpz}, we also have the correspondence
\begin{equation}\label{eq:PhiccPhikpz}
\Phicc(\zeta\mid \sad,\tad)=\Phikpz(\zeta\mid s=-\sad/\tad,T=\tad^3),
\end{equation}
and we refer to $\Phicc$ as the id-PII RHP. For the record, we state the existence of $\Phicc$ formally as a result.

\begin{prop}\label{prop:existencePhicc}
For any $\sad\in \R$ and any $\tad>0$, the solution $\Phicc$ exists and is unique. Furthermore, for any fixed $\sad_0>0$ and $\tad_0\in (0,1)$, the solution $\Phiccp(\zeta)$ remains bounded for $\zeta$ in compacts of $\R$ and $\sad\geq -\sad_0$, $\tad_0\leq \tad\leq 1/\tad_0$.
\end{prop}
\begin{proof}
It is a consequence of \cite[Section~2]{CafassoClaeys2021} that the solution $\Phikpz(\cdot\mid \skpz,\Tkpz)$ exists and is unique, for any $\skpz\in \R$ and $\Tkpz>0$, and from the correspondence \eqref{eq:PhiccPhikpz} the existence and uniqueness of $\Phicc$ is thus granted.

For the boundedness, we start from the representation
$$
\Phicc(\zeta)=\bm I+\frac{1}{2\pi \ii}\int_{\Gamma}\frac{\Phiccm(s)(\bm J_{\Phicc}(s)-\bm I)}{s-\zeta}\dd s,\quad \zeta\in \C\setminus \Gamma,
$$
which follows from the $L^p$ theory of RHPs (see \cite{DeiftParkCity}). The jump matrix admits an analytic continuation to any neighborhood of the real axis, and this analytic continuation remains bounded in compacts, also uniformly for $\sad\geq -\sad_0$ and $\tad_0\leq \tad\leq 1/\tad_0$ (see for instance \eqref{eq:jumpPhicc} for the exact expression). With these observations in mind, the claimed boundedness follows from standard arguments. We skip additional details, but refer to the proof of Theorem~\ref{thm:PhitauAsympt}, in particular \eqref{eq:maxprincipleconseqPsitauad} {\it et seq.}, for similar arguments in a more involved context.
\end{proof}

In this section we collect several results on $\Phicc$ that were obtained in \cite{CafassoClaeys2021,CafassoClaeysRuzza2021} and which will be needed later. 

But before proceeding, a word of caution. As we said, the RHP--{$\Phikpz$} appeared first in \cite{CafassoClaeys2021}, but was also studied in the subsequent work \cite{CafassoClaeysRuzza2021}. The meanings for the variables $s,x$ and $t$ in these two works are not consistent, but we need results from both of them. Comparing to the work \cite{CafassoClaeys2021} by Cafasso and Claeys, the correspondence is
\begin{equation}\label{notation:cc}
s_{\rm CC}=-\frac{\sad}{\tad}\quad \text{and}\quad T_{\rm CC}=\tad^3.
\end{equation}
This correspondence is consistent with \eqref{eq:PhiccPhikpz}. On the other hand, when comparing to the subsequent work \cite{CafassoClaeysRuzza2021} by Cafasso, Claeys and Ruzza, the correspondence between notations is
\begin{equation}\label{notation:ccr}
t_{\rm CCR}=\frac{1}{\tad^{3/2}}, \quad x_{\rm CCR}=-\frac{\sad}{\tad^{3/2}},\qquad \text{that is}\qquad x_{\rm CCR}=-\Sad,\quad t_{\rm CCR}=\Tad,
\end{equation}
where $\Tad,\Sad$ are as in \eqref{eq:corr_tad_Tad}.

In our asymptotic analysis, the most convenient choice of variables to work with is the choice $(\sad,\tad)$ and the correspondence $(\Sad,\Tad)$ from \eqref{eq:corr_tad_Tad} that we have already been using, and which leads to the RHP $\Phicc$ as we introduced. Nevertheless, we will need to collect results from both mentioned works, and when the need arises we refer to the correspondences of variables \eqref{notation:cc}--\eqref{notation:ccr}. 

On the other hand, when making correspondence with integrable systems, in particular the integro-differential Painlevé II equation, it is more convenient to work with the variables $\Sad$ and $\Tad$ as in \eqref{eq:corr_tad_Tad}.

\subsection{Properties of the id-PII parametrix}\label{sec:KPZRHPKPZ}\hfill 

In this section we describe many of the findings from \cite{CafassoClaeys2021,CafassoClaeysRuzza2021}, in a way suitably adapted to our notation and needs. In particular the connection of $\Phicc$ introduced in \eqref{deff:Phicc} with the integro-differential Painlevé II equation is described in this section.

For
\begin{equation}\label{deff:Dkpz}
\Dkpz(\zeta)=\Dkpz(\zeta\mid \skpz,\Tkpz)\deff 
\begin{dcases}
\bm I, & \zeta>0, \\
\bm I+(1+\ee^{\Tkpz^{1/3}(\zeta+\skpz)})\bm E_{21}, & \zeta<0.
\end{dcases}
\end{equation}
the identity 
\begin{equation}\label{eq:QccIntRepPhikpz}
\partial_\skpz \log \Qcc(\skpz,\Tkpz)=\frac{\Tkpz^{1/3}}{2\pi \ii}\int_{-\infty}^\infty \frac{\ee^{\Tkpz^{1/3}(x+\skpz)}}{(1+\ee^{\Tkpz^{1/3}(x+\skpz)})^2}\left[(\Dkpz(x)^{-1}\Phikpz_+(x)^{-1}(\Phikpz_+ \Dkpz )'(x)\right]_{21} \dd x,
\end{equation}
was shown in \cite[Theorem~2.1]{CafassoClaeys2021} and will also be useful for us. With \eqref{eq:PhiccPhikpz} we now rewrite this identity in terms of $\Phicc$.
With the principal branch of the argument, set
\begin{equation}\label{deff:Dcc}
\Dcc(\zeta)=\Dcc(\zeta\mid \sad, \tad)\deff
\begin{cases}
\bm I, & |\arg\zeta|<\frac{2\pi}{3}, \\
\bm I+(1+\ee^{-\sad+\tad\zeta})\bm E_{21}, & |\arg\zeta|>\frac{2\pi}{3}.
\end{cases}
\end{equation}
This function relates to $\Dkpz$ in \eqref{deff:Dkpz} via
$$
\Dccp(\zeta)=\Dkpz(\zeta\mid \skpz=-\sad/\tad, \Tkpz=\tad^3),\quad \zeta\in \R,
$$
and \eqref{eq:QccIntRepPhikpz} rewrites as
\begin{equation}\label{eq:QccIntRepPhicc}
\partial_\skpz \log \Qcc(\skpz=-\sad/\tad,\Tkpz=\tad^3)=\frac{\tad}{2\pi \ii}\int_{-\infty}^\infty \frac{\ee^{\tad x-\sad}}{(1+\ee^{\tad x-\sad})^2}\left[(\Dcc(x)^{-1}\Phiccp(x)^{-1}(\Phiccp \Dcc )'(x)\right]_{21} \dd x.
\end{equation}

For further reference, it is now convenient to state the RHP for $\Phicc$ explicitly.
\begin{enumerate}[\bf $\Phicc$-1.]
\item The matrix $\Phicc:\C\setminus \cSig\to \C^{2\times 2}$ is analytic.
\item Along the interior of the arcs of $\cSig$ the function $\Phicc$ admits continuous boundary values $\Phiccpm$
 related by $\Phiccp(\zeta)=\Phiccm(\zeta)\bm J_{\Phicc}(\zeta)$, $\zeta\in \cSig$, with
\begin{equation}\label{eq:jumpPhicc}
\bm J_{\Phicc}(\zeta)\deff
\begin{dcases}
\bm I+\frac{1}{1+\ee^{-\sad+\tad \zeta}} \bm E_{12}, & \zeta \in  \cSig_0 , \\
\bm I+(1+\ee^{-\sad+\tad\zeta})\bm E_{21}, & \zeta \in \cSig_1\cup\cSig_3, \\
\frac{1}{1+\ee^{-\sad+\tad\zeta}}\bm E_{12}-(1+\ee^{-\sad+\tad\zeta})\bm E_{21}, & \zeta\in\cSig_2.
\end{dcases}
\end{equation}
\item As $\zeta\to\infty$,
\begin{equation}\label{eq:AsymptPhicc}
\Phicc(\zeta)=\left(\bm I+\Boh(1/\zeta)\right)\zeta^{\sp_3/4}\bm U_0^{-1}
\ee^{-\frac{2}{3}\zeta^{3/2}\sigma_3}.
\end{equation}
where we recall that $\bm U_0$ is given in \eqref{def:matrixU0modelprobl}.

\item The matrix $\Phicc$ remains bounded as $\zeta \to 0$.
\end{enumerate}

To compare with \cite{CafassoClaeysRuzza2021} we perform a transformation of this RHP. All the calculations that follow already take into account the correspondence \eqref{notation:ccr} between the notation in the mentioned work and our notation. 

With $\Dcc$ as in \eqref{deff:Dcc} and introducing
$$
\xi=\xi(\zeta)=-\sad+\tad\zeta,\quad \text{with inverse}\quad \zeta=\zeta(\xi)=\frac{\xi+\sad}{\tad},
$$
we transform
\begin{equation}\label{eq:transfPsiccPhicc}
\Psicc(\xi)=\left(\bm I+\frac{\ii \sad^2}{4\tad^{3/2}}\bm E_{12}\right)\tad^{\sp_3/4}\Phicc(\zeta(\xi))\times
\begin{cases}
\Dcc(\zeta(\xi)), & \im \xi>0, \; \arg(\zeta(\xi))\neq 2\pi/3,\\
\Dcc(\zeta(\xi))^{-1}, & \im \xi<0, \; \arg(\zeta(\xi))\neq -2\pi/3.
\end{cases}
\end{equation}
Then $\Psicc$ satisfies the following RHP.

\begin{enumerate}[\bf $\Psicc$-1.]
\item The matrix $\Psicc:\C\setminus \R\to \C^{2\times 2}$ is analytic.
\item Along $\R$ the function $\Psicc$ admits continuous boundary values $\Psiccpm$
 related by 
 $$
 \Psiccp(\xi)=\Psiccm(\xi)
\left(\bm I+\frac{1}{1+\ee^{\xi}} \bm E_{12}\right), \quad \xi\in \R.
$$
%
\item For any $\delta\in (0,2\pi/3)$, as $\xi\to\infty$ the matrix $\Psicc$ has the following asymptotic behavior,
\begin{equation}\label{eq:RHPAsympPsicc}
\Psicc(\xi)=\left(\bm I+\Boh(1/\xi)\right)\xi^{\sp_3/4}\bm U_0^{-1}
\ee^{-\tad^{-3/2}\left(\frac{2}{3}\xi^{3/2}+\sad \xi^{1/2}\right)\sigma_3}
\times
\begin{cases}
\bm I, & |\arg \xi|\leq \pi-\delta, \\
\bm I\pm \bm E_{21}, & \pi-\delta<\pm \arg \xi <\pi.
\end{cases}
\end{equation}
\end{enumerate}

This RHP is the same RHP considered in \cite[page~1120]{CafassoClaeysRuzza2021}\footnote{The keen reader will notice that there is a sign difference between the last term in the right-hand side of \eqref{eq:ModelRHPAsymp} and the corresponding term in \cite[page~1120]{CafassoClaeysRuzza2021}, but the latter is a typo.} with the choice $\sigma(r)=(1+\ee^{-r})^{-1}$ therein and the correspondence of variables \eqref{notation:ccr}.

As a consequence, and with the change of variables $(\sad,\tad)\mapsto (\Sad,\Tad)$ from \eqref{eq:corr_tad_Tad}, we obtain that for some functions $\msf Q=\msf Q(\Sad,\Tad),\msf R=\msf R(\Sad,\Tad), \msf P=\msf P(\Sad,\Tad)$ and
\begin{equation}\label{deff:coeffqrp}
\msf q=\msf q(\sad,\tad)=\msf Q(\Sad,\Tad),\quad \msf r=\msf r(\sad,\tad)=\msf R(\Sad,\Tad),\quad \msf p=\msf p(\sad,\tad)=\msf P(\Sad,\Tad),
\end{equation}
the asymptotic behavior \eqref{eq:RHPAsympPsicc} improves to
\begin{multline}\label{eq:AsympPsiccImprov}
\Psicc(\xi)=\left(\bm I+
\frac{1}{\xi}
\begin{pmatrix}
\msf q & -\ii \msf r \\ \ii \msf p & -\msf q
\end{pmatrix}+
\Boh(\xi^{-2})\right)\xi^{\sp_3/4}\bm U_0^{-1}\\
\times 
\ee^{-\tad^{-3/2}\left(\frac{2}{3}\xi^{3/2}+\sad \xi^{1/2}\right)\sigma_3}
\times
\begin{cases}
\bm I, & |\arg \xi|\leq \pi-\delta, \\
\bm I\pm \bm E_{21}, & \pi-\delta<\pm \arg \xi <\pi,
\end{cases}
\quad \xi\to\infty.
\end{multline}
Stressing that the correspondence \eqref{notation:ccr} is in place, the functions $\msf P$ and $\msf Q$ satisfy the relation \cite[Equation~(3.14)]{CafassoClaeysRuzza2021}
$$
\partial_\Sad \msf P(\Sad,\Tad)=2\msf Q(\Sad,\Tad)+\msf P(\Sad,\Tad)^2.
$$
Furthermore, from \cite[Equations~(3.12),(3.16), Theorem~1.3 and Corollary~1.4]{CafassoClaeysRuzza2021} we see that $\Psicc$ takes the form
\begin{equation}\label{eq:transfPsiccphiad}
\Psicc(\xi\mid \sad,\tad)=
\sqrt{2\pi}\ee^{-\frac{\pi\ii}{4}\sp_3}\left(\bm I- \msf p(\sad,\tad)\bm E_{12}\right)
\begin{pmatrix}
-\partial_\Sad \phiad (\xi\mid \Sad,\Tad) & \ast \\
-\phiad(\xi\mid \Sad,\Tad) & \ast
\end{pmatrix}
\ee^{\frac{\pi\ii}{4}\sp_3},
\end{equation}
where $\phiad=\phiad(\xi\mid \Sad,\Tad)$ solves the NLS equation with potential $2\partial_\Sad \msf P$,
$$
\partial_\Sad^2\phiad(\xi\mid \Sad,\Tad)=(\xi+2\partial_\Sad \msf P(\Sad,\Tad))\phiad(\xi\mid \Sad,\Tad).
$$
In addition, $\msf P$ and $\phiad$ are related through the identity \eqref{eq:intdiffPIIPhip} which, in turn, implies that $\phiad$ is the solution to the integro-differential Painlevé II equation in \eqref{eq:intdiffPII}.

It is convenient to write some quantities of $\Phicc$ directly in terms of the just introduced functions.
%
The first identity we need for later is
\begin{multline}\label{eq:transfPhicckernelphicc}
\left[\left(\Phicc(\zeta(v)\mid \sad,\tad)\Dcc(\zeta(v)\mid \sad,\tad)\right)^{-1}\Phicc(\zeta(u)\mid \sad,\tad)\Dcc(\zeta(u)\mid \sad,\tad)\right]_{21,+} \\ 
=-2\pi\ii \left(\phiad(u\mid \Sad,\Tad)(\partial_{\Sad}\phiad)(v\mid \Sad,\Tad)-\phiad(v\mid \Sad,\Tad)(\partial_{\Sad}\phiad)(u\mid \Sad,\Tad)\right)
\end{multline}
which follows from \eqref{eq:transfPsiccPhicc} and \eqref{eq:transfPsiccphiad} after a straightforward calculation, accounting also that $\det\Phicc=\det \Psicc\equiv 1$.

The second relation we need is an improvement of the asymptotics of $\Phicc$ in \eqref{eq:AsymptPhicc}. With the coefficients
$$
\msf c_1=\msf c_1(\sad,\tad)\deff -\frac{\sad^2}{4\tad^{3/2}},\quad \msf c_2=\msf c_2(\sad,\tad)\deff \frac{\sad^4}{32\tad^3},\quad
\msf c_3=\msf c_3(\sad,\tad)\deff -\frac{\sad^3(\sad^3-16\tad^3)}{384\tad^{9/2}},
$$
and the functions $\msf q,\msf r,\msf p$ in \eqref{deff:coeffqrp}, introduce
%
%
\begin{equation}\label{deff:Phikpz_residue}
\Phicc^{(1)}\deff \frac{1}{\tad} 
\begin{pmatrix}
-\dfrac{\sad}{4}+\msf q+\msf c_2-\msf c_1\msf p-\msf c_1^2 & \ii \tad^{-1/2}\left(-\msf r-2\msf q\msf c_1+\dfrac{\sad}{2}\msf c_1+\msf p\msf c_1^2+\msf c_1\msf c_2-\msf c_3\right) \\
\ii \tad^{1/2}(\msf p+\msf c_1) & \dfrac{\sad}{4}-\msf q+\msf p\msf c_1+\msf c_2
\end{pmatrix}.
\end{equation}
After some cumbersome but straightforward calculations, the asymptotics \eqref{eq:AsympPsiccImprov} improves \eqref{eq:AsymptPhicc} to
$$
\Phicc(\zeta)=\left(\bm I+\frac{\Phicc^{(1)}}{\zeta}+\Boh(\zeta^{-2})\right)\zeta^{\sp_3/4}\bm U_0^{-1}\ee^{-\frac{2}{3}\zeta^{3/2}\sp_3},\quad \zeta\to \infty.
$$

\section{Bounds on the id-PII RHP}\label{sec:boundsKPZRHP}

We need to obtain certain asymptotic bounds on $\Phicc$ in different regimes. These bounds will be used later to show that the model problem $\bm\Phi_\tauad$ converges, as $\tauad\to+\infty$, to $\Phicc$ as already indicated in \eqref{deff:Phit} {\it et seq.} We split these necessary estimates in the next subsections, depending on the regime we are.  

In what follows, for a matrix-valued function $\bm M=(\bm M_{jk})$ and a contour $\Sigma\subset \C$, we also use the pointwise matrix norm
\begin{equation}\label{deff:matrixnorm}
|\bm M(\zeta)|\deff \max_{j,k} |\bm M_{j,k}(\zeta)|,
\end{equation}
and the matrix $L^p$ norm (possibly also with $p=\infty$)
\begin{equation}\label{deff:matrixLpnorm}
\|\bm M\|_{L^p(\Sigma)}\deff \max_{j,k} \| \bm M_{j,k}\|_{L^p(\Sigma)},
\end{equation}
where the measure is always understood to be the arc-length measure. In particular, for any two given matrices $\bm M_1$ and $\bm M_2$ the inequality
$$
\|\bm M_1\bm M_2\|_{L^\infty(\Sigma)}\leq 2 \|\bm M_1\|_{L^\infty(\Sigma)}\|\bm M_2\|_{L^\infty(\Sigma)}
$$ 
is satisfied. Similar straightforward inequalities involving $L^1,L^2$ and $L^\infty$ and the pointwise norm \eqref{deff:matrixnorm} also hold, and will be used without further mention. Sometimes we also write
\begin{equation}\label{deff:matrixLpqnorm}
\|\bm M\|_{L^p\cap L^q(\Sigma)}\deff \max \left\{\| \bm M\|_{L^p(\Sigma)},  \| \bm M\|_{L^q(\Sigma)} \right\},
\end{equation}
to identify that possible convergences are taking place in various norms simultaneously. In a similar manner we define norms $\|\cdot\|_{L^{p_1}\cap L^{p_2}\cap L^{p_3}}$ involving three (or possibly more) function spaces.

\subsection{The singular regime}\label{sec:Phiccsingregime}\hfill 

The first asymptotic regime we consider is 
$$
\sad\geq \sad_0 \quad \text{and}\quad \tad_0\leq \tad \leq \frac{1}{\tad_0},
$$
where $\tad_0\in (0,1)$ is any given value, and $\sad_0=\sad_0(\tad_0)>0$ will be made sufficiently large depending on $\tad_0>0$, but independent of $\tad$ within the range above. With \eqref{notation:ccr} in mind, this is a particular case of the {\it singular regime} in \cite{CafassoClaeysRuzza2021}.

For this asymptotic regime, we need the following result.

\begin{prop}\label{prop:EstPhiccSing}
For any $\tad_0\in (0,1)$ there exists $\sad_0=\sad_0(\tad_0)>0$, $M=M(\tad_0)>0$ and $\eta=\eta(\tad_0)>0$ such that the inequalities
\begin{align*}
& \left|\Phiccp(\zeta)\bm E_{12}\Phiccp(\zeta)^{-1}\right|\leq M\ee^{-\eta \re(\zeta^{3/2})}, && \zeta\in \cSig_0, \\
& \left|\Phiccp(\zeta)\bm E_{21}\Phiccp(\zeta)^{-1}\right|\leq M\ee^{-\eta \re(\zeta^{3/2})}, && \zeta\in \cSig_1\cup\cSig_3, \quad \text{and}\\
& \left|\Phiccp(\zeta)\bm E_{22}\Phiccp(\zeta)^{-1}\right|\leq M|\zeta|^{1/2}, && \zeta\in \cSig_2,
\end{align*}
hold true for any $\sad\geq \sad_0$ and any $\tad\in [\tad_0,1/\tad_0]$.
\end{prop}

The proof of Proposition~\ref{prop:EstPhiccNonasymp} is a recollection of the analysis in \cite{CafassoClaeysRuzza2021}, so before going into the details we need to review some further notions from their work.

Introduce
\begin{equation}\label{deff:Phiai}
\Phiai(\zeta)\deff-\sqrt{2\pi}
\times 
\begin{cases}
\begin{pmatrix}
\ai'(\zeta) & -\ee^{2\pi \ii/3}\ai'(\ee^{-2\pi \ii/3}\zeta ) \\
\ii \ai(\zeta) & -\ii\ee^{-2\pi \ii/3}\ai(\ee^{-2\pi \ii/3}\zeta)
\end{pmatrix}, & \im \zeta>0, \\
\begin{pmatrix}
\ai'(\zeta) & \ee^{-2\pi \ii/3}\ai'(\ee^{2\pi \ii/3}\zeta ) \\
\ii \ai(\zeta) & \ii\ee^{2\pi \ii/3}\ai(\ee^{2\pi \ii/3}\zeta)
\end{pmatrix}, & \im \zeta<0.
\end{cases}
\end{equation}
This is the matrix appearing in \cite[Equation~(2.5)]{CafassoClaeysRuzza2021}. With the correspondence of variables \eqref{notation:ccr} in mind, when we combine our identity \eqref{eq:transfPsiccPhicc} with \cite[Equation~(2.8)]{CafassoClaeysRuzza2021}, we obtain the equality
\begin{equation}\label{eq:IdentPhiccPhiaiY}
\Phicc(\zeta)=\bm Y(\zeta)\Phiai(\zeta)
\times
\begin{cases}
\bm I, & -2\pi/3<\arg \zeta<2\pi/3,\\
\bm I\mp (1+\ee^{-\sad+\tad \zeta})\bm E_{21}, & 2\pi /3<\pm \arg \zeta <\pi.
\end{cases}
\end{equation}
The exact form of the matrix $\bm Y$ is not important for us, but we can interpret this last equality as a defining identity for $\bm Y$. What is important for us is that $\bm Y$ is analytic off the real axis, with a jump matrix $\bm J_{\bm Y}$ on $\R$ which admits an analytic continuation to a neighborhood of the axis. 

The small norm theory for $\bm Y$ in our regime of interest was carried out in \cite[Lemma~5.1 and Section~5.2]{CafassoClaeysRuzza2021}. As a consequence, we obtain that for any $\tad_0>0$ there exist $M=M(\tad_0)>0,\sad_0=\sad_0(\tad_0)>0,\eta=\eta(\tad_0)>0$ such that the inequalities
$$
\|\bm J_{\bm Y}-\bm I\|_{L^{2}\cap L^\infty\cap L^1(\R)}\leq M\ee^{-\eta s},\quad \text{and}\quad \|\bm Y_\pm-\bm I\|_{L^{2}\cap L^\infty\cap L^1(\R)}\leq M\ee^{-\eta s}
$$
hold for any $\sad\geq \sad_0$ and any $\tad \in [\tad_0,1/\tad_0]$. Also as a consequence of the small norm theory, we obtain the expression
$$
\bm Y(\zeta)=\bm I+\frac{1}{2\pi \ii}\int_\R \frac{\bm Y_-(x)(\bm J_{\bm Y}(x)-\bm I)}{x-\zeta} \dd x,\quad \zeta\in \C\setminus \R.
$$
We combine this last identity with the fact that $\bm J_{\bm Y}$ admits an analytic continuation in a neighborhood of $\R$, and learn that there exists $M=M(\msf t_0)>0$ for which
\begin{equation}\label{eq:boundbmY}
|\bm Y(\zeta)^{\pm 1}|\leq M,
\end{equation}
for every $\zeta\in \C$, $\sad\geq \sad_0$ and $\tad\in [\tad_0,1/\tad_0]$.

\begin{proof}[Proof of Proposition~\ref{prop:EstPhiccSing}]
For $\zeta\in \cSig_0=(0,\infty)$, we use \eqref{eq:IdentPhiccPhiaiY} and the definition of $\Phiai$ in 
$$
\Phiccp(\zeta)\bm E_{12}\Phiccp(\zeta)^{-1}=
2\pi \bm Y_+(\zeta)
\begin{pmatrix}
-\ii \ai(\zeta)\ai'(\zeta) & \ai'(\zeta)^2 \\ 
\ai(\zeta)^2 & \ii \ai(\zeta)\ai'(\zeta)
\end{pmatrix}
\bm Y_+(\zeta)^{-1}
$$
Using the bound \eqref{eq:boundbmY}, the continuity and the known asymptotics as $\zeta\to \infty$ of the Airy function and its derivative, the claim along $\cSig_0$ follows. 

The claim for $\zeta\in \cSig_j$ with $j=1,2,3$ follows in exactly the same explicit manner, we skip the details.
\end{proof}

\subsection{The non-asymptotic regime}\label{sec:Phiccnonasymreg}\hfill

In the non-asymptotic regime, we fix {\it any} $\tad_0\in (0,1)$ and $\sad_0>0$ and seek for bounds of certain entries of $\Phicc$ which are valid uniformly within the range
$$
|\sad|\leq \sad_0\quad \text{and}\quad \tad_0\leq \tad\leq \frac{1}{\tad_0}.
$$
For the next result, we recall the matrix norm introduced in \eqref{deff:matrixnorm}.

\begin{prop}\label{prop:EstPhiccNonasymp}
Fix any values $\tad_0\in (0,1)$ and $\sad_0>0$. There exist $M=M(\sad_0,\tad_0)>0$ and $\eta=\eta(\sad_0,\tad_0)>0$ for which the estimates
\begin{align*}
& \left|\Phiccp(\zeta)\bm E_{12}\Phiccp(\zeta)^{-1}\right|\leq M\ee^{-\eta \re(\zeta^{3/2})}, && \zeta\in \cSig_0, \\
& \left|\Phiccp(\zeta)\bm E_{21}\Phiccp(\zeta)^{-1}\right|\leq M\ee^{-\eta \re(\zeta^{3/2})}, && \zeta\in \cSig_1\cup\cSig_3, \quad \text{and}\\
& \left|\Phiccp(\zeta)\bm E_{22}\Phiccp(\zeta)^{-1}\right|\leq M|\zeta|^{1/2}, && \zeta\in \cSig_2.
\end{align*}
hold true uniformly for $|\sad|\leq \sad_0$ and $\tad_0\leq \tad \leq \tad_0^{-1}$.
\end{prop}

\begin{proof}

The asymptotic behavior as $\zeta\to\infty$ in the {\bf RHP}--{$\Phicc$} is valid uniformly up to the boundary values $\Phiccpm$ as well, and also uniformly when the parameters $\sad$ and $\tad$ vary within compact sets, implying that
$$
|\Phiccp(\zeta)\bm E_{12}\Phiccp(\zeta)^{-1}|\leq \ee^{-\frac{4}{3}\re (\zeta^{3/2})} | \zeta^{\sp_3/4}\bm U_0^{-1}\bm E_{12}\bm U_0\zeta^{-\sp_3/4} |(1+\Boh(\zeta^{-1})),\quad \zeta\to\infty.
$$
Combined with the continuity of the boundary value $\Phiccp$ with respect to both $\zeta$ and also $\sad,\tad$, the first estimate follows. The remaining estimates are completely analogous.
\end{proof}

\section{Asymptotic analysis for the model problem with admissible data}\label{sec:asympanalymodelprobl}

We now carry out the asymptotic analysis as $\tauad\to +\infty$ of $\bm\Phi_\tauad$ introduced in \eqref{deff:Phit}. For that, we fix $\sad_0>0$ and $\tad_0\in (0,1)$ and work under the assumption that
\begin{equation}\label{scaling:UTregime}
\tauad\to +\infty\quad \text{with}\quad \sad \geq -\sad_0 \quad \text{and}\quad \tad_0\leq \tad \leq \frac{1}{\tad_0}.
\end{equation}
During this section, $\msf h_\tauad$ always denotes an admissible function in the sense of Definition~\ref{deff:admissibleh}, and $\bm\Phi_\tauad$ is the solution to the associated RHP. 

We also talk about uniformity of error terms in the parameter $\tad$ ranging on a compact interval $K\subset (0,\infty)$, and by this we mean the following. The solution $\bm\Phi_\tauad$ depends on the parameter $\tad$ via the derivative $\msf H'(0)=-\tad$, see Definition~\ref{deff:admissibleh}. We view $\msf H=\msf H_\tad$ as varying with $\tad$ while keeping all the remaining derivatives $\msf H^{(k)}(0)$, $k\neq 1$ fixed. By analyticity this determines $\msf H$ uniquely at $D_\delta(0)$, but not outside this disk. We then consider $\msf H_\tad$ outside $D_\delta(0)$ to be any extension from $D_\delta(0)$ that satisfies Definition~\ref{deff:admissibleh} with the additional requirement that the constants $\eta, \widehat\eta$ and $\epsilon$ in (iii) may depend on $K$ but are independent of $\tad\in K$. Of course, for each $\msf H_\tad$ extended this way there corresponds a solution $\bm\Phi_\tauad$ of the associated RHP. By uniformity in $\tad\in K$ we mean that the error may depend on $K$ and the corresponding values $\eta, \widehat\eta$ and $\epsilon$, but is valid for any $\bm\Phi_\tauad$ obtained with an extension $\msf H_\tad$ constructed with the explained requirement.

The asymptotic analysis itself makes use of somewhat standard arguments and objects in the RHP literature. Some consequences of this asymptotic analysis will be needed later, and we now state them. 

The first such consequence is the existence of a solution with asymptotic formulas relating quantities of interest with the corresponding quantities in the id-PII RHP.

\begin{theorem}\label{thm:PhitauAsympt}
Fix an admissible function $\msf h_\tauad$ in the sense of Definition~\ref{deff:admissibleh}. There exists $\tauad_0=\tauad_0(\sad_0,\tad_0)>0$ for which for any $\tauad\geq \tauad_0$ and any $\sad,\tad$ as in \eqref{scaling:UTregime}, the RHP for $\bm\Phi(\cdot\mid \msf h_\tauad)$ admits a unique solution $\bm\Phi=\bm\Phi_\tauad$ as in \eqref{deff:Phit}.

Furthermore, for any $\kappa\in (0,1)$, the coefficient $\bm\Phi^{(1)}=\bm\Phi^{(1)}_\tauad $ in the asymptotic condition \eqref{eq:ModelRHPAsymp} satisfies
\begin{equation}\label{eq:asympPhin1Phi01}
\bm\Phi^{(1)}_\tauad=\Phicc^{(1)}+\Boh(\tauad^{-\kappa}),\quad \tau\to +\infty,
\end{equation}
where $\bm\Phi^{(1)}_0$ is as in \eqref{deff:Phikpz_residue} and the error term is uniform for $\sad,\tad$ as in \eqref{scaling:UTregime}. Also, still for $\kappa\in (0,1)$ the asymptotic formula
\begin{equation}\label{eq:asympPhitauPhi0}
\bm\Phi_{\tauad,+}(x)=\left(\bm I+\Boh\left(\frac{1}{\tauad^{\kappa}(1+|x|)}\right)\right)\Phiccp(x),\quad \tauad \to +\infty
\end{equation}
holds true uniformly for $x\in \cSig$ with $|x|\leq \tau^{(1-\kappa)/2}$, and uniformly for $\sad,\tad$ as in \eqref{scaling:UTregime}. 
\end{theorem}

The second consequence connects the solution $\bm\Phi_\tauad$ directly with the statistics $\msf Q$. For its statement, set
$$
\bm\Delta_\tau(x)\deff \bm I+(1+\ee^{-\msf h_\tau(x)})\chi_{(-\infty,0)}(x)\bm E_{21},\quad x\in \R.
$$

\begin{theorem}\label{thm:integralPhitauPhicc}
Fix $a,b>0$, $\sad_0>0$ and $\tad_0\in (0,1)$. For any $\kappa\in (0,1)$, the estimate
\begin{multline*}
\frac{1}{2\pi \ii}
\int_{\sad}^{\infty}\int_{-\tauad a}^{\tauad b}\frac{\ee^{\msf h_\tauad(x\mid u)}}{\left(1+\ee^{\msf h_\tauad(x\mid u)}\right)^2}\left[ \bm \Delta_\tauad(x\mid u)^{-1}\bm \Phi_{\tauad,+}(x\mid u)^{-1}\left(\bm \Phi_{\tauad,+} \bm\Delta_\tauad\right)'(x\mid u)\right]_{21}\dd x \dd u
=\\
-\log\Qcc(-\sad/\tad,\tad^3)
+\Boh(\tau^{-\kappa})
\end{multline*}
holds as $\tau\to +\infty$, uniformly for $\sad\geq -\sad_0$ and $\tad_0\leq \tad\leq 1/\tad_0$.
\end{theorem}

For the proof of these results, we compare $\bm\Phi_\tauad$ with the solution $\Phicc$ of the id-PII RHP via the Deift-Zhou nonlinear steepest descent method. The required asymptotic analysis itself is carried out Section~\ref{sec:AsymAnalidRHP}, and the proofs of Theorems~\ref{thm:PhitauAsympt} and \ref{thm:integralPhitauPhicc} are completed in Section~\ref{sec:proofadmRHPtakeI}.

\subsection{Asymptotic analysis}\label{sec:AsymAnalidRHP}\hfill

For $\Phicc$ as introduced in \eqref{eq:PhiccPhikpz} and whose properties were discussed in Section~\ref{sec:KPZRHPKPZ}, we perform the transformation
\begin{equation}\label{eq:transfPhitauPsitau}
\bm\Psi_\tauad(\zeta)=\bm\Phi_\tauad(\zeta)\Phicc(\zeta)^{-1},\quad \zeta\in\C\setminus\cSig.
\end{equation}
Then $\bm\Psi_\tauad$ satisfies the following RHP.

\begin{enumerate}[\bf $\bm \Psi_\tauad$-1.]
\item The matrix $\bm \Psi_\tauad:\C\setminus \cSig\to \C^{2\times 2}$ is analytic.
\item Along the interior of the arcs of $\cSig$ the function $\bm \Psi_\tauad$ admits continuous boundary values $\bm \Psi_{\tauad,\pm}$
 related by $\bm \Psi_{\tauad,+}(\zeta)=\bm \Psi_{\tauad,-}(\zeta)\bm J_{\bm \Psi_\tauad}(\zeta)$, $\zeta\in \cSig$. With
 $$
\lcc(\zeta)\deff \frac{1}{1+\ee^{-\hcc(\zeta)}}, \quad \lambda_\tauad(\zeta)\deff \frac{1}{1+\ee^{-\msf h_\tauad(\zeta)}},
 $$
 where $\hcc$ is as in \eqref{deff:h0}, the jump matrix $\bm J_{\bm \Psi_\tauad}$ is
\begin{equation}\label{eq:jumpPsimodel}
\bm J_{\bm \Psi_\tauad}(\zeta)\deff 
\begin{dcases}
\bm I+\left(\lambda_\tauad(\zeta)-\lcc(\zeta)\right) \Phiccp(\zeta)\bm E_{12}\Phiccp(\zeta)^{-1}, & \zeta \in  \cSig_0 , \\
\bm I+\left(\frac{1}{\lambda_\tauad(\zeta)}-\frac{1}{\lcc(\zeta)}\right)\Phiccp(\zeta)\bm E_{21}\Phiccp(\zeta)^{-1}, & \zeta \in \cSig_1\cup\cSig_3, \\
\frac{\lambda_\tauad(\zeta)}{\lcc(\zeta)}\Phiccp(\zeta)\bm E_{11}\Phiccp(\zeta)^{-1}+\frac{\lcc(\zeta)}{\lambda_\tauad(\zeta)}\Phiccp(\zeta)\bm E_{22}\Phiccp(\zeta)^{-1}, & \zeta\in \cSig_2.
\end{dcases}
\end{equation}
\item For $\bm \Phi_\tauad^{(1)}$ and $\Phicc^{(1)}$ the residues at $\infty$ of $\bm \Phi_\tauad$ and $\Phicc$, respectively, the matrix $\bm\Psi_\tauad$ has the asymptotic behavior
\begin{equation}\label{eq:Phitauad1Phicc1}
\bm \Psi_\tauad(\zeta)=\bm I+\frac{1}{\zeta}(\bm \Phi_\tauad^{(1)}-\Phicc^{(1)})+\Boh(1/\zeta^2)\qquad \text{as}\quad \zeta\to\infty.
\end{equation}

\item The matrix $\bm \Psi_\tauad$ remains bounded as $\zeta \to 0$.
\end{enumerate}

The next step is to verify that the jump matrix decays to the identity in the appropriate norms. The terms in the jump that come from $\Phicc$ are precisely the ones we already estimated in Sections~\ref{sec:Phiccsingregime} and \ref{sec:Phiccnonasymreg}, so it remains to estimate the terms involving the $\lambda$-functions. The basic needed estimate is the following lemma.

\begin{lemma}\label{lem:convhtauhcc}
Fix $\nu\in (0,1/2)$ and $\tad_0\in (0,1)$. The estimate
$$
\ee^{\hcc(\zeta)-\msf h_\tauad(\zeta)}=1+\Boh(\zeta^2/\tau), \quad \tauad\to\infty,
$$
holds true uniformly for $|\zeta|\leq \tauad^\nu$ and uniformly for $\tad_0\leq \tad\leq 1/\tad_0$, where the error term is independent of $\sad\in \R$.
\end{lemma}
\begin{proof}
The Definition~\eqref{deff:admissibleh} of admissibility of $\msf h_\tau$ ensures that for $\tauad$ sufficiently large, we can expand the term $\msf H(\zeta/\tau)$ in power series near the origin and obtain the expansion
$$
\msf h_\tauad(\zeta)=\sad -\tad \zeta +\Boh(\zeta^2/\tauad),\quad \tauad\to +\infty,
$$
valid uniformly for $|\zeta|\leq \tauad^\nu$, $\tad_0\leq \tad\leq 1/\tad_0$, and with error independent of $\sad\in \R$. Recalling that $\hcc(\zeta)=\sad-\tad\zeta$, the proof is complete.
\end{proof}

We are now able to prove the appropriate convergence of $\bm J_{\bm \Psi_\tauad}$ to the identity matrix. We split the analysis into three lemmas, corresponding to different pieces of the contour $\cSig$. In the results that follow we use the matrix norm notations introduced in \eqref{deff:matrixnorm}--\eqref{deff:matrixLpqnorm}.

\begin{lemma}\label{lem:estJpsiadcSig0}
Fix $\tad_0\in (0,1)$, $\sad_0>0$ and $\nu\in (0,1/2)$. There exist $\tauad_0=\tauad_0(\tad_0,\sad_0,\nu)>0$, $M=M(\tad_0,\sad_0,\nu)>0$ and $\eta=\eta(\tad_0,\sad_0,\nu)>0$ for which the inequality
$$
\| \bm J_{\bm\Psi_\tauad}-\bm I \|_{L^1\cap L^2\cap L^\infty(\cSig_0)}\leq M\ee^{-\sad} \max \left\{ \tauad^{-1+2\nu},\ee^{-\eta \tauad^{3\nu/2}} \right\}
$$
holds true for any $\tauad\geq \tauad_0$, $\sad\geq -\sad_0$ and $\tad\in [\tad_0,1/\tad_0]$.
\end{lemma}

\begin{proof}
Because both $\msf h_\tauad$ and $\hcc$ are real-valued along the real line, the inequality
$$
\left|\lambda_\tauad(\zeta)-\lcc(\zeta)\right|=\frac{|\ee^{-\msf h_\tauad(\zeta)}-\ee^{-\hcc(\zeta)}|}{(1+\ee^{-\msf h_\tauad(\zeta)})(1+\ee^{-\hcc(\zeta)})}\leq 
|\ee^{-\msf h_\tauad(\zeta)}-\ee^{-\hcc(\zeta)}|
$$
is immediate. For $0\leq \zeta\leq \tauad^\nu$, we then use Lemma~\ref{lem:convhtauhcc} and the explicit expression for $\hcc$ in \eqref{deff:h0} and obtain
\begin{equation}\label{eq:estlambalcccSig01}
\left|\lambda_\tauad(\zeta)-\lcc(\zeta)\right|=\Boh\left(\frac{\ee^{-\sad+\tad \zeta}}{\tauad^{1-2\nu}}  \right).
\end{equation}
For $\zeta\geq \tauad^\nu$, we instead use that both $\msf h_\tauad$ and $\hcc$ are real-valued along the positive axis and write
$$
|\lambda_\tauad(\zeta)-\lcc(\zeta)| \leq \left|\lambda_\tauad(\zeta)-1  \right| +\left|\lambda_\tauad(\zeta)-1  \right| =\frac{\ee^{-\sad+\tad \zeta}}{1+\ee^{-\sad+\tad\zeta}}+\frac{\ee^{-\sad-\tauad\msf H(\zeta/\tauad)}}{1+\ee^{-\sad-\tauad\msf H(\zeta/\tauad)}}\leq \ee^{-\sad}\left(\ee^{\tad \zeta}+\ee^{-\tauad\msf H(\zeta/\tauad)}\right).
$$
From Definition~\ref{deff:admissibleh}--(iii) we bound $\ee^{-\tauad\msf H(\zeta/\tauad)}\leq \ee^{\wh \eta\zeta^{3/2-\epsilon}}$ and simplify the last inequality to
\begin{equation}\label{eq:estlambalcccSig02}
|\lambda_\tauad(\zeta)-\lcc(\zeta)| \leq \ee^{-\sad}\ee^{\tilde \eta \zeta^{\alpha}},\quad \zeta\geq \tauad^\nu, \; \alpha \deff \max \{1,3/2-\epsilon\}<\frac{3}{2} ,
\end{equation}
for a new value $\tilde \eta>0$.

Recall that $\bm J_{\bm\Psi_\tauad}$ was given in \eqref{eq:jumpPsimodel}. We use \eqref{eq:estlambalcccSig01} and \eqref{eq:estlambalcccSig02} in combination with Propositions~\ref{prop:EstPhiccSing} and \ref{prop:EstPhiccNonasymp} to get the existence of a value $\tauad_0>0$ for which
\begin{equation}\label{eq:PtwEstPhitauadcSig0}
|\bm J_{{\bm \Psi}_\tauad}(\zeta)-\bm I| \leq M \ee^{-\sad}
\ee^{-\eta \zeta^{3/2}} \left( \chi_{(0,\tauad^{\nu})}(\zeta) \frac{\ee^{\tad \zeta}}{\tauad^{1-2\nu}}+\ee^{\tilde \eta \zeta^\alpha}\chi_{(\tauad^{\nu},+\infty)}(\zeta)\right),\quad \tauad\geq \tauad_0,
\end{equation}
where $\eta>0,M>0$ may depend on $\sad_0,\tad_0$ and $\nu\in [0,1/2)$, but are independent of $\sad\geq -\sad_0$ and $\tad_0\leq \tad\leq 1/\tad_0$. After appropriately changing the values of $\tilde\eta,\eta,M$, and having in mind that $\alpha<3/2$, the result follows from this inequality.
\end{proof}

Next, we prove the equivalent result along the pieces of $\cSig$ which are not on the real line.

\begin{lemma}\label{lem:estJpsiadcSig13}
Fix $\tad_0\in (0,1)$, $\sad_0>0$ and $\nu\in (0,1/2)$. There exist $\tauad_0=\tauad_0(\tad_0,\sad_0,\nu)>0$, $M=M(\tad_0,\sad_0,\nu)>0$ and $\eta=\eta(\tad_0,\sad_0,\nu)>0$, for which the inequality
$$
\| \bm J_{\bm\Psi_\tauad}-\bm I \|_{L^1\cap L^2\cap L^\infty(\cSig_1\cup \cSig_3)}\leq M\ee^{-\sad} \max \left\{ \tauad^{-1+2\nu},\ee^{-\eta \tauad^{3\nu/2}} \right\}
$$
holds true for any $\tauad\geq \tauad_0$, $\sad\geq -\sad_0$ and $\tad\in [\tad_0,1/\tad_0]$.
\end{lemma}

\begin{proof}
Write
$$
\frac{1}{\lambda_\tauad(\zeta)}-\frac{1}{\lcc(\zeta)}=\ee^{-\msf h_\tauad(\zeta)}-\ee^{-\hcc(\zeta)}=-\ee^{-\sad +\tad \zeta}(1-\ee^{\hcc(\zeta)-\msf h_\tauad(\zeta)}).
$$
From Lemma~\ref{lem:convhtauhcc}, we estimate for $0\leq |\zeta|\leq \tauad^\nu$,
$$
\frac{1}{\lambda_\tauad(\zeta)}-\frac{1}{\lcc(\zeta)}=\Boh\left(\frac{\ee^{-\sad +\tad \re \zeta}}{\tau^{1-2\nu}}\right),\quad \tau\to \infty,
$$
where the implicit error term is independent of $\sad$ and uniform for $\tad\in [\tad_0,1/\tad_0]$. On the other hand, from the explicit form of $\hcc$ and Definition~\ref{deff:admissibleh}--(iii),
$$
\left|\frac{1}{\lambda_\tauad(\zeta)}-\frac{1}{\lcc(\zeta)}\right|
\leq 
\ee^{-\sad}\left(\ee^{\tad\re\zeta}+\ee^{-\eta|\zeta|}\right).
$$
We combine this inequality with Propositions~\ref{prop:EstPhiccSing} and \ref{prop:EstPhiccNonasymp} and use them on \eqref{eq:jumpPsimodel}. The conclusion is that there exist $M>0$, $\eta_1,\eta_2>0$ and $\tauad_0>0$, depending on $\nu,\tad_0,\sad_0$, for which the inequality
\begin{equation}\label{eq:PtwEstPhitauadcSig13}
\left| \bm J_{\bm \Psi_\tauad}(\zeta)-\bm I\right|\leq M \ee^{-\sad} \ee^{-\eta_1\re \zeta^{3/2}+\eta_2\re \zeta}
\left(  \frac{1}{\tau^{1-2\nu}}\chi_{\{|\zeta|\leq \tauad^\nu\}}(\zeta)+\chi_{\{|\zeta|\geq \tauad^\nu\}}(\zeta)\right),\quad \zeta\in \cSig_1\cup \cSig_3,
\end{equation}
is valid for every $\tauad\geq \tauad_0,$ $\sad\geq -\sad_0$ and $\tad\in [\tad_0,1/\tad_0]$. The definition \eqref{def:contour_sigma} of the contours $\cSig_1$ and $\cSig_3$ assure us that $\re \zeta^{3/2}>0$ and $\re \zeta<0$ on these contours. After possibly changing the values of the constants $\eta_1,\eta_2$ and $M$, the result follows.
\end{proof}

Finally, we now handle the jump on the negative axis.

\begin{lemma}\label{lem:estJpsiadcSig2}
Fix $\tad_0\in (0,1)$, $\sad_0>0$ and $\nu\in (0,1/2)$. There exist $\tauad_0=\tauad_0(\tad_0,\sad_0,\nu)>0, M=M(\tad_0,\sad_0,\nu)>0$ and $\eta=\eta(\tad_0,\sad_0,\nu)>0$ for which the inequality
$$
\| \bm J_{\bm\Psi_\tauad}-\bm I \|_{L^1\cap L^2\cap L^\infty(\cSig_2)}\leq M \ee^{-\sad} \max\left\{ \tauad^{-1+2\nu},\ee^{-\eta\tauad^{\nu}} \right\}
$$
holds true for any $\tauad\geq \tauad_0$, $\sad\geq -\sad_0$ and $\tad\in [\tad_0,1/\tad_0]$.
\end{lemma}

\begin{proof}
The initial step is to rewrite the last line of \eqref{eq:jumpPsimodel} as
\begin{equation}\label{eq:estJpsiadcSig1}
\bm J_{\bm\Psi_{\tauad}}(\zeta)-\bm I=\left(\frac{\lambda_\tauad(\zeta)}{\lcc(\zeta)}-1\right)\bm I+\left(\frac{\lcc(\zeta)}{\lambda_\tauad(\zeta)}-\frac{\lambda_\tauad(\zeta)}{\lcc(\zeta)}\right)\Phiccp(\zeta)\bm E_{22}\Phiccp(\zeta)^{-1},\quad \zeta\in \cSig_2.
\end{equation}
The identities
$$
\frac{\lambda_\tauad(\zeta)}{\lcc(\zeta)}-1=\frac{\ee^{-\hcc(\zeta)}-\ee^{-\msf h_\tauad(\zeta)}}{1+\ee^{-\msf h_\tauad(\zeta)}},\qquad 
\frac{\lcc(\zeta)}{\lambda_\tauad(\zeta)}-1=-\frac{\ee^{-\hcc(\zeta)}-\ee^{-\msf h_\tauad(\zeta)}}{1+\ee^{-\hcc(\zeta)}},
$$
are trivial, and because $\hcc$ and $\msf h_\tauad$ are real-valued along $\cSig_2=(-\infty,0)$, these equalities give
$$
\left|\frac{\lambda_\tauad(\zeta)}{\lcc(\zeta)}-1\right|+\left| \frac{\lcc(\zeta)}{\lambda_\tauad(\zeta)}-1\right|
\leq 
2\left|\ee^{-\hcc(\zeta)}-\ee^{-\msf h_\tauad(\zeta)}\right|,\quad \zeta<0.
$$
For $|\zeta|\leq \tauad^\nu$ we use Lemma~\ref{lem:convhtauhcc} and estimate
$$
\left|\frac{\lambda_\tauad(\zeta)}{\lcc(\zeta)}-1\right|+\left| \frac{\lcc(\zeta)}{\lambda_\tauad(\zeta)}-1\right|=\Boh\left(\frac{\ee^{-\sad -\tad |\zeta|}}{\tau^{1-2\nu}}\right),\quad \tauad\to +\infty,\quad -\tau^{\nu}\leq \zeta\leq 0,
$$
whereas for $\zeta \leq -\tauad^\nu$ we use instead the definition of $\hcc$ in \eqref{deff:h0} and Definition~\ref{deff:admissibleh}--(iii) and write
$$
\left|\frac{\lambda_\tauad(\zeta)}{\lcc(\zeta)}-1\right|+\left| \frac{\lcc(\zeta)}{\lambda_\tauad(\zeta)}-1\right|
\leq 2\ee^{-\sad -(\tad+\eta)|\zeta|}.
$$
We combine these two inequalities with Propositions~\ref{prop:EstPhiccSing} and \ref{prop:EstPhiccNonasymp}, and apply them to \eqref{eq:estJpsiadcSig1}. As a result, we learn that there exist $M>0,\eta>0,\tauad_0>0$ for which the estimate
\begin{equation}\label{eq:EstJPsitauadcSig2}
\left| \bm J_{\bm \Psi_\tauad}(\zeta)-\bm I \right|\leq M|\zeta|^{1/2}\ee^{-\sad-\eta |\zeta|} \left(\frac{1}{\tau^{1-2\nu}}\chi_{(-\tauad^\nu,0]}(\zeta)+\chi_{(-\infty,-\tauad^\nu)}(\zeta)\right),\quad \zeta\leq 0,
\end{equation}
is valid for any $\sad\geq -\sad_0,$ $\tad\in [\tad_0,1/\tad_0]$, $\tauad\geq \tauad_0$. After possibly changing the values of $\eta,M$, the result follows from standard arguments.
\end{proof}

Now that we controlled the asymptotic behavior for the jump matrix $\bm J_{\bm\Psi_\tauad}$, we are ready to obtain small norm estimates for $\bm \Psi_{\tauad}$ itself. We summarize these estimates in the next result. For that, we recall the matrix norm notations introduced in \eqref{deff:matrixnorm},\eqref{deff:matrixLpnorm},\eqref{deff:matrixLpqnorm}.

\begin{theorem}\label{thm:PsitauSmallnorm}
Fix $\tad_0\in (0,1)$ and $\sad_0>0$. There exists $\tauad_0=\tauad_0(\tad_0,\sad_0)>0$ for which the solution $\bm \Psi_\tauad$ uniquely exists for any $\tauad\geq \tauad_0$ and any $\sad \geq -\sad_0, \tad\in [\tad_0,1/\tad_0]$. Furthermore, it satisfies the following asymptotic properties. 

Its boundary value $\bm \Psi_{\tauad,-}$ exists along $\cSig$, and satisfies the estimate
$$
\|\bm \Psi_{\tauad,-}-\bm I \|_{L^2(\cSig)}=\Boh\left(\tauad^{-\kappa}\right),\quad \tau \to +\infty,
$$
for any $\kappa\in (0,1)$, where the error term, for a given $\kappa$, is uniform for $\sad \geq -\sad_0$ and $\tad\in [\tad_0,1/\tad_0]$.

For $\tau$ sufficiently large, the solution $\bm \Psi_\tauad$ admits the representation
\begin{equation}\label{eq:IntReprPsitau}
\bm \Psi_\tauad(\zeta)=\bm I+\frac{1}{2\pi \ii}\int_{\cSig}\frac{\bm \Psi_{\tauad,-}(w)(\bm J_{\bm \Psi_\tauad}(w)-\bm I)}{w-\zeta}\dd w,\quad \zeta\in \C\setminus \cSig.
\end{equation}
Still for $\tau$ sufficiently large, $\bm \Psi_\tauad$ satisfies
\begin{equation}\label{eq:Psitauad1}
\bm \Psi_\tauad(\zeta)=\bm I+\bm\Psi_{\tauad}^{(1)}\frac{1}{\zeta}+\Boh(\zeta^{-2}),\quad \zeta\to \infty,\quad \text{with} \quad 
\bm\Psi_{\tauad}^{(1)}\deff -\frac{1}{2\pi \ii}\int_{\cSig} \bm\Psi_{\tauad,-}(\xi)(\bm J_{\bm \Psi_{\tauad}}(\xi)-\bm I)\dd \xi.
\end{equation}

\end{theorem}
\begin{proof}
The small norm estimates provided by Lemmas~\ref{lem:estJpsiadcSig0}, \ref{lem:estJpsiadcSig13} and \ref{lem:estJpsiadcSig2} allow us to apply the small norm theory for Riemann-Hilbert problems (see for instance \cite{DeiftParkCity, DeiftZhou93mKdV}), and the claims follow with standard methods. We stress that for this statement we only need the $L^2$ and $L^\infty$ estimates from the aforementioned lemmas, but the $L^1$ estimates provided by them will be useful later.
\end{proof}

\subsection{Proof of main results of the section}\label{sec:proofadmRHPtakeI}\hfill 

We are ready to prove the main results of this section. 

\begin{proof}[Proof of Theorem~\ref{thm:PhitauAsympt}]
During the whole proof we identify $\nu=(1-\kappa)/2$.

The matrix $\Phicc$ always exists, whereas Theorem~\ref{thm:PsitauSmallnorm} provides the existence of $\bm\Psi_{\tauad}$ for $\tauad$ sufficiently large. From the relation \eqref{eq:transfPhitauPsitau} we obtain the claimed existence of $\bm\Phi_\tauad$.

Comparing \eqref{eq:Phitauad1Phicc1} with \eqref{eq:Psitauad1} we obtain the identity
$$
\bm\Phi_\tauad^{(1)}=\Phicc^{(1)}+\bm\Psi_{\tauad}^{(1)}.
$$%
Writing 
\begin{equation}\label{eq:EstPsi1tauad}
\bm\Psi_{\tauad}^{(1)}=-\frac{1}{2\pi \ii}\int_{\cSig} (\bm\Psi_{\tauad,-}(\xi)-\bm I)(\bm J_{\bm \Psi_{\tauad}}(\xi)-\bm I)\dd \xi -\frac{1}{2\pi \ii}\int_{\cSig} \bm(\bm J_{\bm \Psi_{\tauad}}(\xi)-\bm I)\dd \xi,
\end{equation}
and using Cauchy-Schwartz,
$$
|\bm\Psi_{\tauad}^{(1)}|\leq \|\bm\Psi_{\tauad,-}-\bm I\|_{L^2(\cSig)} \|\bm J_{\bm \Psi_{\tauad}}-\bm I\|_{L^2(\cSig)}+\|\bm J_{\bm \Psi_{\tauad}}-\bm I\|_{L^1(\cSig)},
$$
and from Lemmas~\ref{lem:estJpsiadcSig0}, \ref{lem:estJpsiadcSig13}, \ref{lem:estJpsiadcSig2} and Theorem~\ref{thm:PsitauSmallnorm}, the right-hand side above is $\Boh(\tau^{-\kappa})$, for any $\kappa\in (0,1)$ and uniformly for $\sad,\tad$ as in \eqref{scaling:UTregime}, proving \eqref{eq:asympPhin1Phi01}.

To prove the asymptotic formula \eqref{eq:asympPhitauPhi0} we follow arguments presented in \cite[Theorem~3.1]{Kuijlaars2003lecnotes} and \cite[Lemma~2]{Aptekarev2002sharp}, with minor modifications to handle the uniformity on the unbounded set $|x|\leq \tauad^{-\nu}$ as claimed.

First off, the jump matrix $\bm J_{\bm \Psi_\tauad}$ is $C^\infty$ on $\cSig$, in particular H\"older continuous, implying that $\bm \Psi_\tauad$ extends continuously to its boundary values $\bm\Psi_{\tauad,\pm}$. Accounting also for the behavior of $\bm\Psi_{\tauad}$ at $\infty$ and combining with the maximum principle, 
\begin{equation}\label{eq:maxprincipleconseqPsitauad}
\|\bm \Psi_{\tauad}\|_{L^\infty(\C\setminus \cSig)}\leq  M_\tauad\deff \max\left\{ \|\bm\Psi_{\tauad,+}\|_{L^\infty(\cSig)},\|\bm\Psi_{\tauad,-}\|_{L^\infty(\cSig)} \right\},
\end{equation}
where the constant $M_\tauad$ is finite.
For a point $x\in \cSig\setminus \{0\}$ and $\varepsilon>0$, we consider the arcs $C^\pm_\varepsilon(x)$ of the disk centered at $x$ and radius $\varepsilon$ which are on the $\pm$ side of $\cSig$. We then set
$$
\cSig^\pm\deff \left(\cSig\setminus D_\varepsilon(x)\right)\cup C^\pm_\varepsilon(x),
$$
with the orientation induced from $\cSig$. We deform contour in the integral representation \eqref{eq:IntReprKernel} and then send $z\to x$, obtaining that
\begin{equation}\label{eq:intreprPsitauadpm}
\bm\Psi_{\tauad,\pm}(x)=\bm I+\frac{1}{2\pi\ii}\int_{\cSig^\mp}\frac{\bm\Psi_{\tauad,-}(s)(\bm J_{\bm \Psi_{\tauad}}(s)-\bm I)}{s-x}\dd s.
\end{equation}
From standard estimates and using \eqref{eq:maxprincipleconseqPsitauad}, the just written equation yields
$$
|\bm \Psi_{\tauad,\pm}(x)|\leq 1+\frac{1}{\pi \varepsilon}\|\bm \Psi_{\tauad,-}\|_{L^\infty(\cSig^\mp)} \|\bm J_{\bm \Psi_\tauad}-\bm I\|_{L^1(\cSig^\mp)}\leq 1+\frac{1}{\pi \varepsilon}M_\tauad \|\bm J_{\bm \Psi_\tauad}-\bm I\|_{L^1(\cSig^\mp)},
$$
and therefore
\begin{equation}\label{eq:estimateMtauad}
M_\tauad\leq \left(1-\frac{1}{\pi\varepsilon}\|\bm J_{\bm \Psi_\tauad}-\bm I\|_{L^1(\cSig^\mp)}\right)^{-1}.
\end{equation}
Lemmas~\ref{lem:estJpsiadcSig0}, \ref{lem:estJpsiadcSig13} and \ref{lem:estJpsiadcSig2} provide $L^p$ estimates for $\bm J_{\bm \Psi_\tauad}-\bm I$ along $\cSig$. Exploring that $\cSig^\pm$ is obtained from $\cSig$ after a small deformation around the point $x=\Boh(\tauad^{-\nu})$, it is straightforward to see that the same estimates hold in $\cSig^\pm$, which can be summarized as
\begin{equation}\label{eq:estJPsicSigpm}
|\bm J_{\bm \Psi_\tauad}(\zeta)-\bm I| \leq M\ee^{-\eta \min\{|\zeta|,|\zeta|^{3/2}\}}\left(\frac{1}{\tau^{\kappa}}\chi_{\{|\zeta|\leq \tauad^\nu\}}(\zeta)+\chi_{\{|\zeta|> \tauad^\nu\}}(\zeta)\right),\quad \zeta\in \cSig^\pm,
\end{equation}
for some constants $\eta,M>0$ which may depend on $\sad_0,\tad_0,\tauad_0$ but are independent of $\sad\geq \sad_0,\tad\in [\tad_0,1/\tad_0],\tauad\geq \tauad_0$, see \eqref{eq:PtwEstPhitauadcSig0}, \eqref{eq:PtwEstPhitauadcSig13} and \eqref{eq:EstJPsitauadcSig2}. Combining with \eqref{eq:estimateMtauad}, we conclude in particular that $M_\tauad\leq 2$, for $\sad,\tad,\tauad$ in the same range of values. Having in mind \eqref{eq:transfPhitauPsitau}, this bound on $M_\tauad$ applied to \eqref{eq:intreprPsitauadpm} is enough to ensure that $\bm \Psi_{\tauad,\pm}-\bm I=\Boh(\tauad^{-\nu})$, but to obtain the decay in $x$ for the error claimed in \eqref{eq:asympPhitauPhi0} a little more care is needed as follows.

First off, we split the integral in \eqref{eq:intreprPsitauadpm} into two, namely along 
$$
J_x\deff \{\zeta\in \cSig^-\mid |\zeta-x|\geq |x|/2\}\quad \text{and}\quad \cSig^-\setminus J_x.
$$
For the integral over $J_x$, we estimate as
$$
\left| \int_{J_x}\frac{\bm\Psi_{\tauad,-}(s)(\bm J_{\bm \Psi_{\tauad}}(s)-\bm I)}{s-x}\dd s \right|\leq 2M_\tau \sup_{\zeta\in J_x}\frac{1}{|\zeta-x|} \|\bm J_{\bm \Psi_{\tauad}}-\bm I \|_{L^1(J_x)}=\Boh(\tau^{-\kappa}|x|^{-1}),
$$
where we used \eqref{eq:estJPsicSigpm}, the fact that $x=\Boh(\tau^{-\kappa})$ and again the bound $M_\tauad\leq 2$. Observing that $|\zeta-x|\geq \varepsilon$ along the remaining piece, a similar argument yields
$$
\left| \frac{\bm\Psi_{\tauad,-}(s)(\bm J_{\bm \Psi_{\tauad}}(s)-\bm I)}{s-x} \right|\leq \frac{4}{\varepsilon} \left|\bm J_{\bm \Psi_{\tauad}}(s)-\bm I \right|,\quad s\in \cSig^-\setminus J_x,
$$
and again from \eqref{eq:estJPsicSigpm} we see that the right-hand side above decays exponentially in $x$ when $x\to \infty$ and is $\Boh(\tau^{-\kappa})$ when $\tau\to \infty$. From \eqref{eq:intreprPsitauadpm} we thus obtain
$$
\left|\bm\Psi_{\tauad,+}(x)-\bm I\right|=\Boh\left(\frac{1}{(1+|x|)\tauad^{\kappa}}\right),\quad \tauad\to \infty,
$$
uniformly for $x\in \cSig, |x|\leq \tauad^{-\nu}$, and uniformly for $\sad\geq -\sad_0, \tad_0\leq \tad\leq 1/\tad_0$. In virtue of \eqref{eq:transfPhitauPsitau}, this proves \eqref{eq:asympPhitauPhi0}. 
\end{proof}

\begin{remark}\label{remark:UnifAsympPhitauad}
For admissible functions $\msf h_{\tauad}$, the asymptotics \eqref{eq:ModelRHPAsymp} as $\zeta\to \infty$ of $\bm \Phi_\tauad$ is valid uniformly in $\tauad, \tad$ and $\sad$, in the sense that for any $\tad_0\in (0,1)$ and any $\sad_0>0$, there exist $K>0$ and $R>0$ such that
$$
\left|
\bm \Phi_\tauad(\zeta)\ee^{\frac{2}{3}\zeta^{2/3}\sp_3}\bm U_0\zeta^{\sp_3/4}-\bm I\right|\leq \frac{K}{|\zeta|},\quad \text{whenever } |\zeta|\geq R, \; \sad\geq -\sad_0, \; \tad_0\leq \tad\leq 1/\tad_0,
$$
and we emphasize that $K$ and $R$ are independent of $\sad,\tad$. To see that this is true, in virtue of \eqref{eq:transfPhitauPsitau} it is enough to show that the asymptotics \eqref{eq:AsymptPhicc} and \eqref{eq:Phitauad1Phicc1} are uniform in the same sense, we indicate the proof for the latter and the former is analogous.

Using the trivial identity $1/(w-\zeta)=s/(s(s-\zeta))-1/\zeta$, we express \eqref{eq:IntReprPsitau} as
$$
\bm\Psi_\tauad(\zeta)=\bm I-\frac{1}{2\pi \ii \zeta}\int_{\cSig}\bm \Psi_{\tauad,-}(w)(\bm J_{\bm \Psi_\tauad}(w)-\bm I)\dd w+\frac{1}{2\pi \ii\zeta}\int_{\cSig}\frac{ w\bm \Psi_{\tauad,-}(w)(\bm J_{\bm \Psi_\tauad}(w)-\bm I)}{w-\zeta}\dd w.
$$
Because $\bm \Psi_{\tauad,-}\in\bm I+ L^1(\cSig)$ and $\bm J_{\bm \Psi_{\tauad}}-\bm I$ decays pointwise exponentially fast (and uniformly in $\sad,\tad$ as claimed), the two integrals can be bounded uniformly in $\zeta,\tau,\sad$ as claimed, and the uniform decay for $\bm \Psi_{\tauad}$ as claimed follows.
\end{remark}

To finish this section, it remains to prove Theorem~\ref{thm:integralPhitauPhicc}, and for that end we first establish a lemma.

\begin{lemma}\label{eq:lemtailsintPhis}
Fix $\sad\in \R$ and $\tad_0\in (0,1)$. For any $\nu\in (0,1/2)$, there exists $\eta=\eta(\nu)>0$ independent of $\sad,\tad,\tauad$ for which the estimates
\begin{equation}\label{eq:EstPhiccinttailest}
\int_{|x|\geq \tauad^{\nu}} \frac{\ee^{\tad x-u}}{(1+\ee^{\tad x-u})^2}\left|\left[(\Dcc(x\mid u,\tad)^{-1}\Phiccp(x\mid u,\tad)^{-1}(\Phiccp \Dcc )'(x\mid u,\tad)\right]_{21} \right|\dd x=\Boh(\ee^{-u-\eta \tauad^{\nu}})
\end{equation}
and
\begin{equation}\label{eq:EstPhitauadinttailest}
\int_{|x|\geq \tauad^{\nu}}  \frac{\ee^{\msf h_\tauad(x\mid u)}}{\left(1+\ee^{\msf h_\tauad(x\mid u)}\right)^2}\left|\left[ \bm \Delta_\tauad(x\mid u)^{-1}\bm \Phi_{\tauad,+}(x\mid u)^{-1}\left(\bm \Phi_{\tauad,+} \bm\Delta_\tauad\right)'(x\mid u)\right]_{21}\right|\dd x =\Boh(\ee^{-u-\eta \tauad^{\nu}})
\end{equation}
are valid as $\tauad\to \infty$, uniformly for $u\geq \sad$ and $\tad_0\leq \tad\leq 1/\tad_0$. In particular, both integrands are integrable over $\R$.
\end{lemma}
\begin{proof}
We prove \eqref{eq:EstPhitauadinttailest} which is slightly more technical because the integrand depends on $\tauad$, the estimate \eqref{eq:EstPhiccinttailest} follows in a similar manner.

Having in mind Remark~\ref{remark:UnifAsympPhitauad}, we use the expansion \eqref{eq:ModelRHPAsymp} to estimate
$$
\bm\Phi_{\tauad,+}(x)\bm\Delta_{\tauad}(x)=\left(1+\Boh(x^{-1})\right)x^{\sp_3/4}\bm U_0^{-1}\left(\bm I+\chi_{(-\infty,0)}(x)\ee^{\frac{4}{3}x^{3/2}_+}(1-\ee^{-\msf h_{\tauad}(x)})\bm E_{21}\right)\ee^{-\frac{2}{3}x_+^{3/2}\sp_3},\quad x\in E_\tauad.
$$
Observe that the factor $\chi_{(-\infty,0)}(x)\ee^{\frac{4}{3}x^{3/2}_+}(1-\ee^{-\msf h_{\tauad}(x)})$ is bounded, thanks to the facts that $x_+^{3/2}\in \ii\R$ and $\re H>0$ for $x<0$ (see Definition~\ref{deff:admissibleh}--(iii)). The identity above can be differentiated, and using it we obtain the crude bound
$$
\left[ \bm \Delta_\tauad(x\mid u)^{-1}\bm \Phi_{\tauad,+}(x\mid u)^{-1}\left(\bm \Phi_{\tauad,+} \bm\Delta_\tauad\right)'(x\mid u)\right]_{21}=\ee^{-\frac{4}{3}x_+^{3/2}}\Boh(|x|^{3}),\quad x\in E_\tauad,
$$
which is non-optimal in $x$, but will be enough for the coming estimates, and which is valid uniformly for $u\geq -\sad_0,\tad\in [\tad_0,1/\tad_0]$ as $\tauad\to\infty$. Thus, to conclude the result it is enough to estimate each of the integrals
$$
I_-\deff\int_{-\infty}^{-\tauad^\nu}|x|^3\frac{\ee^{\msf h_\tauad(x\mid u)}}{(1+\ee^{\msf h_\tauad(x\mid u)})^2}\dd x\quad \text{and}\quad 
I_+\deff\int_{\tauad^{\nu}}^{+\infty}|x|^3\frac{\ee^{-\frac{4}{3}x^{3/2}}\ee^{\msf h_\tauad(x\mid u)}}{(1+\ee^{\msf h_\tauad(x\mid u)})^2}\dd x.
$$
To estimate $I_-$, we use the inequalities $v/(1+v)\leq 1$ and $1/(1+v)\leq 1/v$, valid for $v>0$, to estimate
$$
\frac{\ee^{\msf h_\tauad(x\mid u)}}{(1+\ee^{\msf h_\tauad(x\mid u)})^2}\leq \ee^{-\msf h_{\tauad}(x\mid u)},
$$
and using now the inequality from Definition~\ref{deff:admissibleh}--(iii) along $\cSig_2=(-\infty,0)$, we obtain
$$
I_-\leq \ee^{-u}\int_{-\infty}^{-\tauad^\nu}|x|^3\ee^{-\eta |x|}\dd x=\Boh(\ee^{-u-\wt\eta\tauad^{\nu}/2})
$$
for some $\wt\eta>0$. In a similar manner, we also obtain that
$$
I_+\leq \ee^{-u} \int_{\tau^\nu}^{+\infty} |x|^3 \ee^{-\frac{4}{3}x^{3/2}+\wh\eta x^{2/3-\varepsilon} }\dd x=\Boh(\ee^{-u-\tauad^{3\nu/2}}),
$$
where now for the last equality we used Definition~\ref{deff:admissibleh}--(iii) along $\cSig_0=(0,\infty)$.

\end{proof}

\begin{proof}[Proof of Theorem~\ref{thm:integralPhitauPhicc}]
As in the previous proof, we identify $\kappa\in (0,1)$ from Theorem~\ref{thm:PhitauAsympt} with $\nu=(1-\kappa)/2\in (0,1/2)$. Thanks to \eqref{eq:EstPhitauadinttailest},
\begin{multline*}
\int_{-a\tau}^{b\tau}  \frac{\ee^{\msf h_\tauad(x\mid u)}}{\left(1+\ee^{\msf h_\tauad(x\mid u)}\right)^2}\left[ \bm \Delta_\tauad(x\mid u)^{-1}\bm \Phi_{\tauad,+}(x\mid u)^{-1}\left(\bm \Phi_{\tauad,+} \bm\Delta_\tauad\right)'(x\mid u)\right]_{21}\dd x\\
=
\int_{-\tauad^\nu}^{\tauad^\nu}  \frac{\ee^{\msf h_\tauad(x\mid u)}}{\left(1+\ee^{\msf h_\tauad(x\mid u)}\right)^2}\left[ \bm \Delta_\tauad(x\mid u)^{-1}\bm \Phi_{\tauad,+}(x\mid u)^{-1}\left(\bm \Phi_{\tauad,+} \bm\Delta_\tauad\right)'(x\mid u)\right]_{21}\dd x +
\Boh(\ee^{-u-\eta \tauad^{\nu}}),
\end{multline*}
valid as $\tauad\to \infty$ and uniformly for $u\geq \sad$ and $\tad_0\leq \tad\leq 1/\tad_0$. Next, we use Lemma~\ref{lem:convhtauhcc} and \eqref{eq:asympPhitauPhi0} to ensure that
\begin{multline*}
\int_{-\tauad^\nu}^{\tauad^\nu}  \frac{\ee^{\msf h_\tauad(x\mid u)}}{\left(1+\ee^{\msf h_\tauad(x\mid u)}\right)^2}\left[ \bm \Delta_\tauad(x\mid u)^{-1}\bm \Phi_{\tauad,+}(x\mid u)^{-1}\left(\bm \Phi_{\tauad,+} \bm\Delta_\tauad\right)'(x\mid u)\right]_{21}\dd x=\\
(1+\Boh(\tau^{-\kappa})) \int_{-\tauad^\nu}^{\tauad^\nu}  \frac{\ee^{\hcc(x)}}{(1+\ee^{\hcc(x)})^2}\left[(\Dcc(x\mid u,\tad)^{-1}\Phiccp(x\mid u,\tad)^{-1}(\Phiccp \Dcc )'(x\mid u,\tad)\right]_{21} \dd x.
\end{multline*}
With the help of the calculation
$$
\frac{\ee^{\hcc(x)}}{(1+\ee^{\hcc(x)})^2}=\frac{1}{(1+\ee^{\hcc(x)})(1+\ee^{-\hcc(x)})}=\frac{\ee^{-\hcc(x)}}{(1+\ee^{-\hcc(x)})^2}=\frac{\ee^{\tad x-u}}{(1+\ee^{\tad x-u})^2}
$$
we recognize the integrand from \eqref{eq:EstPhiccinttailest}, and then conclude that
\begin{multline*}
\int_{-a\tau}^{b\tau}  \frac{\ee^{\msf h_\tauad(x\mid u)}}{\left(1+\ee^{\msf h_\tauad(x\mid u)}\right)^2}\left[ \bm \Delta_\tauad(x\mid u)^{-1}\bm \Phi_{\tauad,+}(x\mid u)^{-1}\left(\bm \Phi_{\tauad,+} \bm\Delta_\tauad\right)'(x\mid u)\right]_{21}\dd x =\\
(1+\Boh(\tau^{-\kappa})) \int_{-\infty}^{\infty}  \frac{\ee^{\tad x-u}}{(1+\ee^{\tad x-u})^2}\left[(\Dcc(x\mid u,\tad)^{-1}\Phiccp(x\mid u,\tad)^{-1}(\Phiccp \Dcc )'(x\mid u,\tad)\right]_{21} \dd x + \Boh(\ee^{-u-\eta\tau^{\nu}}).
\end{multline*}
Finally, with arguments very similar to the ones used in the proof of Lemma~\ref{eq:lemtailsintPhis}, we see that this remaining integral on the right-hand side is $\Boh(\ee^{-u})$. Everything combined, we just proved that
\begin{multline*}
\int_{-\tauad a}^{\tauad b}\frac{\ee^{\msf h_\tauad(x\mid u)}}{\left(1+\ee^{\msf h_\tauad(x\mid u)}\right)^2}\left[ \bm \Delta_\tauad(x\mid u)^{-1}\bm \Phi_{\tauad,+}(x\mid u)^{-1}\left(\bm \Phi_{\tauad,+} \bm\Delta_\tauad\right)'(x\mid u)\right]_{21}\dd x 
=\\
\int_{-\infty}^\infty \frac{\ee^{\tad x-u}}{(1+\ee^{\tad x-u})^2}\left[(\Phiccp(x\mid u,\tad)\Dcc(x\mid u,\tad))^{-1}(\Phiccp \Dcc )'(x\mid u,\tad)\right]_{21} \dd x
+\Boh(\tau^{-\kappa}\ee^{-u}),\quad \tau\to +\infty.
\end{multline*}
We now integrate in $u\in [-\sad,+\infty)$ and use Equation~\eqref{eq:QccIntRepPhicc} and the limit
\begin{equation*}
\lim_{s\to -\infty}\Qcc(s,T)=1,
\end{equation*}
which is valid by dominated convergence since $|1+\ee^\alpha|^{-1}\leq 1$ for every $\alpha\in \R$, to conclude the proof.
\end{proof}

\section{The underlying equilibrium measure and related quantities}\label{sec:eqmeasuretau}

As we mentioned earlier, one of the main objects we will analyze is the RHP for the orthogonal polynomials associated to \eqref{deff:Ln}--\eqref{deff:DeformedPartFunction}. In this analysis, the model problem we discussed in the last few sections plays a central role, and another important quantity is the associated equilibrium measure and related objects that we discuss in this section.

\subsection{The equilibrium measure}\label{sec:equilibriummeasure}\hfill 

A major role in our calculations is played by the equilibrium measure for the polynomial potential $V$, which is the unique probability measure $\mu_V$ on $\R$ for which the quantity
$$
\iint \log \frac{1}{|x-y|}\dd \mu(x) \dd\mu(y)+\int V(x)\dd \mu(x)
$$
attains its minimum over all Borel probability measures $\mu$ supported on $\R$. Its existence and uniqueness is assured by standard results, see for instance \cite{Saff_book}, and its regularity that we now discuss is of particular relevance. 

The measure $\mu_V$ is supported on a finite union of bounded intervals and is absolutely continuous with respect to the Lebesgue measure \cite{deift_kriecherbauer_mclaughlin}. Following Assumption~\ref{MainAssumption}--(i), we assume that $\mu_V$ is one-cut, that is, it is supported on a single interval that we take to be of the form
$$
\supp\mu_V=[-a,0],\quad a>0,
$$
and regular, meaning that the density of $\mu_V$ vanishes as a square-root at a neighborhood of the endpoints, does not vanish on $(-a,0)$, and the Euler-Lagrange equations are valid with strict inequality outside the support,
\begin{equation}\label{eq:eulerlagrange}
\int \log\frac{1}{|x-y|}\dd \mu_V(y)+\frac{1}{2} V(x)+\ell_V \;
\begin{cases}
> 0, & x\in \R\setminus \supp\mu_V, \\
=0, & x\in \supp\mu_V,
\end{cases}
\end{equation}
for some constant $\ell_V\in \R$. The notions just introduced are consistent with the notions and notations that we already introduced and used in Section~\ref{sec:StatementResults}. 

A transformation of the equilibrium measure of particular interest is its {\it Cauchy transform},
$$
C^{\mu_V}(z)\deff \int \frac{\dd \mu_V(x)}{x-z},\quad z\in \C\setminus \supp\mu_V.
$$
Using the Euler-Lagrange identity, it can be shown that $C^{\mu_V}$ satisfies an algebraic equation of the form
\begin{equation}\label{eq:spectralcurvehV}
\left(C^{\mu_V}(z)+\frac{V'(z)}{2}\right)^2=\frac{1}{4}z(z+a)h_V(z)^2,\quad z\in \C\setminus \supp\mu_V,
\end{equation}
for some polynomial $h_V$ which does not vanish on $[-a,0]$  \cite{kuijlaars_silva_s_curves, MartinezOriveRakhmanov2015}.

We also associate to the equilibrium measure its $\phi$ function
\begin{equation}\label{eq:functionphitau}
\phi(z)\deff \int_{0}^z\left(C^{\mu_V}(s)+\frac{1}{2} V'(s)\right)\dd s,\quad z\in \C\setminus (-\infty,0].
\end{equation}
The next result summarizes some properties of $\phi$ that will be needed later.

\begin{prop}\label{prop:functionphitau}
The function $\phi$ has the following properties.
\begin{enumerate}[(i)]
\item The function $\phi$ is analytic on $\C\setminus (-\infty,0]$.

\item For $x\in (-a,0)$, 
$$
\phi_{+}(x)+\phi_{-}(x)=0, \quad \text{and}\quad \phi_{+}(x)-\phi_{-}(x)=-2\pi\ii \mu_V((x,0)).
$$

\item For $x\in (-\infty,-a)$,
$$\phi_{+}(x)-\phi_{-}(x)=-2\pi \ii.$$

\item For $x\in \R\setminus [-a,0]$, 
$$\re \phi_{+}(x)=\re \phi_{-}(x) >0.$$

\item As $z\to \infty$ and some constant $\phi_\infty$,
$$
\phi(z)=\frac{V(z)}{2}+\ell_V-\log z+\frac{\phi_\infty}{z}+\Boh(z^{-1}),
$$
where $\ell_V$ is as in \eqref{eq:eulerlagrange}.
\item The function $\phi$ satisfies the estimate
\begin{equation}\label{eq:phi_local_beh}
\phi(z)= \frac{1}{3} h_V(0)a^{1/2} z^{3/2}(1+\Boh(z)), \quad z\to 0.
\end{equation}
\end{enumerate}
\end{prop}

\begin{proof}
The proof is standard using the properties of the equilibrium measure, see for instance \cite{deift_book}.
\end{proof}

\subsection{The conformal map}\hfill 

Finally, using $\phi$ we construct a conformal map $\psi$, introduced formally with the next result.

\begin{prop}\label{prop:conformalmappsi}
The function
$$
\psi(z)\deff \left(\frac{3}{2}\phi(z)\right)^{2/3},
$$
is a conformal map from a neighborhood of the origin to a disk $D_{2\delta}(0)$, with $U_0\deff \psi^{-1}(D_\delta(0))$, and admits an expansion of the form
\begin{equation}\label{eq:derivativepsicpsi}
\psi(z)=\msf c_V z(1+\Boh(z)),\quad z\to 0,\qquad \msf c_V\deff  2^{-2/3}h_V(0)^{2/3}a^{1/3}>0.
\end{equation}
\end{prop}
\begin{proof}
The proof is also standard, and follows essentially from Proposition~\ref{prop:functionphitau}-(iv). We omit the details.
\end{proof}

In the previous proposition, the factor $2\delta$ instead of $\delta$ is chosen just for later convenience, for the statement of Proposition~\ref{prop:ConformExtH}. Later on, we use $\psi$ only over the smaller neighborhood $D_\delta(0)$.

As it is customary in RHP analysis, at a later stage we will need to glue the model problem as a local parametrix for the original RHP for orthogonal polynomials. This gluing procedure is done, in our case, using the conformal map $\psi$. In usual situations, the jump matrices of the model local problem are piecewise constant or yet homogeneous, and as such this procedure of using the conformal map does not significantly alter them. However, in our situation the jump involves the function $Q$, and consequently the jump will be altered by the conformal map in a nontrivial way. 

With the next result we introduce the necessary quantities needed to keep track of this transformation. Recall the half rays $\cSig_j$, $j=0,1,2,3,$ which were introduced in \eqref{def:contour_sigma}. For the next statement, we talk about {\it neighborhoods} of $\cSig_j$, by which we mean open connected sets that contain $\cSig_j\setminus \{0\}$ in their interior.

\begin{prop}\label{prop:ConformExtH}
There exist neighborhoods $\mcal S_j$ of $\cSig_j$ and a function
$$
\msf H_Q:\mcal S\to \C,\quad \mcal S\deff \bigcup_{j=0}^3 \mcal S_j
$$
with the following properties.

\begin{enumerate}[(i)]
\item For the value $\delta>0$ in Proposition~\ref{prop:conformalmappsi}, the inclusions
$$
D_\delta(0)\subset \mcal S\quad \text{and}\quad \mcal S_j\cap \mcal S_k\subset D_\delta(0),
$$
hold true for any $j\neq k$.

\item The function $\msf H_Q$ is $C^\infty$ on $\mcal S$, and is an extension of $Q\circ \psi^{-1}$ from $D_\delta(0)$, that is,
\begin{equation}\label{eq:RelHQQpsi}
\msf H_Q(w)=Q(\psi^{-1}(w)),\quad |w|<\delta.
\end{equation}

\item The function $\msf H_Q$ is analytic on $D_\delta(0)$, extends continuously up to the boundary of $D_\delta(0)$ and satisfies
\begin{equation}\label{eq:expansionHQ}
\msf H_Q(w)=-\msf c_{\msf H}w+\Boh(w),\quad \msf c_{\msf H}\deff \frac{\tad}{\msf c_V},
\end{equation}
uniformly for $|w|\leq \delta$, where we recall that $\tad$ and $\msf c_V$ are as in \eqref{deff:tQprime} and \eqref{eq:derivativepsicpsi}, respectively.

\item For some constants $\wh\eta>\eta>0$, the function $\msf H_Q$ satisfies the estimates
$$
-\wh\eta |w|\leq \re \msf H_Q(w)\leq -\eta |w| \qquad \text{for every }w\in \mcal S_0\setminus D_\delta(0),
$$
and
$$
\re \msf H_Q(w)\geq \eta |w| \qquad \text{for every } w\in \left(\mcal S_1\cup\mcal S_2\cup \mcal S_3\right)\setminus D_\delta(0).
$$

\end{enumerate}
\end{prop}

\begin{proof}
We construct the set $\mcal S_j$ as tubular neighborhoods of $\cSig_j$ away from the origin, and as disks near the origin, namely
$$
\mcal S_j=\mcal S_j(\delta')=\left\{ w\in \C \mid \inf_{w'\in \cSig_j}|w-w'|< \delta' \right\}\cup D_\delta(0).
$$
By choosing $\delta'>0$ sufficiently small, in particular smaller than $\delta>0$, property (i) is immediate. 

The function
$$
\msf H_Q(w)\deff Q(\psi^{-1}(w)),\quad |w|<\delta,
$$
is obviously analytic on $D_{\delta}(0)$, satisfies (iii) and admits an extension to the larger open set $D_{2\delta}(0)$. A standard argument using partitions of unity allows us to extend it to the sets $\mcal S_j$'s as claimed by (ii), also making sure that it satisfies (iv).
\end{proof}

\section{Associated orthogonal polynomials}\label{sec:rhpapproachops}

The first and arguably major step towards understanding $\msf L^Q_n(\sad)$ is to study several quantities related to the orthogonal polynomials for the varying weight \eqref{def:perturbedweight}, as we introduce next.

\subsection{Orthogonal polynomials and related quantities}\hfill

Denote by $\msf P_k=\msf P_k^{(n,\sad)}$ the monic orthogonal polynomial of degree $k$ for the weight $\omega_n$ in \eqref{def:perturbedweight},
\begin{equation}\label{def:opsdeformed}
\msf P_k^{(n,\sad)}(x)=x^k+\text{(lower degree terms)},\qquad \int_{\R}\msf P_k^{(n,\sad)}(x)x^j\omega_n(x)\dd x=0, \quad j=0,\hdots, k-1.
\end{equation}
These polynomials depend on $\omega_n$, so ultimately also on $Q$, but we refrain from stressing this dependence in the notation. We also denote by $\msfga_k^{(n,Q)}=\msfga_k^{(n,Q)}(\sad)>0$ the corresponding norming constant, determined by
\begin{equation}\label{def:normingdeformed}
\frac{1}{\upgamma_k^{(n)}(\sad)^2}=\int_\R \msf P_k^{(n,\sad)}(x)^2\omega_n(x)\dd x.
\end{equation}

We associate to the orthogonal polynomials their Christoffel-Darboux kernel,
\begin{equation}\label{deff:Kndefweight}
\msf K^Q_n(x,y)=\msf K^Q_n(x,y\mid \sad)\deff \sum_{k=0}^{n-1} \msfga_k^{(n)}(\sad)^2\msf P_k^{(n,\sad)}(x)\msf P_k^{(n,\sad)}(y),
\end{equation}
stressing that we are not including the weight $\omega_n$ in this definition. 
In particular, the identity
\begin{equation}\label{eq:IntReprKernel}
\int_{-\infty}^\infty \msf K_n(x,x\mid \sad)\omega_n(x\mid \sad)\dd x=n
\end{equation}
holds true for any $\sad \in \R$ and follows immediately from \eqref{def:normingdeformed}.


In a similar manner, we introduce the related quantities for the undeformed weight $\ee^{-nV}$. The partition function $Z_n$ already appeared in \eqref{deff:nondeformedpartfunction}, the orthogonal polynomials $P_k=P_k^{(n)}$ are determined by
\begin{equation*}
P_k^{(n)}(x)=x^k+\text{(lower degree terms)},\qquad \int_{\R}P_k^{(n)}(x)x^j\ee^{-nV(x)}\dd x=0, \quad j=0,\hdots, k-1,
\end{equation*}
and the norming constants and Christoffel-Darboux kernel are determined from
$$
\frac{1}{(\gamma_k^{(n)})^2}=\int_\R P_k^{(n)}(x)^2\ee^{-nV(x)}\dd x,\quad K_n(x,y)\deff \sum_{k=0}^{n-1}(\gamma_k^{(n)})^2P_k^{(n)}(x)P_k^{(n)}(y).
$$

The orthogonal polynomials $\msf P_k^{(n,\sad)}$ vary continuously with $\sad$, which is a consequence of Heine's formula \cite[Equation~(3.10)]{deift_book}. In particular, when taking the limit $\sad\to+\infty$ we have that $x^k\omega_n(x)\to x^k\ee^{-nV}$ both uniformly in compacts and also in $L^1$, and $|x|^k\omega_n(x)\leq |x|^k\ee^{-nV(x)}$. Thus, dominated convergence then gives that all the just introduced undeformed quantities are recovered from their deformed versions in the limit $s\to+\infty$. This means, for instance, that the Christoffel-Darboux kernel $K_n$ and the partition function $Z_n$ are recovered via
\begin{equation}\label{eq:KernelContin}
K_n(x,y)=\msf K_n(x,y\mid \sad=+\infty)\quad \text{and}\quad Z_n=\msf Z_n(\sad=+\infty).
\end{equation}
The next result will be key into transforming asymptotics for the orthogonal polynomials to asymptotics for $\msf L^Q_n(s)$ itself.

\begin{prop}\label{prop:intreprLn}
The identity
\begin{equation}\label{eq:PartitionFctionNormingctt}
\log\msf L_n^Q(\sad) = -\int_\sad^{\infty}\int_{-\infty}^{\infty} \msf K^Q_n(x,x\mid u)\frac{\omega_n(x\mid u)}{1+\ee^{u+n^{2/3}Q(x)}}\dd x\; \dd u
\end{equation}
holds true for every $\sad\in \R$.
\end{prop}

\begin{remark}
While we were finishing this manuscript, the work \cite{ClaeysGlesner2021} was posted to the ArXiv. Therein, they also derive the formula \eqref{eq:PartitionFctionNormingctt} in more general terms, using the underlying RHP for IIKS-type integrable operators, see the first displayed formula in page 28 therein. Similar formulas play a fundamental role in the recent works \cite{CafassoClaeys2021,CafassoClaeysRuzza2021}, see for instance \eqref{eq:QccIntRepPhikpz} above. Our proof of \eqref{eq:PartitionFctionNormingctt} relies solely on orthogonality properties, so we decided to present it nevertheless.
\end{remark}

\begin{proof}
The equality
$$
\msf Z^Q_n(\sad)=n!\prod_{k=0}^{n-1}\msfga_k^{(n)}(\sad)^{-2}
$$
is standard in random matrix theory. From this identity, \eqref{def:normingdeformed} and the orthogonality relations we derive the deformation formula
\begin{equation*}
\partial_\sad \log \msf Z^Q_n(\sad)  =-\int_\R \partial_\sad\msf K^Q_n(x,x\mid \sad)\omega_n(x\mid \sad)\dd x,
\end{equation*}
which in fact is valid for general weights depending on an additional parameter $\sad$. To our knowledge, this last identity was first observed by Krasovsky \cite[Equation~(14)]{Krasovsky07}. We fix constants $L>\sad>0$ and integrate the identity above, 
\begin{equation}\label{eq:propLndef1}
\log\msf Z^Q_n(\sad) =\log\msf Z^Q_n(L)+\int_{\sad}^L \int_{-\infty}^\infty \msf \partial_\sad \msf K^Q_n(x,x\mid u)\omega_n(x\mid u)\dd x\dd u.
\end{equation}
We want to interchange the order of integration in the above. The derivative $\partial_\sad\msf K^Q_n$ is a polynomial of degree at most $2n-2$ in $x$, and from Heine's formula for orthogonal polynomials we see that the polynomial coefficients of $\msf K_n$ are continuous functions of $\sad$. Therefore, for given $\sad,L$ there exists a constant $M=M(\sad,L)>0$ for which the pointwise bound
$$
\left|\partial_\sad \msf K^Q_n(x,x\mid u)\right| \leq M \sup_{0\leq k\leq 2n-2}|x|^k,\quad x\in \R,
$$
is valid for every $u\in [\sad,L]$. Together with the inequalities $0\leq \omega_n(x)\leq \ee^{-nV(x)}$, this bound ensures that we can interchange order of integration in \eqref{eq:propLndef1}. After integration by parts, we then obtain
\begin{multline*}
\log\msf Z^Q_n(\sad) =\log\msf Z^Q_n(L)+\int_{-\infty}^\infty \msf K^Q_n(x,x\mid L)\omega_n(x\mid L)\dd x-\int_{-\infty}^\infty \msf K^Q_n(x,x\mid \sad)\omega_n(x\mid \sad)\dd x \\ 
-\int_{-\infty}^{\infty} \int_\sad^L \msf K^Q_n(x,x\mid u)\frac{\omega_n(x\mid u)}{1+\ee^{u+n^{2/3}Q(x)}}\dd u
\end{multline*}
From the identity \eqref{eq:IntReprKernel} the two single integrals cancel one another. The integrand of the double integral is positive, so by Tonelli's Theorem we can interchange order of integration. After this interchange, we take the limit $L\to +\infty$ and use \eqref{eq:KernelContin} and \eqref{deff:Ln} to conclude the proof.
\end{proof}

\subsection{The Riemann-Hilbert Problem for orthogonal polynomials}\hfill 

We are ready to introduce the RHP for orthogonal polynomials for the weight $\omega_n$ in \eqref{def:perturbedweight}. During this section, we keep using the matrix notation that was already used in previous sections, recall for instance \eqref{deff:matrixnot1},\eqref{deff:matrixnot2} and \eqref{deff:matrixnot3}.

The RHP for orthogonal polynomials for the weight \eqref{def:perturbedweight} asks for finding a $2\times 2$ matrix-valued function $\bm Y$ with the following properties.
\begin{enumerate}[\bf Y-1.]
\item The matrix $\bm Y:\C\setminus \R\to \C^{2\times 2}$ is analytic.
\item The function $\bm Y$ has continuous boundary values
$$
\bm Y_\pm(x)\deff \lim_{\varepsilon\searrow 0}\bm Y(x\pm \ii \varepsilon),\quad x\in \R,
$$
which are related by the jump condition $\bm Y_+(x)=\bm Y_-(x)\bm J_{\bm Y}(x)$, $x\in \R$, with
$$
\bm J_{\bm Y}(x)\deff \bm I+\omega_n(x)\bm E_{12}.
$$
\item As $z\to \infty$,
$$
\bm Y(z)=\left(\bm I+\Boh(z^{-1})\right)z^{n\sp_3}.
$$
\end{enumerate}

Observe that $\bm Y=\bm Y^{(n)}(\cdot \mid \sad,Q)$ depends on the index $n$ and also on $\sad$ and $Q$, although we do not make this dependence explicit in our notation. As shown by Fokas, Its and Kitaev \cite{FokasItsKitaev92}, for each $n$ the RHP above has a unique solution, which is explicitly given by
$$
\bm Y(z)=
\begin{pmatrix}
\msf P^{(n,\sad)}_{n}(z) &\ds{ \frac{1}{2\pi \ii}\int_{\R} \frac{P_n^{(n,\sad)}(x)}{x-z}\omega_n(x)\dd x }\\[10pt]
-2\pi \ii \msfga^{(n,Q)}_{n-1}(\sad)^2 \msf P^{(n,\sad)}_{n-1}(z) & \ds{-\msfga^{(n,Q)}_{n-1}(\sad)^2\dfrac{1}{2\pi \ii}\int_{\R} \frac{\msf P^{(n,\sad)}_{n-1}(x)}{x-z}\omega_n(x)\dd x}
\end{pmatrix},
$$
where $\msfga_k^{(n,Q)}(\sad)$ and $\msf P_k^{(n,\sad)}$ are as in \eqref{def:opsdeformed} and \eqref{def:normingdeformed}.

In particular, from this identity we obtain the relation
\begin{equation}\label{eq:relationnormingY1}
\msfga_{n-1}^{(n,Q)}(\sad)^2=-\frac{1}{2\pi \ii}\left(\bm Y^{(n,1)}\right)_{21},
\end{equation}
where $\bm Y^{(n,1)}=\bm Y^{(n,1)}(\sad,Q)$ is the matrix determined from the more detailed expansion
\begin{equation}\label{eq:ExpYRHPInfinity}
\bm Y(z)=\bm Y^{(n)}(z)=\left(\bm I+\frac{1}{z}\bm Y^{(n,1)}+\frac{1}{z^2}\bm Y^{(n,2)}+\Boh(z^{-3})\right)z^{n\sp_3},\quad z\to \infty.
\end{equation}
Also, the Christoffel-Darboux kernel \eqref{deff:Kndefweight} can be recast directly from $\bm Y$ from the identity
\begin{equation}\label{eq:KnYxy}
\msf K^Q_n(x,y\mid \sad)=\frac{1}{2\pi\ii}\frac{1}{x-y}\bm e_2^\tp \bm Y_+(y)^{-1}\bm Y_+(x) \bm e_1,\quad x,y\in \R, \; x\neq y.
\end{equation}
In the confluent limit $x=y$, this formula yields
\begin{equation}\label{eq:KnYxx}
\msf K^Q_n(x,x\mid \sad)=\frac{1}{2\pi\ii}\bm e_2^\tp \bm Y_+(x)^{-1}\bm Y'_+(x) \bm e_1,\quad x\in \R.
\end{equation}
%

The remainder of this paper is dedicated to applying the Deift-Zhou method for this RHP and collecting its consequences, analysis which will ultimately lead to the proofs of our main results.

\section{The RHP analysis for the orthogonal polynomials}\label{sec:rhpanalysis}

With all the preliminary work completed, we are finally at the stage of performing the asymptotic analysis for the {\bf RHP-$\bm Y$} for orthogonal polynomials that was introduced in Section~\ref{sec:rhpapproachops}. Most of the transformations are standard, so we go over them quickly and without much detail. Care will be taken in the construction of the parametrices, which are the steps where the introduction of the factor $\sigma_n$ plays a major role. We also remind the reader that the function $V$ and $Q$ are always assumed to satisfy Assumptions~\ref{MainAssumption}.

The function $\sigma_n=\sigma_n(z\mid \sad,Q)$ depends on $\sad\in \R$ and $Q$, and as such in all the steps below several quantities will also depend on these parameters. Nevertheless, in most of the work that follows the parameter $\sad$ and the function $Q$ do not play a major role so we omit them in our notations unless when needed to avoid confusion.

\subsection{First transformation: normalization at infinity}\hfill

Recall the function $\phi$ introduced in \eqref{eq:functionphitau}. The first transformation, which has the effect of normalizing the RHP as $z\to \infty$, takes the form
\begin{equation}\label{eq:transfYT}
\bm T(z)\deff \ee^{-n\ell_V \sp_3} \bm Y(z)\ee^{n\left(\phi(z)-\frac{1}{2}V(z)\right)\sp_3},\quad z\in \C\setminus \R.
\end{equation}

From the RHP for $\bm Y$ and the properties from Proposition~\ref{prop:functionphitau}, we obtain that $\bm T$ satisfies the following RHP.

\begin{enumerate}[\bf T-1.]
\item The matrix $\bm T:\C\setminus \R\to \C^{2\times 2}$ is analytic.
\item For $z\in \R$, it satisfies the jump $\bm T_+(z)=\bm T_-(z)\bm J_{\bm T}(z)$, with
$$
\bm J_{\bm T}(z)
\deff
\begin{pmatrix}
\ee^{n(\phi_{+}(z)-\phi_{-}(z))} & \sigma_{n}(z)\ee^{-n(\phi_{+}(z)+\phi_{-}(z))} \\
0 & \ee^{-n(\phi_{+}(z)-\phi_{-}(z))}
\end{pmatrix},\quad z\in \R.
$$
\item As $z\to \infty$,
$$
\bm T(z)=\bm I+\frac{1}{z}\bm T_1+\Boh(z^{-2}),
$$
where the coefficient $\bm T_1$ is 
$$
\bm T_1\deff \ee^{-n\ell_V \sp_3}\bm Y^{(n,1)}\ee^{n\ell_V \sp_3}+\phi_\infty \sp_3,
$$
and we recall that $\bm Y^{(n,1)}$ and $\phi_\infty$ were introduced in \eqref{eq:ExpYRHPInfinity} and in Proposition~\ref{prop:functionphitau}--(v), respectively.

\end{enumerate}

From the properties (ii) of Proposition~\ref{prop:functionphitau}, the jump matrix for $\bm T$ simplifies in convenient ways.
For $-a<z<0$,
\begin{align*}
\bm J_\bm T(z) & =
\begin{pmatrix}
\ee^{2n\phi_{+}(z)} & \sigma_n(z) \\ 
0 & \ee^{-2n\phi_{+}(z)}
\end{pmatrix} \\
 & =\left(\bm I+\frac{1}{\sigma_n(z)} \ee^{-2n\phi_{+}(z)}\bm E_{21}\right)\left(\sigma_n(z)\bm E_{12}-\frac{1}{\sigma_n(z)}\bm E_{21}\right)\left(\bm I+\frac{1}{\sigma_n(z)} \ee^{2n\phi_{+}(z)}\bm E_{21}\right) \\
& =\left(\bm I+\frac{1}{\sigma_n(z)}\ee^{2n\phi_{-}(z)}\bm E_{21}\right)\left(\sigma_n(z)\bm E_{12}-\frac{1}{\sigma_n(z)}\bm E_{21}\right)\left(\bm I+\frac{1}{\sigma_n(z)} \ee^{2n\phi_{+}(z)}\bm E_{21}\right),
\end{align*}
and for $z\in \R\setminus [-a,0]$,
\begin{align*}
\bm J_\bm T(z)=\bm I+\sigma_n(z) \ee^{-2n\phi_{+}(z)}\bm E_{12}.
\end{align*}

\subsection{Second transformation: opening of lenses}\hfill

From the identities just written for $\bm J_{\bm T}$ and Proposition~\ref{prop:functionphitau}-(ii), it follows that the diagonal entries of $\bm J_{\bm T}$ are highly oscillatory on $(-a,0)$ as $n\to\infty$. In the second transformation of the RHP we perform the so-called opening of lenses, which has the effect of moving this oscillatory behavior to a region where it becomes exponentially decaying. 

Define regions $\mcal G^\pm$ on the $\pm$-side of $(-a,0)$ (the lenses, see Figure~\ref{fig:openinglenses}), assuming in addition that for $U_0$ as in Proposition~\ref{prop:conformalmappsi} these regions satisfy
\begin{equation}\label{eq:transformationlenses}
\psi(\partial \mcal G^\pm \cap U_0 )\subset (0,\ee^{\pm 2\pi \ii /3}\infty)\cup (0,\infty),
\end{equation}
which can always be achieved because $\psi:$ is conformal from a neighborhood of $U_0$ to a neighborhood of $D_\delta(0)$.

The function $\sigma_n$ has no zeros and may have singularities, but these are all poles due to the analyticity of $Q$ in a neighborhood of the real axis. Therefore, the fraction $1/\sigma_n$ is analytic on a neighborhood of the real axis. We use this fraction to transform
$$
\bm S(z)\deff
\begin{dcases}
\bm T(z)\left(\bm I\mp \frac{1}{\sigma_n(z)} \ee^{2n\phi(z)}\bm E_{21}\right),& z\in \mcal G^\pm, \\
\bm T(z), & \text{elsewhere}.
\end{dcases}
$$
\begin{figure}[t]
		\centering
		\begin{tikzpicture}[scale=1]
%

\node at (0,0) (a) {};
\node at (5,0) (b) {};
\node at (6,0) (b1) {};
\fill (a) circle[radius=2.5pt] node [above left] {$-a$};
\fill (b1) circle[radius=2.5pt] node [above right] {$0$};

\draw [thick, postaction={vmid arrow={black,scale=1.5}{.7}}] (a.center) to (b.center);
\draw [thick, postaction={vmid arrow={black,scale=1.5}{.6}}] (b.center) to (b1.center);
\draw [thick, postaction={vmid arrow={black,scale=1.5}{.5}}] ($(a)+(-2,0)$) to (a.center);
\draw [thick, postaction={vmid arrow={black,scale=1.5}{.5}}] (b1.center) to ($(b1)+(2,0)$);
\draw [thick, postaction={vmid arrow={black,scale=1.5}{.6}}] (a.center) to [out=45,in=180-45,edge node={node [pos=0.5,below,yshift=-8pt,xshift=-5pt] {$\mcal G^+$}}]  (b1.center);
\draw [thick, postaction={vmid arrow={black,scale=1.5}{.6}}] (a.center) to [out=-45,in=180+45,edge node={node [pos=0.5,above,yshift=8pt,xshift=-5pt] {$\mcal G^-$}}]  (b1.center);
		\end{tikzpicture}
	\caption{The regions used for the opening of lenses in the transformation $\bm T\mapsto \bm S$.}\label{fig:openinglenses}
\end{figure}
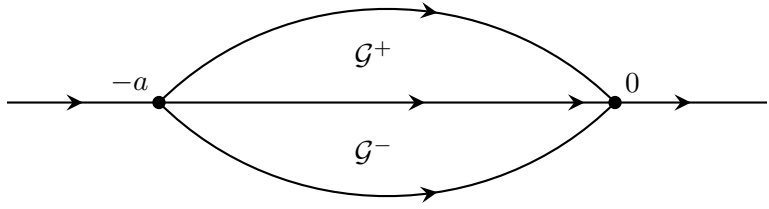

With
$$
\Gamma_{\bm S}\deff\R\cup \partial \mcal G^+\cup \partial \mcal G^-,
$$
and using the jump properties of $\phi$ listed in Proposition~\ref{prop:functionphitau}, the matrix $\bm S$ satisfies the following RHP.
\begin{enumerate}[\bf S-1.]
\item The matrix $\bm S:\C\setminus \Gamma_{\bm S}\to \C^{2\times 2}$ is analytic.
\item For $z\in \Gamma_{\bm S}$, it satisfies the jump $\bm S_+(z)=\bm S_-(z)\bm J_{\bm S}(z)$, with
$$
\bm J_{\bm S}(z)\deff
\begin{dcases}
\sigma_n(z)\bm E_{12}-\frac{1}{\sigma_n(z)}\bm E_{21}, & -a<z<0, \\
\bm I+ \frac{1}{\sigma_n(z)}\ee^{2n\phi(z)}\bm E_{21}, & z\in \partial \mcal G^\pm\setminus \R, \\
\bm I+\sigma_n(z)\ee^{-2n\phi_{+}(z)}\bm E_{12}, & z\in \R\setminus (-a,0).
\end{dcases}
$$
\item As $z\to \infty$,
$$
\bm S(z)=\bm I+\frac{\bm S_1}{z}+\Boh(z^{-2}),\qquad \text{with}\quad \bm S_1\deff \bm T_1.
$$
\item The matrix $\bm S$ remains bounded near the points $z=-a,0$.
\end{enumerate}

Before moving to the construction of the mentioned parametrices, we conclude this section with the needed estimate for the jump matrix $\bm J_{\bm S}$ away from $[-a,0]$. For that, recall the matrix norm notation introduced in \eqref{deff:matrixnorm},\eqref{deff:matrixLpnorm},\eqref{deff:matrixLpqnorm}. 

\begin{prop}\label{prop:decayJS}
For $U_0$ as in Proposition~\ref{prop:conformalmappsi}, introduce the set
$$
\Gamma_\varepsilon\deff \Gamma_{\bm S}\setminus \left([-a,0]\cup U_0\cup D_\varepsilon(-a)\right)
$$
For some $\varepsilon>0$, and possibly reducing $U_0$ if necessary, there is an $\eta>0$ such that
$$
\|\bm J_{\bm S}-\bm I\|_{L^1\cap L^2\cap L^\infty(\Gamma_\varepsilon)}=\Boh(\ee^{-\eta n}),
$$
as $n\to \infty$.
\end{prop}
\begin{proof}
From Proposition~\ref{prop:functionphitau}-(iv) we obtain that for any $\varepsilon>0$ there is a constant $\eta'>0$ for which
$$
\re \phi_{+}(x)\geq \eta',\quad \text{for every } x\in (-a-\varepsilon,\varepsilon).
$$
On the other hand, the jump conditions on Proposition~\ref{prop:functionphitau}-(ii) combined with Cauchy-Riemann equations imply in a standard way that $\re \phi\leq - \eta$ along the lipses of the lenses and away from the endpoints $-a$ and $0$, as long as the lens stay within a positive but small distance from the interval $[-a,0]$. From these pointwise estimates, the growth of $\re \phi$ as $z\to \pm\infty$ and the fact that $\sigma_n$ remains bounded on $\R$ and $1/\sigma_n$ grows at most with $\Boh(\ee^{c n^{2/3}})$ on compacts of $\C$, the claimed $L^p$ estimates follow in a standard manner.
\end{proof}

\subsection{Global parametrix}\hfill

The global parametrix problem, obtained after neglecting the jumps of $\bm S$ that are exponentially close to the identity, is the following RHP.

\begin{enumerate}[\bf G-1.]
\item $\bm G:\C\setminus [-a,0] \to \C^{2\times 2}$ is analytic.
\item For $z\in (-a,0)$, it satisfies the jump 
$$
\bm G_+(z)=\bm G_-(z)\left(\sigma_n(z)\bm E_{12}-\frac{1}{\sigma_n(z)}\bm E_{21}\right).
$$
\item As $z\to \infty$,
$$
\bm G(z)=\bm I+\Boh(z^{-1}).
$$

\item $\bm G$ has square-integrable singularities at $z=-a,0$.
\end{enumerate}

The construction of the global parametrix follows standard techniques. First, one introduces a function that we denote $\msf q(z)$, with the aim at transforming the RHP for $\bm G$ to a RHP with constant jumps. Then, by diagonalizing the resulting jump matrix, we further reduce the problem to two scalar-valued RHPs. With the help of Plemelj's formula, we then solve these scalar RHPs, and by tracing back all the transformations we recover the matrix $\bm G$ itself.

The procedure just described is standard in RHP literature, see for instance \cite[Appendix~A.1]{BKMMbook}, so we refrain from completing it in detail and instead only describe the final form of the solution.

The function $\sigma_n$ does not vanish and is real and positive over the real axis, so its real logarithm over the real axis is well defined. With this in mind, introduce
$$
\msf q(z)\deff \frac{((z+a)z)^{1/2}}{2\pi }\int_{-a}^{0} \frac{\log \sigma_n(x)}{\sqrt{|x|(x+a)}} \frac{\dd x}{x-z},\quad z\in \C\setminus [-a,0],
$$
where $(\cdot)^{1/2}$ stands for the principal branch of the square root and $\sqrt{\cdot}$ is reserved for the standard positive real root of positive real numbers.
This function $\msf q$ depends on $n$ but we do not make this dependence explicit for ease of notation. It is analytic on $\C\setminus [-a,0]$, and it is chosen to satisfy the jump condition
$$
\msf q_+(x)+\msf q_-(x)=-\log\sigma_n(x),\quad -a<x<0.
$$
Furthermore, standard calculations show that
$$
\msf q(z)=\Boh(1),\quad z\to -a,0,\qquad \text{and} \qquad \msf q(z)=\msf q_0+\frac{\msf q_1}{z} + \Boh(z^{-2}), \quad z\to\infty,
$$
with coefficients given by
%
%
\begin{equation}\label{eq:def_q0q1}
\msf q_0=\msf q_0(n)\deff-\frac{1}{2\pi}\int_{-a}^{0} \frac{\log\sigma_n(x)}{\sqrt{|x|(x+a)}}\dd x \quad \text{and}\quad 
\msf q_1=\msf q_1(n)\deff-\frac{1}{2\pi}\int_{-a}^{0} \frac{x\log\sigma_n(x)}{\sqrt{|x|(x+a)}}\dd x +\frac{a\msf q_0}{2}.
\end{equation}
Next, set
\begin{equation}\label{eq:def_U0}
\bm U_0\deff\frac{1}{\sqrt{2}}
\begin{pmatrix}
1 & \ii \\ \ii & 1
\end{pmatrix},
\quad 
\msf m(z)\deff \frac{z}{z+a},
\end{equation}
which is consistent with \eqref{def:matrixU0modelprobl}, and introduce
\begin{equation}\label{eq:def_globalparM}
\bm M(z)\deff
\bm U_0
\msf m(z)^{\sp_3/4}
\bm U_0^{-1}
=
\frac{1}{2}
\begin{pmatrix}
\msf m(z)+\dfrac{1}{\msf m(z)} & -\ii \left(\msf m(z)-\dfrac{1}{\msf m(z)}\right) \\
\ii \left(\msf m(z)-\dfrac{1}{\msf m(z)}\right) & \msf m(z)+\dfrac{1}{\msf m(z)}
\end{pmatrix}
\end{equation}

This matrix $\bm M$ satisfies
$$
\bm M_+(z)=\bm M_-(z)\left(\bm E_{12}-\bm E_{21}\right), \; -a<z<0, \quad \text{and}\quad \bm M(z)=\bm I-\frac{a}{4z}\sp_2+\Boh(z^{-2}),
$$
where $\sp_2$ is the second Pauli matrix (recall \eqref{deff:matrixnot2}).

Then the solution to the global parametrix {\bf RHP-}{$\bm G$} is
\begin{equation}\label{eq:constructionGparam}
\bm G(z) = \ee^{-\msf q_0\sp_3}\bm M(z) \ee^{\msf q(z)\sp_3},\quad z\in \C\setminus [-a,0].
\end{equation}
This solution $\bm G$ satisfies
\begin{equation}\label{eq:definitionG1}
\bm G(z)=\msf I+\frac{\bm G_1}{z}+\Boh(z^{-2}),\; z\to \infty,\quad \text{with}\quad \bm G_1\deff
\begin{pmatrix}
\msf q_1 & \dfrac{\ii a}{4}\ee^{-2\msf q_0} \\ -\dfrac{\ii a }{4}\ee^{2\msf q_0} & -\msf q_1
\end{pmatrix}.
\end{equation}

Recall that $U_0$ denotes the neighborhood of the origin given in Proposition~\ref{prop:conformalmappsi}.  We will also need some control on $\msf q$ inside $U_0$.

For the next result, set
\begin{equation}\label{deff:msfqo1}
\begin{aligned}
&  F_{\beta}(\sad)\deff \int_0^\infty v^{\beta}\log(1+\ee^{-\msf s-v})\dd v, \\
&\msf q_0^{(1)}=\msf q_0^{(1)}(\sad,\tad)\deff  \frac{\tad^{1/2}}{2\pi a^{1/2}}F_{-1/2}(\sad), \\
& \msf q^{(1)}(z)=\msf q^{(1)}(z\mid \sad,\tad)\deff \frac{\msf t^{1/2}}{2\pi a^{1/2}\msf m(z)^{1/2}}F_{-1/2}(\sad),
\end{aligned}
\end{equation}
which are $n$-independent quantities. The index $\beta$ does not have any specific meaning for what comes later, but it arises naturally from the asymptotic analysis resulting in the following result.
\begin{lemma}\label{lem:estimate_q_q0}
For any fixed $\sad_0>0$, the estimate
$$
\msf q_0=\frac{1}{n^{1/3}}\msf q_0^{(1)}+\Boh(n^{-2/3}),\quad n\to \infty,
$$
is valid uniformly for $\sad \geq -\sad_0$. In addition, the estimates
$$
\msf q(z)=\frac{1}{n^{1/3}}\msf q^{(1)}(z)+\Boh(n^{-2/3}),\quad \text{and}\quad \msf q'(z)=\Boh(n^{-1/3}),
$$
are valid uniformly for $z$ on compacts of $\C\setminus [-a,0]$ (in particular on $\partial U_0$) and uniformly for $\sad \geq -\sad_0$, and carry through to boundary values $\msf q_\pm(x)$ for $x$ along $\R\setminus \{-a,0\}$.

Finally,
$$\bm M(z)=\left(\bm I+\frac{1}{n^{1/3}}\left( \msf q^{(1)}_0\sp_3 -\msf q^{(1)}(z) \bm M(z)\sp_3\bm M(z)^{-1}\right)+ \Boh(n^{-2/3})  \right)\bm G(z),\quad n\to\infty,$$
uniformly for $z\in \partial U_0$ and $s\geq -\sad_0$.
\end{lemma}

\begin{proof}
The estimate for $\msf q_0(n)$ follows immediately from an application of Proposition~\ref{prop:LaplaceLogIntApp}. The estimates for $\msf q(z)$ and $\msf q'(z)$ also follow from Proposition~\ref{prop:LaplaceLogIntApp}, once we observe that the integrals defining them can be slightly deformed to the upper/lower half plane in a neighborhood of the unique point in the intersection $\partial U_0\cap (-a,0)$.

Finally, using the first part of the statement and the fact that $\bm M$ is bounded for $z\in\partial U_0$ and independent of $n$, we expand the exponentials in series and write
\begin{align*}
\bm M(z)\bm G(z)^{-1} & =\bm M(z)\ee^{-\msf q(z)\sp_3}\bm M(z)^{-1}\ee^{\msf q_0\sp_3}\\ 
& =\left(\bm I-\msf q(z)\bm M(z)\sp_3\bm M(z)^{-1}+\Boh(n^{-2/3})\right)\left(\bm I+\msf q_0\sp_3+\Boh(n^{-2/3})\right),
\end{align*}
and the last claim follows after rearranging the terms in this expansion.
\end{proof}

\subsection{Local Parametrix near $-a$}\hfill

The local parametrix $\bm P=\bm P^{(a)}$ near $z=-a$ is constructed in a neighborhood of $z=-a$ which without loss of generality can be taken to be the disk $D_\delta(-a)$ of radius $\delta$ around $a$, and it is the solution to the following RHP.
\begin{enumerate}[\bf $\bm P^{({a})}$-1.]
\item The matrix $\bm P^{(a)}:D_\delta(-a)\setminus \Gamma_{\bm S}\to \C^{2\times 2}$ is analytic.
\item For $z\in \Gamma_{\bm S}\cap \partial D_\delta(-a)$, it satisfies the jump $\bm P^{(a)}_+(z)=\bm P^{(a)}_-(z)\bm J_{\bm S}(z)$.
\item Uniformly for $z\in \partial D_\delta(-a)$,
$$
\bm P^{(a)}(z)=\left(\bm I+\boh(1)\right)\bm G(z),\quad n\to\infty.
$$
\item The matrix $\bm P^{(a)}$ remains bounded as $z\to -a$.
\end{enumerate}

The asymptotic condition {\bf $\bm P^{({a})}$-3.} above will be improved to \eqref{eq:decayparametrixa} below.

From the conditions on $Q$, we know that there exists a value $\eta>0$ for which
$$
\re Q(z)\geq 2\eta,\quad |z+a|<\delta.
$$
This value is uniform for $\tad\in [\tad_0,1/\tad_0],$ for any $\tad_0\in (0,1)$ fixed, and it is independent of $\sad\in \R$.
In particular, once we fix $\sad_0>0$ and assume that $\sad\geq -\sad_0$, from this inequality we obtain
$$
|\ee^{-\sad-n^{2/3}Q(z)}|\leq \ee^{-n^{2/3}\eta},\quad |z+a|<\delta,\quad \text{for large enough } n.
$$
This way, for $n>0$ sufficiently large the function $\sigma_n$ admits an analytic continuation to the whole disk $D_\delta(-a)$, and this continuation does not have zeros on the same disk. Thus, a branch of $\log\sigma_n$ is well defined in a neighborhood of $z=-a$, and the just mentioned estimate also shows that
\begin{equation}\label{eq:EstLogSigman}
\log\sigma_n(z)=\Boh(\ee^{-\eta n^{2/3}}),\quad n\to\infty,
\end{equation}
uniformly for $z$ in a neighborhood of $z=-a$ and $\sad\geq -\sad_0$.

With this in mind, the parametrix $\bm P^{(a)}$ can be constructed explicitly out of Airy functions in a standard way, see for instance \cite[Section~7.6]{deift_book}. Since it involves a somewhat nonstandard matching analytic prefactor that accounts for $\sigma_n$, we briefly go over this construction. 

Recall the contour $\cSig$ introduced in \eqref{def:contour_sigma}. With appropriate Airy functions, we construct a $2\times 2$ matrix $\bm\Psi_{\ai}$, which is analytic on $\C\setminus \cSig$ and satisfies
$$
\bm \Psi_{\ai,+}(\zeta)=\bm \Psi_{\ai,-}(\zeta)\times 
\begin{cases}
\bm I-\bm E_{12}, & \zeta \in \cSig_0, \\
\bm I-\bm E_{21}, & \zeta \in \cSig_1\cup\cSig_3, \\
-\bm E_{12}+\bm E_{21}, & \zeta \in \cSig_2, \\
\end{cases}
$$
and
\begin{equation}\label{eq:AsympAiryPar}
\bm\Psi_{\ai}(\zeta)=\zeta^{\sp_3/4}\bm U_0\left(\bm I+\Boh(\zeta^{-3/2})\right)\ee^{-\frac{2}{3}\zeta^{3/2}\sp_3},\quad \zeta\to \infty.
\end{equation}
In fact, $\bm\Psi_{\ai}$ can be obtained with a modification of the matrix $\Phiai$ which we previously used in \eqref{deff:Phiai}. We will not need its explicit form, so we do not write it down explicitly.

Using the properties of $\phi$ we construct a conformal map $\varphi$ from a neighborhood of $-a$ to a neighborhood of the origin, with
$$
\varphi(-a)=0, \quad \varphi'(-a)<0\quad \text{and}\quad \frac{2}{3}\varphi(z)^{3/2}=\phi(z)+2\pi \ii \Z, \; z\in D_\delta(-a)\setminus \cSig_{\bm S}.
$$
With standard arguments (see for instance the proof of Proposition~\ref{prop:estimateL} below for similar arguments), one shows that the matrix
\begin{equation}\label{eq:FFactorRegPar}
\bm F^{(a)}(z)\deff \bm G(z)\ee^{\frac{1}{2}\log\sigma_n(z)\sp_3}\bm U_0^{-1}(n^{2/3}\varphi(z))^{-\sp_3/4}
\end{equation}
is analytic on a neighborhood of $z=-a$. The local parametrix then takes the form
\begin{equation}\label{deff:LocalParRegEdge}
\bm P^{(a)}(z)=\bm F^{(a)}(z)\bm \Psi_{\ai}(n^{2/3}\varphi(z))\ee^{-\frac{1}{2}\log\sigma_n(z)\sp_3}\ee^{n\phi(z)\sp_3},\quad z\in D_\delta(-a)\setminus \Gamma_{\bm S}.
\end{equation}
As a result, the error term in fact takes on the stronger form
\begin{equation}\label{eq:decayparametrixa}
\bm P^{(a)}(z)=\bm G(z)\left(\bm I+\Boh(n^{-1})\right),\quad n\to\infty,
\end{equation}
which is valid uniformly for $z\in \partial D_\delta(-a)$ and uniformly for $\sad\geq -\sad_0$ and $\tad\in [\tad_0,1/\tad_0]$, for any $\sad_0>0$ and $\tad_0\in (0,1)$ fixed. 

\subsection{Local Parametrix near the origin}\hfill

The local parametrix near the origin requires the model problem from Section~\ref{sec:modelproblem}.

Recall the neighborhood $U_0$ of the origin introduced in Proposition~\ref{prop:conformalmappsi}. The initial local parametrix we seek for should be the solution to the following RHP.
\begin{enumerate}[\bf $\bm P^{(0)}$-1.]
\item The matrix $\bm P^{(0)}:U_0\setminus \Gamma_{\bm S}\to \C^{2\times 2}$ is analytic.
\item For $z\in \Gamma_{\bm S}\cap U_0$, it satisfies the jump $\bm P^{(0)}_+(z)=\bm P^{(0)}_-(z)\bm J_{\bm S}(z)$.
\item Uniformly for $z\in \partial U_0$,
$$
\bm P^{(0)}(z)=\left(\bm I+\boh(1)\right)\bm G(z),\quad n\to\infty.
$$
\item $\bm P^{(0)}$ remains bounded as $z\to 0$.
\end{enumerate}

To construct the solution $\bm P^{(0)}$ required above, some work is needed. Aiming at removing $\phi$ from the jump of $\bm P^{(0)}$, we change this {\bf RHP-$\bm P$} with the transformation
\begin{equation}\label{local_param_LP_rel}
\bm L(z)=\bm P^{(0)}(z)\ee^{-n\phi(z)\sp_3 },\quad z\in U_0\setminus \Gamma_{\bm S}.
\end{equation}
Then the matrix $\bm L$, should it exist, must satisfy the following RHP.
\begin{enumerate}[\bf $\bm L$-1.]
\item The matrix $\bm L:U_0\setminus \Gamma_{\bm S}\to \C^{2\times 2}$ is analytic.
\item For $z\in \Gamma_{\bm S}\cap U_0$, it satisfies the jump $\bm L_+(z)=\bm L_-(z)\bm J_{\bm L}(z)$,
with
\begin{equation}\label{eq:jumpJL}
\bm J_{\bm L}(z)\deff
\begin{dcases}
\sigma_n(z)\bm E_{12}-\frac{1}{\sigma_n(z)}\bm E_{21}, & z\in U_0\cap (-a,0), \\
\bm I+\frac{1}{\sigma_n(z)}\bm E_{21}, & z\in \partial\mcal G^\pm\cap U_0,\\
\bm I+\sigma_n(z)\bm E_{12}, & z\in U_0\cap (0,\infty).
\end{dcases}
\end{equation}
\item Uniformly for $z\in \partial U_0$,
$$
\bm L(z)=\left(\bm I+\boh(1)\right)\bm G(z)\ee^{-n\phi(z)\sp_3 },\quad n\to\infty.
$$
\item The matrix $\bm L$ remains bounded as $z\to 0$.
\end{enumerate}

Based on the usual way of matching the local parametrix with a model problem, one is tempted to moving the non-constant part of the jump - namely $\sigma_n$ - to the behavior at $\partial U_0$ as well. This would be done so including a term of the form $\sigma_n^{\sp_3/2}=\ee^{\sp_3\log\sigma_n/2}$ into the transformation $\bm P\mapsto \bm L$, in much the same way we did in \eqref{deff:LocalParRegEdge}. However, as we discussed in Section~\ref{sec:issues}, for any $s\in \R$ fixed there are poles of $\sigma_n$ accumulating too fast near the origin, so $\sigma_n$ fails to be analytic in any small neighborhood of the origin and we have to stick to the non-constant jumps as above.

The {\bf RHP-$\bm L$} has a solution if, and only if, {\bf RHP-$\bm P^{(0)}$} has a solution. Such solutions need not be unique, as one could possibly improve on the asymptotic matching conditions on $\partial U_0$. The goal of the rest of this section is to describe a solution $\bm L$, and consequently a solution $\bm P^{(0)}$ related by \eqref{local_param_LP_rel}, with a more explicit control of the error term in {\bf RHP-$\bm L$-3}. For that, we use the model problem thoroughly studied in Sections~\ref{sec:modelproblem} and \ref{sec:asympanalymodelprobl}.

The construction that follows needs several quantities that appeared before. These are the conformal map $\psi$ appearing in Proposition~\ref{prop:conformalmappsi}, the function $\msf H_Q$ introduced with the help of Proposition~\ref{prop:ConformExtH}, the model RHP solution $\bm \Phi=\bm \Phi(\cdot \mid \msf h)$ introduced in Section~\ref{sec:modelproblem} and further discussed in Section~\ref{sec:asympanalymodelprobl}, and the constant $\msf q_0$ and matrices $\bm U_0$ and $\bm M(z)$ from \eqref{eq:def_q0q1}--\eqref{eq:def_globalparM}. With all these quantities at hand, we set
\begin{equation}\label{local_par_L_Phi}
\begin{aligned}
& \bm L(z)\deff \bm F_n(z)\bm \Phi_n(z),\quad z\in U_0\setminus \cSig_{\bm S},\qquad \text{with the choices}  \\
& \wh{\bm\Phi}_n(\zeta)=\bm \Phi\left(\zeta \mid \msf h=\msf h_n\right), \quad \msf h_n(\zeta)\deff \sad+n^{2/3}\msf H_Q(\zeta/n^{2/3}), \\
& {\bm \Phi}_n(z)\deff \wh{\bm \Phi}_n\left(\zeta=n^{2/3}\psi(z)\right),\\
& {\bm F}_n(z)\deff \bm M(z)\bm U_0(n^{2/3}\psi(z))^{-\sp_3/4}=\bm U_0\msf m(z)^{\sp_3/4}(n^{2/3}\psi(z))^{-\sp_3/4}.
\end{aligned}
\end{equation}
With the identification $\tauad=n^{2/3}$ and thanks to Proposition~\ref{prop:ConformExtH}, the function $\msf h_n$ becomes admissible in the sense of Definition~\ref{deff:admissibleh}, so the notation for the corresponding solution $\bm\Phi_n=\bm \Phi(\cdot\mid \msf h_n)$ chosen above is consistent with the solution $\bm \Phi_\tauad(\zeta)\big|_{\tauad=n^{2/3}}$ in \eqref{deff:Phit}. For later reference, we keep track of the expansion
\begin{equation}\label{eq:expansionmsfhnorigin}
\msf h_n(\zeta)=\sad-\msf c_{\msf H}\zeta+\Boh(\zeta^{-2}),\quad \zeta\to 0,\quad \text{where we recall that}\quad \msf c_{\msf H}=\frac{\tad}{\msf c_V},
\end{equation}
compare \eqref{eq:expansionHQ} with Definition~\ref{deff:admissibleh}--(ii).

\begin{prop}\label{prop:propertieshatPhi}
Fix $\sad_0>0$ and $\tad_0\in (0,1)$. There exists $n_0=n_0(\sad_0,\tad_0)$ for which for any $s\geq -\sad_0$ and any $\tad\in [\tad_0,1/\tad_0]$ the matrix ${\bm \Phi}_n$ exists for every $n\geq n_0$. This matrix satisfies the jump
$$
{\bm \Phi}_{n,+}(z)={\bm \Phi}_{n,-}(z)\bm J_{\bm L}(z),\quad z\in U_0\cap \Gamma_{\bm S}.
$$
Furthermore, for the matrix
$$
{\bm \Phi}^{(1)}_n\deff \bm \Phi^{(1)}(\msf h=\msf h_n),\quad \text{with }\bm \Phi^{(1)}(\msf h) \text{ as in \eqref{eq:ModelRHPAsymp}},
$$
the asymptotic expansion
\begin{equation*}
{\bm \Phi}_n(z)=\left(\bm I+\frac{1}{n^{2/3}}\frac{1}{\psi(z)}{\bm \Phi}_n^{(1)}+\Boh(n^{-4/3})\right)(n^{2/3}\psi(z))^{\sp_3/4} \bm U_0^{-1}  
\ee^{-n\phi(z)\sp_3} ,\quad n\to \infty,
\end{equation*}
holds true uniformly for $z\in \partial U_0$ and uniformly for $\sad\geq -\sad_0$ and $\tad\in [\tad_0,1/\tad_0]$. 

\end{prop}
\begin{proof}
The existence of ${\bm \Phi}_n$ is granted by the first claim of Theorem~\ref{thm:PhitauAsympt}. The jumps for ${\bm \Phi}_{n,+}$ follow from the jumps in \eqref{eq:jumpPhimodel}, the definition of $\msf h_n$ taken in \eqref{local_par_L_Phi}, the correspondence \eqref{eq:RelHQQpsi} and the conformality of the change of variables $\zeta=n^{2/3}\psi(z)$. The asymptotic expansion for ${\bm \Phi}_n$ is immediate from \eqref{eq:ModelRHPAsymp}
\end{proof}

The introduction of the additional notation $\wh{\bm \Phi}_n$, which plays the role of the local parametrix in the variable $\zeta$, is convenient for later calculations. At that moment, some of its properties will be needed, and we keep track of these properties with the next result. For the formal statement, we recall that $\phiad(\xi\mid \Sad,\Tad)$ is the solution to the integro-differential PII that already appeared in \eqref{eq:intdiffPII} and \eqref{eq:transfPsiccphiad}, and $\bm\Phi_0^{(1)}$ is the residue matrix from \eqref{deff:Phikpz_residue} that collects the functions $\msf P(\Sad,\Tad)$ and $\msf Q(\Sad,\Tad)$ which, in turn, are related to $\phiad$ as explained in \eqref{deff:coeffqrp} {\it et seq.}.

\begin{prop}
Fix $\sad_0>0$ and $\tad_0\in (0,1)$ and $\nu\in (0,2/3)$, and let $\msf c_{\msf H}$ be as in \eqref{eq:expansionHQ}. The following asymptotic formulas hold true uniformly for $\sad\geq -\sad_0$ and $\tad\in [\tad_0,1/\tad_0]$.

The matrix $\bm\Phi_n^{(1)}$ from Proposition~\ref{prop:propertieshatPhi} satisfies
\begin{equation}\label{eq:asympPhin1Phi01residue}
\left({\bm \Phi}_n^{(1)}\right)_{21}=\frac{\ii }{\msf c_{\msf H}^{1/2}}\left(\msf p(\sad,\msf c_{\msf H})-\frac{\sad^2}{4\msf c_{\msf H}^{3/2}}\right)+\Boh(n^{-\nu}),\quad n\to\infty.
\end{equation}
Furthermore, for $\Phicc$ being the solution to the RHP in \eqref{eq:jumpPhicc}--\eqref{eq:AsymptPhicc}, the estimate
\begin{equation}\label{eq:asympPhinPhi0}
\wh{\bm \Phi}_{n,+}(\zeta)=(\bm I+\Boh(n^{-\nu}))\Phiccp(\zeta\mid \sad,\msf c_{\msf H}),\quad n\to\infty,
\end{equation}
holds true uniformly for $\zeta\in \cSig$ with $|\zeta|\leq n^{\frac{1}{3}-\frac{1}{2}\nu}$. 
\end{prop}

\begin{proof}
It is immediate from the identification $\tauad=n^{2/3}$, Equations~\eqref{local_par_L_Phi} and Theorem~\ref{thm:PhitauAsympt}.
\end{proof}

Next, we now verify that $\bm L$ given as in \eqref{local_par_L_Phi} indeed solves the {\bf RHP-{$\bm L$}}.
\begin{prop}\label{prop:estimateL}
The matrix $\bm L$ solves the {\bf RHP-{$\bm L$}}. 

Furthermore, setting
\begin{equation}\label{eq:AsymptLlocalpar}
 \bm L_1(z)\deff \frac{\left({\bm \Phi}^{(1)}_n\right)_{21}}{\psi(z)^{1/2}}\bm M(z)\bm U_0 \bm E_{21} \bm U_0^{-1}\bm M(z)^{-1}, \quad z\in \partial U_0,
\end{equation}
the condition $\bm L${\bf -3.} is improved to
\begin{equation}\label{eq:estimateL_proof1} 
\bm L(z)=\left(\bm I+\frac{1}{n^{1/3}}\bm L_1(z) + \Boh(n^{-2/3})\right)\bm M(z)\ee^{-n\phi(z)\sp_3},\quad n\to \infty,
\end{equation}
uniformly for $z\in \partial U_0$ and $\sad\geq -\sad_0$, $\tad\in [\tad_0,1/\tad_0]$, for any $\sad_0>0,\tad_0\in (0,1)$.
\end{prop}
\begin{proof}
First we prove that ${\bm F}$ is analytic. For that, notice that a jump for it may come only from the factors $\bm M$ and $\psi^{1/4}$, and therefore only possibly in the interval $(-a,0)\cap U_0$. However, along this interval it is simple to compute that $(\psi^{1/4})_\pm = \ee^{\pm \pi \ii/4}|\psi|^{1/4}$, and also that 
$$
\bm U_0^{-1}\bm J_{\bm M}(x)\bm U_0=\ii \sp_3,
$$
from which we obtain
$$
{\bm F}_-(x)^{-1}{\bm F}_+(x)=n^{\sp_3/6}|\psi|^{\sp_3/4}\ee^{-\pi\ii\sp_3/4}\ii \sp_3 \ee^{-\pi\ii\sp_3/4}|\psi|^{-\sp_3/4}n^{-\sp_3/6}=\bm I,
$$
so ${\bm F}$ is indeed analytic across $(-a,0)\cap U_0$. In principle, ${\bm F}$ may have an isolated singularity at $0$, but because 
$$
\bm M(z),\bm \varphi(z)^{1/4}=\Boh(z^{1/4})\quad \text{as }z\to 0,
$$
we see that $z=0$ is a removable singularity. In virtue of the definition of $\bm L$ in \eqref{local_par_L_Phi} and the jump for $\bm \Phi_n$ from Proposition~\ref{prop:propertieshatPhi}, the analyticity of $\bm F$ is enough to conclude that $\bm L$ satisfies {\bf RHP.}$\bm L${\bf -1}.

Knowing that ${\bm F}$ is analytic, the jump for $\bm L$ is precisely the same as the jump for ${\bm \Phi}_n$, and by Proposition~\ref{prop:propertieshatPhi} we thus have that $\bm L$ satisfies {\bf RHP.}$\bm L${\bf -2}.

Finally, to the asymptotic condition \eqref{eq:AsymptLlocalpar}. For that, we use the asymptotic condition for ${\bm \Phi}_n$ given by Proposition~\ref{prop:propertieshatPhi} and the definition of ${\bm F}$ and write
\begin{multline*}
\bm L(z)=\bm M(z)\bm U_0n^{-\sp_3/6}\psi(z)^{-\sp_3/4}\left(\bm I+ \frac{1}{n^{2/3}\psi(z)}{\bm \Phi}^{(1)}_n+\Boh(n^{-4/3}) \right) \\
\times 
n^{\sp_3/6}\psi(z)^{\sp_3/4}\bm U_0^{-1} \ee^{-n\phi(z)\sp_3},\quad n\to\infty,
\end{multline*}
and where the error is uniform for $z\in \partial U_0$ and $\sad,\tad$ as claimed. Since $\partial U_0$ remains within a positive distance from the unique zero $z=0$ of $\psi$, the function $|\psi^{1/4}|$ remains bounded from below away from zero, so the corresponding conjugation of the error by the term $\psi^{\sp_3/4}$ does not change the order of the error. Next, the conjugation by $n^{\sp_3/6}$ contributes at most to an error of order $n^{1/3}$, and only in the $(2,1)$-entry. The remaining term $\bm M$ is bounded along $\partial U_0$, so it can be commuted with the error term above without changing its order, leading to \eqref{eq:estimateL_proof1}. And \eqref{eq:estimateL_proof1} is indeed an improvement of the asymptotic condition {\bf RHP.}$\bm L${\bf -3}, because from Lemma~\ref{lem:estimate_q_q0} we know that $\bm M=(\bm I+\Boh(n^{-1/3}))\bm G$ uniformly as claimed.
\end{proof}

We now trace back the transformation $\bm L\mapsto \bm P^{(0)}$ and construct the latter as
\begin{equation}\label{deff:ParamOrigin}
\bm P^{(0)}(z)\deff {\bm F}_n(z){\bm \Phi}_n(z)\ee^{n\phi(z)\sp_3},\quad z\in U_0\setminus \cSig_{\bm S},
\end{equation}
keeping track that ${\bm \Phi}_n$ and $\bm F_n$ are introduced in \eqref{local_par_L_Phi}. With this construction and thanks to Lemma~\ref{lem:estimate_q_q0} and \eqref{eq:estimateL_proof1}, the matching condition {\bf RHP.}$\bm P^{(0)}${\bf -3} is 
\begin{equation}\label{eq:errorParOriginfinalform}
\bm P^{(0)}(z) =\left(\bm I+\frac{1}{n^{1/3}}\left(\bm L_1(z)+\msf q^{(1)}_0\sp_3 -\msf q^{(1)}(z) \bm M(z)\sp_3\bm M(z)^{-1}\right)+\Boh(n^{-2/3})\right)\bm G(z),\quad n\to \infty,
\end{equation}
and the errors are uniform for $\sad\geq -\sad_0$ and $\tad_0\leq \tad\leq 1/\tad_0$.

\subsection{Final transformation}\label{sec:finalRops}\hfill

The final transformation combines the local and global parametrices to remove the non-decaying jumps from $\bm S$. 

Set
$$
U\deff U_0\cup D_\delta(-a),
$$
orienting $\partial U_0$ and $\partial D_\delta(-a)$ in the clockwise direction. With $\bm P^{(a)}$ being the local parametrix near $-a$, $\bm P^{(0)}$ the local parametrix near the origin and $\bm G$ the global parametrix, we introduce a last parametrix $\bm P$ with unified notation as
$$
\bm P(z)=
\begin{cases}
\bm P^{(a)}(z),& z\in D_\delta(-a)\setminus \Gamma_{\bm S}, \\ 
\bm P^{(0)}(z), & z\in U_0\setminus \Gamma_{\bm S}, \\
\bm G(z), & \text{elsewhere on } \C\setminus (\Gamma_{\bm S}\cup \partial U).
\end{cases}
$$
The final transformation $\bm S\mapsto \bm R$ is then
$$
\bm R(z)\deff 
\bm S(z)\bm P(z)^{-1},\quad z\in \C\setminus (\Gamma_{\bm S}\cup \partial U).
$$
With this transformation, the jumps that $\bm S$ has in common with the parametrices $\bm G$ and $\bm P$ get canceled, and we remain with jumps only away from $[-a,0]$ and $U$. With
$$
\Gamma_{\bm R}\deff \partial U\cup \Gamma_{\bm S}\setminus \left([-a,0]\cup U\right),
$$
the matrix $\bm R$ satisfies the following RHP.

\begin{enumerate}[\bf R-1.]
\item The matrix $\bm R:\C\setminus \Gamma_{\bm R}\to \C^{2\times 2}$ is analytic.
\item For $z\in \Gamma_{\bm R}$, it satisfies the jump $\bm R_+(z)=\bm R_-(z)\bm J_{\bm R}(z)$,
with
\begin{equation}\label{eq:jumpJR}
\bm J_{\bm R}(z)\deff
\begin{dcases}
\bm G(z)\bm J_{\bm S}(z)\bm G(z)^{-1}, & z\in \Gamma_{\bm R}\setminus \partial U, \\
\bm P^{(0)}(z)\bm G(z)^{-1}, & z\in \partial U_0,  \\
\bm P^{(a)}(z)\bm G(z)^{-1}, & z\in \partial D_\delta(-a), 
\end{dcases}
\end{equation}
\item With $\bm R_1\deff \bm S_1-\bm G_1$,
\begin{equation}\label{eq:expansionR1infinity}
\bm R(z)=\bm I+\frac{1}{z}\bm R_1+\Boh(1/z^2),\quad z\to\infty.
\end{equation}
\end{enumerate}

To conclude that $\bm R$ is asymptotically close to the identity, we control its jumps. We use the matrix norm notation introduced in \eqref{deff:matrixnorm}--\eqref{deff:matrixLpqnorm}.

\begin{prop}\label{prop:estimateJR}
Fix $\sad_0>0$ and $\tad_0\in (0,1)$. There is $\eta>0$ for which the jump matrix for $\bm R$ satisfies the estimates
$$
\|\bm J_{\bm R}-\bm I\|_{L^1\cap L^2\cap L^\infty(\Gamma_{\bm R}\setminus \partial U)}=\Boh(\ee^{-\eta n})
$$
and
$$
\|\bm J_{\bm R}-\bm I\|_{L^1\cap L^2\cap L^\infty(\partial U_0)}=\Boh(n^{-1/3}),\quad 
\|\bm J_{\bm R}-\bm I\|_{L^1\cap L^2\cap L^\infty(\partial D_\delta(-a))}=\Boh(n^{-1})
$$
as $n\to\infty$, uniformly for $\sad\geq -\sad_0$ and $\tad_0\leq \tad\leq 1/\tad_0$.
\end{prop}

\begin{proof}
The function $\bm G$ remains uniformly bounded away from the interval $[-a,0]$, and the first claimed decay then follows from \eqref{eq:jumpJR} and Proposition~\ref{prop:decayJS}.

Using \eqref{eq:decayparametrixa} and \eqref{eq:jumpJR}, we obtain the uniform estimate of the jump along $\partial D_\delta(-a)$ directly. Since $\partial D_\delta(-a)$ is bounded, the $L^1$ and $L^2$ estimates also follow. Finally, estimate for the jump along $\partial U_0$ follows similarly, once we recall \eqref{eq:errorParOriginfinalform}.
\end{proof}

With Proposition~\ref{prop:estimateJR} at hand, the small norm theory of Riemann-Hilbert problems yields
\begin{theorem}\label{thm:Rasymptops}
Fix $\sad_0>0$ and $\tad_0>0$. The matrix $\bm R$ satisfies
$$
\bm R(z)= \bm I+\Boh(n^{-1/3}),\quad \bm R'(z)=\Boh(n^{-1/3}) \quad \text{and}\quad  \bm R(w)^{-1}\bm R(z)=\bm I+\Boh\left(\frac{z-w}{n^{1/3}}\right),\quad n\to\infty,
$$
and the residue matrix $\bm R_1$ satisfies
\begin{equation}\label{eq:intreprresidueR1}
\bm R_1=-\frac{1}{2\pi \ii}\int_{\partial U_0}\bm R_-(s)(\bm J_{\bm R}(s)-\bm I)\dd s+\Boh(n^{-1}),
\end{equation}
where the error terms are all uniform for $z,w$ on the same connected components of $\C\setminus \Gamma_{\bm R}$, and they are also uniform for $\sad\geq -\sad_0$ and $\tad_0\leq \tad\leq 1/\tad_0$.
\end{theorem}

\begin{proof}
The arguments are standard, so we only sketch them. The small norm theory for RHPs ensures the representation
\begin{equation}\label{eq:spectralformR}
\bm R(z)=\bm I+\frac{1}{2\pi \ii}\int_{\Gamma_{\bm R}}\frac{\bm R_-(s)(\bm J_{\bm R}(s)-\bm I)}{s-z}\dd s,\quad z\in \C\setminus \Gamma_{\bm R}.
\end{equation}
This ensures the estimates for $\bm R$ and $\bm R'$ away from $\Gamma_{\bm R}$ and, because the jump matrix is analytic on a neighborhood of $\Gamma_{\bm R}$, we are able to extend these estimates also to $\Gamma_{\bm R}$. To obtain the estimate involving $z$ and $w$, we then write, with the help of Cauchy's formula,
\begin{equation}\label{eq:EstRwz}
\bm R(w)^{-1}\bm R(z)=\bm I+\bm R(w)^{-1}\left(\bm R(w)-\bm R(z)\right)=\bm I+\bm R(w)^{-1}\frac{w-z}{2\pi \ii}\oint \frac{\bm R(s)}{(s-w)(s-z)}\dd s,
\end{equation}
where the integral is over a contour on $\bm C\setminus \Gamma_{\bm R}$ encircling both $w$ and $z$. The right-hand side is now $\bm I+\Boh((w-z)n^{-1/3})$ by the estimate on $\bm R$ already proven.

Finally, the estimate for $\bm R_1$ follows expanding \eqref{eq:spectralformR} as $z\to\infty$, and then using that the jump matrix decays to the identity at least as $\Boh(n^{-1})$ away from $\partial U_0$.
\end{proof}

Later on, we also need the first term in the asymptotic expansion for $\bm R$, we state it as a separate result.

\begin{prop}
Fix $\sad_0>0$ and $\tad_0>0$. Setting
\begin{equation}\label{deff:whR1}
\wh{\bm R}_1(z)\deff \frac{1}{2\pi \ii} \int_{\partial U_0}\left(\frac{(\wh{\bm \Phi}_n^{(1)})_{21}}{\psi(s)^{1/2}\msf m(s)^{1/2}}\bm E_{21}-\msf q^{(1)}(s)\sp_3\right)\frac{\dd s}{s-z}+\msf q_0^{(1)}\sp_2,\quad z\in U_0,
\end{equation}
the matrix $\bm R$ has the expansion
\begin{equation}\label{eq:ExpRwhR1}
\bm R(z)= \bm I+\frac{1}{n^{1/3}}\bm U_0 \wh{\bm R}_1(z)\bm U_0^{-1} + \Boh(n^{-2/3}),\quad n\to \infty
\end{equation}
where the error term is uniform for $z\in U_0$, and it is also uniform for $\sad\geq -\sad_0$ and $\tad_0\leq \tad\leq 1/\tad_0$.
\end{prop}

\begin{proof}
Starting from the representation \eqref{eq:spectralformR}, we write
$$
\bm R(z)-\bm I=\frac{1}{2\pi \ii}\int_{\partial U_0}\frac{\bm R_-(s)(\bm J_{\bm R}(s)-\bm I)}{s-z}\dd s+\Boh(n^{-1})=\frac{1}{2\pi \ii}\int_{\partial U_0}\frac{\bm J_{\bm R}(s)-\bm I}{s-z}\dd s+\Boh(n^{-2/3}),
$$
and the result follows combining \eqref{eq:jumpJR} and \eqref{eq:errorParOriginfinalform} and performing straightforward calculations.
\end{proof}

All the transformations $\bm Y\mapsto \bm T\mapsto \bm S\mapsto \bm R$ involve only analytic factors in their construction. As such, we can actually slightly deform the contour $\Gamma_{\bm R}$ in all these steps, so that in fact the estimates in the previous result are valid everywhere on $\C$, interpreting them as with boundary values when $z,w\in \Gamma_{\bm R}$. We will use this fact without further warning.

With this theorem at hand, in the next sections we recover the needed asymptotic formulas for the proof of our main results.

\section{Proof of main results}\label{sec:ProofMainResults}

The next step in the direction of concluding our asymptotic analysis is to unravel the transformations $\bm R\mapsto \bm S\mapsto \bm T\mapsto \bm Y$. Introduce
$$
\bm\Lambda_n(x)\deff \bm I+\frac{\ee^{2n\phi_+(x)}}{\sigma_n(x)}\chi_{(-a,0)}(x)\bm E_{21}\quad \text{and}\quad  \bm \Delta_n(x)\deff \left(\bm I+\frac{\chi_{(-a,0)}(x)}{\sigma_n(x)}\bm E_{21}\right),\quad x\in \R,
$$
which are related by
$$
\bm \Lambda_n(x)\ee^{-n\phi_+(x)\sp_3}=\ee^{-n\phi_+(x)\sp_3}\bm \Delta_n(x).
$$
Then the result of the unfolding of the transformations is
\begin{equation}\label{eq:unfoldtransfYR}
\bm Y_+(z)=\ee^{n\ell_V\sp_3}\bm R_+(x)\bm P_+(x)\bm \Lambda_n(x)\ee^{-n(\phi_+(x)-V(x)/2)\sp_3},\quad x\in \R.
\end{equation}
We split the proofs of our results in the next few sections.

\subsection{Proof of Theorem~\ref{thm:limitingkernel}}\hfill 

Thanks to \eqref{eq:unfoldtransfYR} and Theorem~\ref{thm:Rasymptops}, the expression \eqref{eq:KnYxy} reduces to the asymptotic formula
\begin{equation}\label{eq:msfKnunfold}
\ee^{-n(V(x)+V(y))/2}\msf K_n^Q(x,y)=\frac{\ee^{-n(\phi_+(x)+\phi_+(y))}}{2\pi\ii(x-y)}\bm e^\tp_2\bm\Lambda_{n}(y)^{-1}\bm P_+(y)^{-1}\left(\bm I+\Boh\left(\frac{x-y}{n^{1/3}}\right)\right)\bm P_+(x)\bm\Lambda_{n}(x)\bm e_1,
\end{equation}
valid as $n\to\infty$, $x,y\in \R$ and uniformly for $\sad,\tad$ as in Theorem~\ref{thm:Rasymptops}.

For the value $\msf c_V>0$ introduced in \eqref{eq:derivativepsicpsi}, we scale $x=u_n$ and $y=v_n$ as in \eqref{deff:scalingunvn}, 
where $u,v$ are on any given compact of the real axis. For such values of $u$, for $n$ large enough the points $u_n$ are always on the neighborhood $U_0$ of the origin where $\bm P=\bm P^{(0)}$, and from \eqref{deff:ParamOrigin} we simplify
$$
\bm e_2^\tp \bm \Lambda_{n}^{-1}\bm P_+^{-1}=\ee^{n\phi_+}\bm e_2^\tp \bm \Delta_n^{-1}\bm \Phi_{n,+}^{-1}\bm F_n^{-1}\quad \text{and}\quad 
\bm P_+\bm \Lambda_{n}\bm e_1=\bm F_n\bm \Phi_{n,+}\bm \Delta_n\bm e_1\ee^{n\phi_+}.
$$
where all the quantities above are evaluated at $u_n$. The next step is to plug these identities into \eqref{eq:msfKnunfold}, when doing so we also use the estimates
$$
\bm F_n(u_n)=\bm U_0\left(\bm I+\Boh(n^{-2/3})\right)(\msf c_V a)^{-\sp_3/4}n^{-\sp_3/6},\quad n\to\infty,
$$
which follows from \eqref{local_par_L_Phi}, \eqref{eq:def_U0} and \eqref{eq:derivativepsicpsi}, and its immediate consequence
$$
\bm F_n(v_n)^{-1}\bm F_n(u_n)=n^{\sp_3/6}\left(\bm I+\Boh\left(\frac{u-v}{n^{2/3}}\right)\right)n^{-\sp_3/6}=\bm I+\Boh\left(\frac{u-v}{n^{1/3}}\right),\quad n\to\infty,
$$
which is obtained with arguments similar to \eqref{eq:EstRwz}. These estimates are uniform for $u,v$ in compacts of $\R$, and are also uniform for $\sad\geq -\sad_0$ and $\tad\in [\tad_0,1/\tad_0]$. Equation~\eqref{eq:msfKnunfold} then simplifies to
\begin{multline*}
\ee^{-n(V(u_n)+V(v_n))/2}\msf K_n^Q(u_n,v_n)= \\
\frac{\msf c_V n^{2/3}}{2\pi \ii(u-v)}
\bm e_2^\tp
\bm \Delta_n(v_n)^{-1}
\bm \Phi_{n,+}(v_n)^{-1}
\left(\bm I+\Boh\left(\frac{u-v}{n^{1/3}}\right)\right)\bm \Phi_{n,+}(u_n)\bm \Delta_n(u_n)\bm e_1.
\end{multline*}
Now, with the definition of $\bm \Phi_n$ in \eqref{local_par_L_Phi} at hand, the aid of  \eqref{eq:asympPhinPhi0} and the constant $\msf c_{\msf H}$ appearing in \eqref{eq:expansionmsfhnorigin}, we get
$$
\bm\Phi_{n,+}(u_n)=\wh{\bm \Phi}_{n,+}(n^{2/3}\psi(u_n))=\left(\bm I+\Boh(n^{-\nu})\right)\Phiccp(u\mid \sad,\msf c_{\msf H}),\quad n\to\infty,
$$
for any $\nu\in (0,2/3)$, with the error being valid uniformly for $u$ in compacts of $\R$, and also uniformly for $\sad\geq -\sad_0, \tad_0\leq \tad\leq 1/\tad_0$. Also, thanks to Theorem~\ref{prop:existencePhicc} we know that $\Phiccp(u\mid \sad,\msf c_{\msf H})^{\pm 1}$ remain bounded for $u$ in compacts and $\sad\geq -\sad_0, \tad_0\leq \tad\leq 1/\tad_0$. Therefore,
$$
\bm \Phi_{n,+}(v_n)^{-1}\left(\bm I+\Boh\left(\frac{u-v}{n^{1/3}}\right)\right)\bm \Phi_{n,+}(u_n)
=
\left(\bm I+\Boh(n^{-\nu})\right)\Phiccp(v\mid \sad,\msf c_{\msf H})^{-1}\Phiccp(u\mid \sad,\msf c_{\msf H})+\Boh\left(\frac{u-v}{n^{1/3}}\right)
$$
as $n\to\infty$.
In addition, the estimate
\begin{equation}\label{eq:estimatesigmanfinalproofkernel}
\frac{1}{\sigma_n(u_n)}=1+\ee^{-\sad+\tad u/\msf c_V}(1+\Boh(n^{-2/3}))=1+\ee^{-\sad+\msf c_{\msf H} u}(1+\Boh(n^{-2/3})),\quad n\to\infty,
\end{equation}
is valid uniformly for $u$ in compacts, uniformly for $\sad\in \R$ and $\tad_0\leq \tad\leq 1/\tad_0$, and is immediate from \eqref{def:perturbedweight}.
Combining everything, and denoting $\chi_+=\chi_{[0,\infty)}$, we obtained the asymptotic formula
\begin{multline*}
\frac{\ee^{-n(V(u_n)+V(v_n))/2}}{\msf c_V n^{2/3}}\msf K_n^Q(u_n,v_n)= 
\frac{1+\Boh(n^{-\nu})}{2\pi \ii(u-v)}
\bm e_2^\tp
\left(\bm I-(1+\ee^{-\sad+\msf c_{\msf H} v})\chi_{+}(v)\bm E_{21}\right)\\
\times
\Phiccp(v\mid \sad,\msf c_{\msf H})^{-1}\Phiccp(u\mid \sad,\msf c_{\msf H})\left(\bm I+(1+\ee^{-\sad+\msf c_{\msf H} u})\chi_{+}(u)\bm E_{21}\right)\bm e_1+\Boh(n^{-1/3}),\quad n\to\infty,
\end{multline*}
which is valid uniformly for $u,v$ in compacts of $\R$, and also uniformly for $\sad\geq -\sad_0$ and $\tad_0\leq \tad\leq 1/\tad_0$. With $\Dcc$ as in \eqref{deff:Dcc}, this identity rewrites as
\begin{multline*}
\frac{\ee^{-n(V(u_n)+V(v_n))/2}}{\msf c_V n^{2/3}}\msf K_n^Q(u_n,v_n)= 
\frac{1+\Boh(n^{-\nu})}{2\pi \ii(u-v)}\\
\times \left[\left(\Phicc(v\mid \sad,\msf c_{\msf H})\Dcc(v\mid \sad,\msf c_{\msf H})\right)^{-1}\Phicc(u\mid \sad,\msf c_{\msf H})\Dcc(u\mid \sad,\msf c_{\msf H})\right]_{21,+}
+\Boh(n^{-1/3}),\quad n\to\infty,
\end{multline*}
and the proof of Theorem~\ref{thm:limitingkernel} is now completed using Equation~\eqref{eq:transfPhicckernelphicc}.

\subsection{Proof of Theorem~\ref{thm:asymptoticsnormingconstant}}\hfill 

Now that the RHP asymptotic analysis is completed, the proof of Theorem~\ref{thm:asymptoticsnormingconstant} follows standard steps. In our case, there is an additional cancellation that has to be accounted for at a later step, in virtue of the presence of the factor $\msf q$ in both the global and local parametrices, see \eqref{eq:constructionGparam} and \eqref{eq:errorParOriginfinalform}. So we opt for presenting the detailed calculation.

Starting from the representation \eqref{eq:relationnormingY1}, we unravel the transformations and obtain
\begin{equation}\label{eq:gammanasympt1}
-2\pi \ii \left(\msfga_{n-1}^{(n,Q)}\right)^2=\ee^{-2n\ell_V}\left(\bm R_1+\bm G_1\right)_{21}=\ee^{-2n\ell_V}\left((\bm R_1)_{21}- \frac{\ii a}{4} -\frac{1}{n^{1/3}}\frac{\tad^{1/2}a^{1/2}}{4\pi \ii}F_{-1/2}(\sad) +\Boh(n^{-2/3}) \right),
\end{equation}
with $\bm R_1$ as in \eqref{eq:expansionR1infinity}, and where for the second identity we used the definition of $\bm G_1$ in \eqref{eq:definitionG1} and the estimate for $\msf q_0$ from Lemma~\ref{lem:estimate_q_q0}.

It remains to estimate $\bm R_1$, which we do so starting from \eqref{eq:intreprresidueR1}. Using Cauchy-Schwartz, Propositions~\ref{prop:estimateJR} and Theorem~\ref{thm:Rasymptops}, we write
$$
\bm R_1=-\frac{1}{2\pi \ii}\int_{\partial U_0}(\bm J_{\bm R}(s)-\bm I)\dd s+\Boh(n^{-2/3}),\quad n\to\infty.
$$
The matrix $\bm J_\bm R$ is in \eqref{eq:jumpJR}, and combining its explicit expression along $\partial U_0$ with \eqref{eq:errorParOriginfinalform} and \eqref{eq:def_globalparM}, after a cumbersome but straightforward calculation we arrive at
$$
\left(\bm R_1\right)_{21}=-\frac{1}{2\pi \ii n^{1/3}}\int_{\partial U_0}\left(\frac{\ii \msf q^{(1)}(s)(\msf m(s)-1)}{2\msf m(s)^{1/2}}+(\bm L_1(s))_{21}\right)\dd s+\Boh(n^{-2/3}),\quad n\to\infty,
$$
where we recall that $\msf m$, $\msf q^{(1)}$ and $\bm L_1$ are given in \eqref{eq:def_U0}, \eqref{deff:msfqo1} and \eqref{eq:AsymptLlocalpar}, respectively. Using this explicit expression for $\bm L_1$, we see that
$$
(\bm L_1(s))_{21}=\frac{\left({\bm \Phi}^{(1)}_n\right)_{21}}{2\msf m(s)^{1/2}\psi(s)^{1/2}}.
$$
Both functions $\msf m$ and $\psi$ are analytic in a neighborhood of the origin, and vanish linearly therein. Combining in addition with \eqref{eq:derivativepsicpsi}, we see that the product $\msf m(z)^{1/2}\psi(z)^{1/2}$ admits an analytic continuation near the origin, with the expansion
$$
\msf m(s)^{1/2}\psi(s)^{1/2}=\frac{\msf c_V}{a}s(1+\Boh(s)),\quad s\to 0.
$$
Thus, computing residues
$$
\int_{\partial U_0}(\bm L_1(s))_{21}\dd s=-\pi \ii \frac{a }{\msf c_V}\left({\bm \Phi}^{(1)}_n\right)_{21}.
$$
Similarly, using now \eqref{deff:msfqo1} we obtain
$$
\int_{\partial_{U_0}}\frac{\ii \msf q^{(1)}(s)(\msf m(s)-1)}{2\msf m(s)^{1/2}}\dd s=-\msf t^{1/2}a^{1/2} F_{-1/2}(\sad)
$$
Hence,
$$
\left(\bm R_1\right)_{21}=
\frac{1}{n^{1/3}}\left(\frac{a}{2\msf c_V}\left(\bm \Phi_n^{(1)}\right)_{21}+\frac{\tad^{1/2}a^{1/2}}{4\pi\ii} F_{-1/2}(\sad) \right)+\Boh(n^{-2/3}),
$$
and \eqref{eq:gammanasympt1} updates to
$$
\left(\msfga_{n-1}^{(n,Q)}(\sad)\right)^2=\ee^{-2n\ell_V}\left(\frac{a}{8\pi}-\frac{1}{n^{1/3}}\frac{a}{4\pi \ii}\left(\bm \Phi_n^{(1)}\right)_{21}+\Boh(n^{-2/3})\right),\quad n\to\infty,
$$
uniformly for $\sad\geq -\sad_0$ and $\tad_0\leq \tad\leq 1/\tad_0$. We now need only to apply \eqref{eq:asympPhin1Phi01residue} with any $\nu\geq 1/3$ to complete the proof.

\subsection{Proof of Theorem~\ref{thm:asymptoticsqLaplacetransf}}\hfill 

Unlike the already proven major results, the proof of Theorem~\ref{thm:asymptoticsqLaplacetransf} is not a straightforward consequence of the steepest descent analysis concluded with Theorem~\ref{thm:Rasymptops}. It does rely substantially on Theorem~\ref{thm:Rasymptops}, but several other inputs are also needed along the way. Equipped with \eqref{eq:KnYxy}, the idea is to account for the different approximations for $\msf K_n^Q$ on the different components 
$$
\R\setminus (U_0\cup D_\delta(-a)\cup (-a,0)),\quad \R\cup U_0, \quad (-a,-\delta,-a+\delta) \quad \text{and}\quad (-a+\delta,0)\setminus U_0
$$ 
that arise from the RHP, and integrate each such approximation. With Proposition~\ref{prop:intreprLn} in mind, we are thus able to recover asymptotics for $\msf L_n^Q$ itself. As one would expect, it turns out that in this process the terms that arise away from $U_0$ become all exponentially negligible, and only the contribution from $U_0$ survives in the leading contribution. The contribution that arrives this way involves $\bm \Phi_n$, and we further need to split it into different parts and still account for some exact cancellations to arrive at the leading asymptotic contribution. We postpone this analysis to the next section, where it is split into several different lemmas, and summarize the outcome with the next result. For its statement, we recall that $\msf h_n$ and $\wh{\bm \Phi}_n$ were introduced in \eqref{local_par_L_Phi}, and we also set
\begin{equation}\label{deff:Deltan0}
\wh{\bm \Delta}_n(x)\deff \bm I+ (1+\ee^{-\msf h_n(u)})\chi_{(0,+\infty)(x)}\bm E_{21},\quad x\in \R.
\end{equation}

\begin{prop}\label{prop:finalestKnQ}
Fix $\sad_0>0$ and $\tad_0\in (0,1)$. There exists a function $\msf R(u)=\msf R_1(u\mid \tad)$ satisfying
$$
\int_{\sad}^\infty \msf R(u)\dd u=\Boh(n^{-1/3}),\quad n\to\infty,
$$
uniformly for $\sad\geq -\sad_0$ and $\tad_0\leq \tad\leq 1/\tad_0$, and for which the identity
\begin{multline*}
\int_{-\infty}^{\infty}\msf K_n^Q(x,x\mid u)\frac{\omega_n(x\mid u)}{1+\ee^{u+n^{2/3}Q(x)}}\dd x\\ =\frac{1}{2\pi \ii}
\int_{n^{2/3}\zeta_0}^{n^{2/3}\zeta_1}\frac{\ee^{\msf h_n(u)}}{\left(1+\ee^{\msf h_n(u)}\right)^2}\left[ \wh{\bm \Delta}_n(x)^{-1}\wh{\bm \Phi}_{n,+}(x)^{-1}\left(\wh{\bm \Phi}_{n,+} \wh{\bm\Delta}_n\right)'(x)\right]_{21}\dd x+\msf R(u),
\end{multline*}
holds true for every $u\geq -\sad_0$ and every $\tad\in [\tad_0,1/\tad_0]$.
\end{prop}

In words, Proposition~\ref{prop:finalestKnQ} is saying that the major contribution to $\msf L_n^Q$ comes from the neighborhood $U_0$, so from the local parametrix $\wh{\bm\Phi}_n$. But according to the developments from the previous sections, this local parametrix is close to the id-PII parametrix $\Phicc$, and we are ready to conclude our last major result.

\begin{proof}[Proof of Theorem~\ref{thm:asymptoticsqLaplacetransf}]
We combine Propositions~\ref{prop:intreprLn} and \eqref{prop:finalestKnQ} to obtain
$$
\log\msf L_n^Q(\sad)=-\frac{1}{2\pi \ii}\int_{\sad}^\infty
\int_{n^{2/3}\zeta_0}^{n^{2/3}\zeta_1}\frac{\ee^{\msf h_n(u)}}{\left(1+\ee^{\msf h_n(u)}\right)^2}\left[ \wh{\bm \Delta}_n(x)^{-1}\wh{\bm \Phi}_{n,+}(x)^{-1}\left(\wh{\bm \Phi}_{n,+} \wh{\bm\Delta}_n\right)'(x)\right]_{21}\dd x\dd u+\Boh(n^{-1/3}),
$$
valid as $n\to\infty$ and uniformly for $\sad\geq -\sad_0$ and $\tad_0\leq \tad\leq 1/\tad_0$. Next, with the identifications \eqref{local_par_L_Phi}--\eqref{eq:expansionmsfhnorigin} in mind, we estimate the integral on the right-hand side with the help of Theorem~\ref{thm:integralPhitauPhicc} and conclude the proof.
\end{proof}

\section{Technical Lemmas}\label{sec:technicalestimates}

It remains to prove Proposition~\ref{prop:finalestKnQ}, analysis which we split into several technical lemmas in this section. 

Our starting point is the integral representation for $\msf L_n^Q$ from Proposition~\ref{prop:intreprLn} with the asymptotic information for $\msf K_n^Q$ provided by the RHP analysis. For $x\in\R$, set
\begin{equation}\label{deff:msfA}
\msf A(x)\deff \left[ \bm \Lambda_{n}(x)^{-1}\bm P_+(x)^{-1}\bm R_+(x)^{-1}\bm R_+'(x)\bm P_+(x)\bm\Lambda_n(x)+\bm \Lambda_{n}(x)^{-1}\bm P_+(x)^{-1}\bm P_+'(x)\bm \Lambda_{n}(x) \right]_{21}.
\end{equation}
Recalling \eqref{eq:KnYxx}, the unwrap of the transformations of the RHP yields the identity
$$
\msf K_n^Q(x,x\mid u)\frac{\omega_n(x\mid u)}{1+\ee^{u+n^{2/3}Q(x)}}=\frac{\sigma_n(x\mid u)\ee^{-2n\phi_+(x)}}{2\pi \ii (1+\ee^{u+n^{2/3}Q(x)})}\left(\msf A(x)+\left(\frac{\ee^{2n\phi_+(x)}}{\sigma_n(x)}\right)'\chi_{\mcal G^+}(x)\right)
$$
We now need to integrate each of the terms on the right-hand side above, first in $x\in \R$ and then in $u\in (\sad,+\infty)$. Each term will be analyzed individually, also depending on whether we integrate $x$ in the bulk, each of the edges or away from the support $[-a,0]$ of the equilibrium measure. To simplify notation, it is convenient to introduce the additional notation for each relevant integral, and for an arbitrary set $J\subset \R$ denote
\begin{equation}\label{deff:I1int}
\msf I_1(J)\deff \int_J \frac{\sigma_n(x\mid u)\ee^{-2n\phi_+(x)}}{1+\ee^{u+n^{2/3}Q(x)}}\msf A(x) \dd x
\end{equation}
and
\begin{equation}\label{deff:I2int}
\begin{aligned}
\msf I_2(J) \deff & \int_J \frac{\sigma_n(x\mid u)\ee^{-2n\phi_+(x)}}{1+\ee^{u+n^{2/3}Q(x)}}\left(\frac{\ee^{2n\phi_+(x)}}{\sigma_n(x)}\right)'\dd x \\
 =& \int_J \left(2n\phi_{n,+}(x)-\frac{n^{2/3}Q'(x)}{1+\ee^{u+n^{2/3}Q(x)}}\right)\frac{1}{1+\ee^{u+n^{2/3}Q(x)}}\dd x,
\end{aligned}
\end{equation}
which are functions of $u\in \R$ as well. With $\varepsilon_0,\varepsilon_1>0$ being determined by $(-\varepsilon_0,\varepsilon_1)=U_0\cap \R$, the split
\begin{multline}\label{eq:intKnIsplit}
2\pi \ii \int_{-\infty}^{\infty}\msf K_n^Q(x,x\mid u)\frac{\omega_n(x\mid u)}{1+\ee^{u+n^{2/3}Q(x)}}\dd x=
\msf I_1(-\infty,-a-\delta)+\msf I_1(-a-\delta,a+\delta)
\\ +\msf I_1(a+\delta,-\varepsilon_0)+\msf I_1(-\varepsilon_0,\varepsilon_1)+\msf I_1(\varepsilon_1,+\infty)
+\msf I_2(-a,-\varepsilon_0)+\msf I_2(-\varepsilon_0,0)
\end{multline}
is immediate, and with the next series of lemmas we estimate each of the terms on the right-hand side. 

\begin{lemma}\label{lem:I1inftyI1aI2a}
Fix $\sad\in \R$ and $\tad_0\in (0,1)$. There exists $\eta>0$ for which the estimate
$$
\msf I_1(-\infty,-a-\delta)+\msf I_1(-a+\delta,-\varepsilon_0)+\msf I_2(-a,-\varepsilon_0)=\Boh\left(\ee^{-u}\ee^{-n^{2/3}\eta}\right)
$$
holds true uniformly for $u\geq \sad$ and $\tad_0\leq \tad\leq 1/\tad_0$.
\end{lemma}

\begin{proof}
On the intervals $(-\infty,-a-\delta)$ and $(\varepsilon_1,\infty)$ the function $\bm \Lambda_{n}$ is identically the identity matrix. On the interval $(a+\delta,-\varepsilon_0)$ the nontrivial entry of $\bm \Lambda_{n}$ is $\ee^{2n\phi_+}/\sigma_n$ and, because $Q<0$ and $\phi_{+}$ is purely imaginary in this interval, this quotient is bounded. Also, away from the endpoints $-a$ and $0$ we have $\bm P\equiv \bm G$. Both $\bm P$ and $\bm G$, and their $x$-derivatives, decay as $x\to \pm \infty$ and remain bounded as $n\to\infty$, all uniformly in $u$ and $\tad$ as claimed. All of these facts combined together, we obtain that for some constants $K>0$ and $n_0\geq 1$
\begin{equation}\label{eq:boundmsfA}
|\msf A(x)|\leq K,\quad \text{for all}\; x\in \R\setminus \left((-a-\delta,-a+\delta)\cup (-\varepsilon_0,\varepsilon_1)\right),\; n\geq n_0,\; \tad_0\leq \tad\leq 1/\tad_0, \; u\geq \sad.
\end{equation}
Next, we now use that $Q\geq 0$ on the interval $(-\infty,0)$ to bound
\begin{equation}\label{eq:I1estimate1}
0\leq \frac{\sigma_n(x\mid u)}{1+\ee^{u+n^{2/3}Q(x)}}\leq \frac{1}{1+\ee^{u+n^{2/3}Q(x)}}\leq \ee^{-u},\quad x\leq 0, 
\end{equation}
which is valid for any $n\geq 1, u\in \R,\tad >0$. Thus, combining everything we obtain
$$
|\msf I_1(-\infty,-a-\delta)|\leq K\ee^{-u}\int_{-\infty}^{-a-\delta}\ee^{-2n\phi_+(x)}\dd x,
$$
and using Proposition~\ref{prop:functionphitau}--(iv),(v), this proves the bound for $\msf I_1(-\infty,-a-\delta)$.

For the second integral, we term $\phi_+$ is oscillatory, so to obtain the exponential decay in the $x$-integral we have to argue differently and as follows. The function
$$
v\mapsto \frac{1}{1+\ee^{-v}}\frac{1}{1+\ee^v}=\frac{\ee^{-v}}{(1+\ee^{-v})^2}
$$
is strictly increasing on $(-\infty,0)$ and strictly decreasing on $(0,+\infty)$. Because $Q>0$ on $(-\infty,0)$ and $Q(0)=0$, by reducing $U_0$ if necessary we can assume without loss of generality that $Q(x)\geq Q(-\varepsilon_0)$ for every $x\in [-a,-\varepsilon_0]$. This way, $v=u+n^{2/3}Q(x)\geq v_0\deff u+n^{2/3}Q(-\varepsilon_0)$ and, because $u$ is assumed to be bounded from below, we can make sure that $v_0>0$ for every $u$. Therefore
$$
0\leq \frac{\sigma_n(x\mid u)}{1+\ee^{u+n^{2/3}Q(x)}}\leq \frac{\ee^{-v_0}}{(1+\ee^{-v_0})^2}\leq \ee^{-u-n^{2/3}Q(-\varepsilon_0)},
$$
and upon integration and using Proposition~\ref{prop:functionphitau}--(ii), we obtain
$$
|\msf I_2(-a+\delta,-\varepsilon_0)|\leq (a-\delta+\varepsilon_0)K\ee^{-u-n^{2/3}Q(-\varepsilon_0)}.
$$
for every $n\geq n_0$ and $\tad,u$ as claimed. 

For the estimate for $\msf I_2(-a,-\varepsilon_0)$ we use again the last inequality in \eqref{eq:I1estimate1} and also that both $\phi'_{+}$ and $Q'$ are continuous, and hence bounded, on $(-a,-\varepsilon_0)$. This concludes the proof.
\end{proof}

For the integral over $(\varepsilon_1,+\infty)$, it is easier to actually estimate its $u$-integral directly.

\begin{lemma}\label{lem:estI1epsilon1infty}
Fix $\sad_0>0$ and $\tad_0\in (0,1)$. There exists $\eta>0$ for which the estimate
$$
\int_{\sad}^\infty\msf I_1(\varepsilon_1,+\infty)\dd u=\Boh\left(\ee^{-n\eta}\right)
$$
holds true uniformly for $\sad\geq -\sad_0$ and $\tad_0\leq \tad\leq 1/\tad_0$.
\end{lemma}

\begin{proof}
With the bound \eqref{eq:boundmsfA} we see that it is enough to estimate the integral
$$
\int_{\sad}^\infty \int_{\varepsilon_1}^\infty \frac{\sigma_n(x\mid u)\ee^{-2n\phi(x)}}{1+\ee^{u+n^{2/3}Q(x)}}\dd x\dd u=\int_{\varepsilon_1}^\infty\frac{\ee^{-2n\phi_+(x)}}{1+\ee^{\sad+n^{2/3}Q(x)}}\dd x,
$$
where for the equality we used Tonelli's Theorem to interchange the order of integration, and then integrated exactly. The term $(1+\ee^{\sad+n^{2/3}Q(x)})^{-1}$ is bounded by $1$, and using Proposition~\ref{prop:functionphitau}--(iv),(v) we see that the integral of $\ee^{-2n\phi}$ is $\Boh(\ee^{-\eta n})$ for some $\eta>0$ independent of $\sad$ and $\tad$, which concludes the proof.
\end{proof}

Next, we analyze the contribution coming from a neighborhood of the endpoint $z=a$.

\begin{lemma}\label{lem:I1adelta}
Fix $\sad\in \R$, $\tad_0\in (0,1)$. There exists $\eta>0$ such that the estimate
$$
\msf I_1(-a-\delta,-a+\delta) =\Boh(\ee^{-u-\eta n^{2/3}}),\quad n\to\infty,
$$
is valid uniformly for $u\geq \sad$ and $\tad_0\leq \tad\leq 1/\tad_0$.
\end{lemma}

\begin{proof}
On the neighborhood $D_\delta(-a)$ the function $\bm P=\bm P^{(a)}$ is the Airy local parametrix \eqref{deff:LocalParRegEdge}, which involves the function $\bm\Phi_{\ai}$ evaluated at the argument $\zeta=n^{2/3}\varphi(z)$. In the $\zeta$-plane, we fix $R>0$ for which the asymptotic expansion \eqref{eq:AsympAiryPar} is valid for $|\zeta|>R$, and split the analysis into two cases, namely for $|n^{2/3}\varphi(z)|\leq R$ and $|n^{2/3}\varphi(z)|\geq R$, and for
$$
A_n\deff \{x\in (-a-\delta,-a+\delta)\mid |n^{2/3}\varphi(x)|\leq R \},\quad B_n\deff (-a-\delta,-a+\delta)\setminus A_n
$$
we write
$$
\msf I_1(-a-\delta,-a+\delta)=\msf I_1(A_n)+\msf I_1(B_n).
$$
and now analyze each integral on the right-hand side separately.

For $|n^{2/3}\varphi(z)|\leq R$, the terms $\bm \Phi_{\ai,+}(\zeta=n^{2/3}\varphi(z))$ and $\bm \Phi'_{\ai,+}(\zeta=n^{2/3}\varphi(z))$ consist of continuous functions evaluated inside the compact interval $[-R,R]$, and therefore they are bounded. By the same reason, the expression \eqref{eq:FFactorRegPar} shows that $\bm F$ is bounded for $|n^{2/3}\varphi(z)|\leq R$. On the other hand, without further analysis we obtain that $\bm F'$ may grow at most as $\Boh(n^{1/6})$. Finally, combining with \eqref{eq:EstLogSigman} and the fact that the determinant of $\bm P$ is identically $1$, we conclude that for $|n^{2/3}\varphi(x)|\leq R$,
\begin{equation}\label{eq:EstShortRangePplus1}
\bm P_+(x)=\Boh(1)\ee^{n\phi_+(x)\sp_3},\quad \bm P_+(x)^{-1}=\ee^{-n\phi_+(x)\sp_3}\Boh(1)\quad \text{and}\quad \bm P_+'(x)=\Boh(n)\ee^{-n\phi_+(x)\sp_3} 
\end{equation}
which is valid as $n\to\infty$ and uniformly for $x\in (-a-\delta,-a+\delta)$, also uniformly for $u\geq \sad$ and $\tad_0\leq \tad\leq 1/\tad_0$.

We use \eqref{eq:EstShortRangePplus1} back into \eqref{deff:msfA}, and combined with the fact that $|\sigma_n^{-1}|\leq 2$ for $x\leq (-\infty,0)$ the result is that
$$
\msf A(x)=\left[ \bm\Lambda_{n}(x)^{-1}\ee^{-n\phi_+(x)\sp_3}\Boh(n)\ee^{n\phi_+(x)\sp_3} \bm\Lambda_{n}(x)\right]_{21}=\ee^{2n\phi_+(x)}\Boh(n),\quad n\to\infty,
$$
where the last error term is a scalar error, valid uniformly in $x,u,\tad$ as before. Integrating, we obtain that for some absolute constant $K>0$,
\begin{equation}\label{eq:estimateI1An}
|\msf I_1(A_n)|\leq Kn \int_{A_n}  \frac{\sigma_n(x\mid u)}{1+\ee^{u+n^{2/3}Q(x)}}\dd x\leq Kn \int_{A_n}\ee^{-u-n^{2/3}Q(x)}\dd x=\Boh(\ee^{-u-\eta n^{2/3}}),\quad n\to \infty,
\end{equation}
where for the second inequality we used that $0<\sigma_n\leq 1$ and for the last estimate we used that $Q$ is strictly positive on $(-a-\delta,-a+\delta)\supset A_n$.

Next, we consider the case $|n^{2/3}\varphi(z)|\geq R$. In such situation, the asymptotics \eqref{eq:AsympAiryPar} take place, and when combined with \eqref{eq:FFactorRegPar} they yield
\begin{equation}\label{eq:EstLongRangePplus1}
\bm P_+(x)=\bm G_+(x)\ee^{-\sp_3\log\sigma_n(x)/2}\left(\bm I+\Boh(n^{-1})\right)\ee^{\sp_3\log\sigma_n(x)/2}=\bm G_+(x)\left(\bm I+\Boh(n^{-1})\right).
\end{equation}
where for the last equality we used \eqref{eq:EstLogSigman}.

The matrix $\bm P_+$ is bounded as $x\to 0$, whereas $\bm G$ is not, but the cancellation that leads to this boundedness of $\bm P$ is not captured by this asymptotics. Nevertheless, as we need some uniform control over $x$, we now account for the behavior of $x\to 0$ in a rough manner as follows. We are assuming that $|n^{2/3}\varphi(x)|\geq R$, and because $\varphi$ is conformal with $\varphi(0)=0$ this means that $|x+a|\geq c/n^{2/3}$, for some fixed $c>0$. The function $\bm G$ has a fourth-root singularity at $x=0$, and therefore we arrive at the rough estimate
\begin{equation}\label{eq:EstLongRangePplus2}
\bm P_+(x)=\Boh(n^{1/6}),\quad n\to\infty.
\end{equation}
The estimate \eqref{eq:EstLongRangePplus1} can be differentiated term by term. With arguments similar to the ones we just applied, we arrive at the rough estimate
\begin{equation}\label{eq:EstLongRangePplus3}
\bm P'_+(x)=\Boh(n^{7/6}),\quad n\to\infty.
\end{equation}
These latter two estimates are valid uniformly for $u\geq \sad$ and $\tad_0\leq \tad\leq 1/\tad_0$ as $n\to\infty$.

Finally, on the interval $(-a,0)$ the factor $\phi_+$ is purely imaginary, implying that $\bm\Lambda_{n}^{\pm 1}$ remains bounded therein. All combined, we obtained that
$$
\msf A(x)=\Boh(n),\quad n\to\infty,
$$
uniformly for $x\in B_n$ and $u,\tad$ as claimed. Proceeding as in \eqref{eq:estimateI1An} we obtain a bound for $I_1(B_n)$ and complete the proof.
\end{proof}

It remains to analyze the two integrals $\msf I_1(-\varepsilon_0,\varepsilon_1)$ and $\msf I_2(-\varepsilon_0,0)$. The hardest analysis is the integral $\msf I_1$ which, as will turn out, contains both the leading contribution, a term to cancel $\msf I_2$ and asymptotically negligible terms. For ease of presentation, we explore the expression for $\msf A$ in \eqref{deff:msfA} as a sum and split
\begin{equation}\label{eq:splitI1J1J2}
\msf I_1(-\varepsilon_0,\varepsilon_1)=\msf J_1(-\varepsilon_0,\varepsilon_1)+\msf J_2(-\varepsilon_0,\varepsilon_1),
\end{equation}
where, for any measurable set $J\subset \R$,
\begin{equation}\label{deff:J1J2}
\begin{aligned}
\msf J_1(J) & \deff\int_J
\left[ \bm \Lambda_{n}(x)^{-1}\bm P_+(x)^{-1}\bm R(x)^{-1}\bm R'(x)\bm P_+(x)\bm \Lambda_{n}(x)\right]_{21}\frac{\sigma_n(x\mid u)\ee^{-2n\phi_+(x)}}{1+\ee^{u+n^{2/3}Q(x)}} \dd x,\\
\msf J_2(J) & \deff\int_J
\left[ \bm \Lambda_{n}(x)^{-1}\bm P_+(x)^{-1}\bm P_+'(x)\bm \Lambda_{n}(x) \right]_{21}\frac{\sigma_n(x\mid u)\ee^{-2n\phi_+(x)}}{1+\ee^{u+n^{2/3}Q(x)}}\dd x,
\end{aligned}
\end{equation}
and analyze each of these terms separately. For the estimation of $\msf J_1$, it is also easier to perform the $u$-integral, just as we did in Lemma~\ref{lem:estI1epsilon1infty}.

\begin{lemma}
Fix $\sad_0>0$ and $\tad_0\in (0,1)$. The estimate
\begin{equation}
\int_{\sad}^\infty\msf J_1(-\varepsilon_0,\varepsilon_1)\dd u=\Boh(n^{-1/3})
\end{equation}
holds true uniformly for $\sad\geq -\sad_0$ and $\tad_0\leq \tad\leq 1/\tad_0$.
\end{lemma}

\begin{proof}
The idea is similar to the proof of Lemma~\ref{lem:I1adelta}. We fix a number $R>0$ for which the asymptotic expansion \eqref{eq:ModelRHPAsymp} for $\bm\Phi=\wh{\bm \Phi}_n$ is valid for $|\zeta|\geq R$, uniformly in $\tad,\sad$ as claimed (see Remark~\ref{remark:UnifAsympPhitauad}). Then, introduce
\begin{equation}\label{deff:setsCnDn}
C_n\deff \{ x\in (-\varepsilon_0,\varepsilon_1)\mid |n^{2/3}\psi(x)|\leq R \},\quad D_n\deff (-\varepsilon_0,\varepsilon_1)\setminus C_n.
\end{equation}
We now find bounds for the integrands, with separate arguments for each of $C_n$ and $D_n$. 

Recalling \eqref{deff:ParamOrigin}, on the interval $(-\varepsilon_0,\varepsilon_1)$ the parametrix is $\bm P=\bm P^{(0)}=\bm F_n\bm\Phi_n\ee^{n\phi\sp_3}$, with $\bm F_n$ and $\bm \Phi_n$ as in \eqref{local_par_L_Phi}. Using the definition of $\bm F_n$ in \eqref{local_par_L_Phi} and Theorem~\ref{thm:Rasymptops}, we express 
\begin{equation}\label{eq:EstRRprimeF}
\bm F_n^{-1}\bm R^{-1}\bm R'\bm F_n=\bm F_n^{-1}(\bm R^{-1}-\bm I)\bm R'\bm F_n-\bm F_n^{-1}\bm R'\bm F_n=-\bm F_n^{-1}\bm R'\bm F_n+\Boh(n^{-1/3}),\quad n\to\infty,
\end{equation}
valid uniformly when evaluated at $x\in C_n$ and also uniformly for $\sad,\tad$ as claimed. Using again the definition of $\bm F_n$, the fact that $\psi/\msf m$ remains bounded near $z=0$ and \eqref{eq:ExpRwhR1}, 
$$
\bm F_n^{-1}\bm R'\bm F_n=\frac{1}{n^{1/3}}\psi^{\sp_3/4}\msf m^{-\sp_3/4}n^{\sp_3/6}\wh{\bm R}_1' n^{-\sp_3/4}\msf m^{\sp_3/4}\psi^{-\sp_3/4}+\Boh(n^{-1/3}).
$$
A careful inspection on \eqref{deff:whR1} shows that $(\wh{\bm R}_1(z))_{12}$ is independent of $z$, so $R'$ has zero $(1,2)$-entry. Therefore we conclude that
$$
\bm F_n^{-1}\bm R^{-1}\bm R'\bm F_n=\Boh(n^{-1/3}),
$$
stressing that this is valid for $x\in C_n$. Also, from \eqref{eq:asympPhinPhi0} we learn that $\bm\Phi_n(x)=\Boh(1)$ uniformly on $C_n$. Everything we have so far combined yields that 
$$
\left[\bm \Lambda_{n}^{-1}\bm P_+^{-1}\bm R^{-1}\bm R'\bm P\bm \Lambda_{n}\right]_{21}=\left[\bm \Lambda_{n}^{-1} \ee^{-n\phi\sp_3}\Boh(n^{-1/3})\ee^{n\phi_3\sp_3} \bm \Lambda_{n}\right]_{21}=\ee^{-2n\phi}\left[\bm \Delta_{n}^{-1}\Boh(n^{-1/3})\bm \Delta_{n}\right]_{21}.
$$
Along $(-a,0)$ we have $Q>0$ so that $1/\sigma_n$ remains bounded uniformly, and consequently $\bm \Delta_{n}=\Boh(1)$ in the same interval. Using this information in the last displayed equation, we thus obtain that
\begin{equation}\label{eq:J1Cn}
\msf J_1(C_n)=\int_{C_n}\Boh(n^{-1/3}) \frac{\sigma_n(x\mid u)}{1+\ee^{u+n^{2/3}Q(x)}}\dd x,\quad n\to\infty,
\end{equation}
uniformly for $u\geq \sad$ and $\tad_0\leq \tad\leq 1/\tad_0$, and where the error term is uniform for $x\in C_n$.

Next, along $D_n$ we now use the asymptotic expansion \eqref{eq:ModelRHPAsymp} for $\bm\Phi=\wh{\bm \Phi}_n$ and the definition of $\msf F_n$ and obtain that
$$
\bm P(x)=\bm U_0\msf m_+(x)^{\sp_3/4}\bm U_0^{-1}\left(1+\Boh(n^{-1/3})\right)=\bm U_0
\begin{pmatrix}
\Boh(1) & 0 \\
0 & \Boh(n^{1/6})
\end{pmatrix}
\bm U_0^{-1}
\left(1+\Boh(n^{-1/3})\right)
,\quad x\in D_n,
$$
and similarly
$$
\bm P(x)=
\left(1+\Boh(n^{-1/3})\right)
\bm U_0
\begin{pmatrix}
\Boh(n^{1/6}) & 0 \\
0 & \Boh(1)
\end{pmatrix}
\bm U_0^{-1}
,\quad x\in D_n,
$$
all valid as $n\to\infty$, uniformly for $x\in D_n$ and uniformly in $u,\tad$ as claimed

With the same arguments that we applied in \eqref{eq:EstRRprimeF}, we obtain along $D_n$ as well
$$
\bm P^{-1}\bm R^{-1}\bm R'\bm P=\Boh(n^{-1/3}),\quad \text{so that}\quad \msf J_1(D_n)=\int_{D_n}\Boh(n^{-1/3}) \frac{\sigma_n(x\mid u)\ee^{-2n\phi_+(x)}}{1+\ee^{u+n^{2/3}Q(x)}}\dd x,\quad n\to\infty.
$$
The factor $\phi_+$ is purely imaginary on $(-a,0)$ and positive on $[0,+\infty)$, so the term $\ee^{-2n\phi_+}$ in the integrand above is bounded by a uniform constant.  Combining with \eqref{eq:J1Cn}, we conclude
$$
\msf J_1(-\varepsilon_0,\varepsilon_1)=\Boh(n^{-1/3})\int_{-\varepsilon_0}^{\varepsilon_1}\frac{\sigma_n(x\mid u)}{1+\ee^{u+n^{2/3}Q(x)}}\dd x
$$
Finally, we now integrate in $u$ and use Tonelli's Theorem to interchange the order of integration, obtaining just like in the proof of Lemma~\ref{lem:estI1epsilon1infty} that
$$
\int_{\sad}^\infty\int_{-\varepsilon_0}^{\varepsilon_1}\frac{\sigma_n(x\mid u)}{1+\ee^{u+n^{2/3}Q(x)}}\dd x \dd u=
\int_{-\varepsilon_0}^{\varepsilon_1}\frac{1}{1+\ee^{\sad+n^{2/3}Q(x)}}\dd x\leq \varepsilon_1+\varepsilon_0,
$$
which concludes the proof.
\end{proof}

Next, in our pursuit of analyzing \eqref{eq:intKnIsplit}, the final missing piece of the puzzle is the term $\msf J_2(-\varepsilon_0,\varepsilon_0)$ that arises from $\msf I_1(-\varepsilon_0,\varepsilon_0)$. To state the rigorous results, we recall once again that the conformal map $\psi$ was introduced in Proposition~\ref{prop:conformalmappsi}, the function $\msf h_n$ is given in \eqref{local_par_L_Phi}, the function $\wh{\bm \Delta}_n$ is in \eqref{deff:Deltan0} and in addition also set
$$
\zeta_0\deff\psi(-\varepsilon_0)<0,\quad \zeta_1\deff \psi(\varepsilon_1)>0.
$$

\begin{lemma}\label{lem:estJ2}
Fix $\sad_0>0$ and $\tad_0\in (0,1)$. There exists a function $\msf R_1(u)=\msf R_n(u\mid \tad)$ satisfying the estimate 
$$
\int_{\sad}^\infty \msf R(u)\dd u=\Boh(n^{-2/3}),\quad n\to\infty,
$$
uniformly for $\sad\geq -\sad_0$ and $\tad_0\leq \tad\leq 1/\tad_0$, and for which the identity
\begin{multline}\label{eq:I3I1Int}
\msf J_2(-\varepsilon_0,\varepsilon_0)=-\msf I_2(-\varepsilon_0,0)\\ 
+\int_{n^{2/3}\zeta_0}^{n^{2/3}\zeta_1}\frac{\ee^{\msf h_n(u)}}{\left(1+\ee^{\msf h_n(u)}\right)^2}\left[ \wh{\bm \Delta}_n(x)^{-1}\wh{\bm \Phi}_{n,+}(x)^{-1}\left(\wh{\bm \Phi}_{n,+} \wh{\bm\Delta}_n\right)'(x)\right]_{21}\dd x+\msf R_1(u),
\end{multline}
holds true for every $u\in \R$ and every $\tad >0$.
\end{lemma}

\begin{proof}
Recall that $\msf J_2$ was introduced in \eqref{deff:J1J2}. In the interval $(-\varepsilon_0,\varepsilon_1)$ the local parametrix $\bm P=\bm P^{(0)}$ coincides with \eqref{deff:ParamOrigin}. A direct calculation from \eqref{local_par_L_Phi} shows that the matrix $\bm F_n$ therein satisfies the identity
$$
\bm F_n'(z)=\frac{1}{4}\left(\frac{\msf m'(z)}{\msf m(z)}-\frac{\psi'(z)}{\psi(z)}\right)\bm F_n(z)\sp_3.
$$
With the help of this identity, we express
$$
\bm P_+^{-1}\bm P_+'=
\ee^{-n\phi_+\sp_3}\left[\frac{1}{4}\left(\frac{\msf m'}{\msf m}-\frac{\psi'}{\psi}\right) \bm \Phi_{n,+}^{-1}\sp_3 \bm \Phi_{n,+}+\bm \Phi_{n,+}^{-1}\bm \Phi_{n,+}'+n\phi'_+\sp_3 \right]\ee^{n\phi_+\sp_3},
$$
where both sides are evaluated at $x\in (-\varepsilon_0,\varepsilon_1)$. Thus,
\begin{multline}\label{eq:IntegrI3:1}
\Bigg[\bm\Lambda_{n}^{-1}
\bm P_+^{-1}\bm P_+'
\bm\Lambda_{n}\Bigg]_{21}\\
\begin{aligned}
& = -2n\frac{\phi_+' \chi_0\ee^{2n\phi_+}}{\sigma_n}+\ee^{2n\phi_+}\left[ \bm \Delta_{n} \bm \Phi_{n,+}^{-1}\bm \Phi_{n,+}' \bm\Delta_n\right]_{21}
+\frac{\ee^{2n\phi_+}}{4}\left(\frac{\msf m'}{\msf m}-\frac{\psi'}{\psi}\right) \left[\bm\Delta_n^{-1}\bm \Phi_{n,+}^{-1}\sp_3 \bm \Phi_{n,+}\bm\Delta_n\right]_{21} \\
& 
\begin{multlined}
= -2n\frac{\phi_+' \chi_0\ee^{2n\phi_+}}{\sigma_n}
-\ee^{2n\phi_+}\left[\bm\Delta_n^{-1} \bm\Delta_n'\right]_{21}
+\ee^{2n\phi_+}\left[ \bm \Delta_{n} \bm \Phi_{n,+}^{-1}\left(\bm \Phi_{n,+} \bm\Delta_n\right)'\right]_{21}\\
\phantom{-2n\frac{\phi_+' \chi_0\ee^{2n\phi_+}}{\sigma_n}
-\ee^{2n\phi_+}\left[\bm\Delta_n^{-1} \bm\Delta_n'\right]_{21}}
+\frac{\ee^{2n\phi_+}}{4}\left(\frac{\msf m'}{\msf m}-\frac{\psi'}{\psi}\right) \left[\bm\Delta_n^{-1}\bm \Phi_{n,+}^{-1}\sp_3 \bm \Phi_{n,+}\bm\Delta_n\right]_{21} 
\end{multlined}
\end{aligned}
\end{multline}

We now multiply this last expression by $\ee^{-2n\phi_+}\sigma_n/(1+\ee^{u+n^{2/3}Q})$ and integrate. A direct calculation gives
$$
\frac{\sigma_n}{1+\ee^{u+n^{2/3}Q}}\left(2n\frac{\phi_+' \chi_0}{\sigma_n}+\left[\bm\Delta_n^{-1} \bm\Delta_n'\right]_{21}\right)=
\frac{\sigma_n\chi_0}{1+\ee^{u+n^{2/3}Q}}\left(2n\frac{\phi_+' }{\sigma_n}+\left(\frac{1}{\sigma_n}\right)'\right)=\frac{\chi_0\sigma_n\ee^{-2n\phi_+}}{1+\ee^{u+n^{2/3}Q}}\left(\frac{\ee^{2n\phi_+}}{\sigma_n}\right)',
$$
which is the integrand of $\msf I_2(-\varepsilon_0,\varepsilon_1)=\msf I_2(-\varepsilon_0,0)$, see \eqref{deff:I2int}. Thus, from \eqref{eq:IntegrI3:1} we obtain
\begin{equation}\label{eq:J2I1intRu}
\msf J_2(-\varepsilon_0,\varepsilon_1)=-\msf I_1(-\varepsilon_0,0)+\int_{-\varepsilon_0}^{\varepsilon_1} \frac{\sigma_n(x)}{1+\ee^{u+n^{2/3}Q(x)}}\left[\bm\Delta_n(x)^{-1}\bm \Phi_{n,+}(x)^{-1}\left(\bm \Phi_{n,+}(x) \bm\Delta_n(x)\right)'\right]_{21}\dd x+\msf R_1(u)
\end{equation}
where we have set
$$
\msf R_1(u)\deff \frac{1}{4}\int_{-\varepsilon_0}^{\varepsilon_1} \left(\frac{\msf m'(x)}{\msf m(x)}-\frac{\psi'(x)}{\psi(x)}\right) \frac{\sigma_n(x)}{1+\ee^{u+n^{2/3}Q(x)}}\left[\bm\Delta_n(x)^{-1}\bm \Phi_{n,+}(x)^{-1} \sp_3\bm \Phi_{n,+}(x) \bm\Delta_n(x)\right]_{21}\dd x.
$$
It is worth mentioning that both $\msf m$ and $\psi$ have a simple zero at $x=0$, so the first term in the integrand of $\msf R$ remains bounded near $x=0$.

Recalling \eqref{local_par_L_Phi}, the integrand written explicitly in \eqref{eq:J2I1intRu} is of the form
$$
f_1(n^{2/3}\psi(x))\left[\bm f_2(n^{2/3}\psi(x))\left(\bm f_3(n^{2/3}\psi(x))\right)'\right]_{21}=n^{2/3}\psi'(x)f_1(n^{2/3}\psi(x))\left[\bm f_2(n^{2/3}\psi(x))\bm f_3'(n^{2/3}\psi(x))\right]_{21}
$$
with obvious choices of the functions $f_1,\bm f_2,\bm f_3$, and performing the change of variables $u=n^{2/3}\psi(x)$ we obtain the integral in the right-hand side of \eqref{eq:I3I1Int}. 

To conclude, it remains to show the bound for $\msf R$. For that, we use again the sets $C_n$ and $D_n$ from \eqref{deff:setsCnDn} and analyze the integral over each of these sets separately. 

Along $C_n$, the convergence \eqref{eq:asympPhinPhi0} ensures $\bm \Phi_{n,+}$ remains bounded uniformly, and combined with the boundedness of all the other terms on the whole interval $(-\varepsilon_0,\varepsilon_1)$ we obtain
\begin{multline*}
\int_{C_n}\left(\frac{\msf m'(x)}{\msf m(x)}-\frac{\psi'(x)}{\psi(x)}\right) \frac{\sigma_n(x)}{1+\ee^{u+n^{2/3}Q(x)}}\left[\bm\Delta_n(x)^{-1}\bm \Phi_{n,+}(x)^{-1} \sp_3\bm \Phi_{n,+}(x) \bm\Delta_n(x)\right]_{21}\dd x\\ 
=\Boh(1)\int_{C_n} \frac{\sigma_n(x)}{1+\ee^{u+n^{2/3}Q(x)}}\dd x.
\end{multline*}
The function $\psi$ is conformal, and consequently $|z|\leq cn^{-2/3}$ for $z\in C_n$ and some constant $c$ independent of $\tad,u,n$. In particular, this ensures that $n^{2/3}Q(x)\leq -\tilde c$ for every $x\in C_n$ and a constant $\tilde c>0$, and therefore
$$
\frac{\sigma_n(x)}{1+\ee^{u+n^{2/3}Q(x)}}\leq \frac{1}{1+\ee^{u-\tilde c}}\leq \ee^{-u+\tilde c}.
$$
Still because $\psi$ is conformal, we are sure that the Lebesgue measure of $C_n$ is $\Boh(n^{-2/3})$. Everything combined, we conclude the estimate
$$
\int_{C_n}\left(\frac{\msf m'(x)}{\msf m(x)}-\frac{\psi'(x)}{\psi(x)}\right) \frac{\sigma_n(x)}{1+\ee^{u+n^{2/3}Q(x)}}\left[\bm\Delta_n(x)^{-1}\bm \Phi_{n,+}(x)^{-1} \sp_3\bm \Phi_{n,+}(x) \bm\Delta_n(x)\right]_{21}\dd x=\Boh(\ee^{-u}n^{-2/3}),
$$
as $n\to\infty$, which is valid uniformly for $u,\tad$ as claimed by the Lemma.

Finally, on $D_n$ we use the expansion \eqref{eq:AsymptPhicc} for $\bm \Phi=\bm\Phi_n$, which provides
$$
\bm\Phi_{n,+}(x)^{-1}\sp_3\bm \Phi_{n,+}(x)=\ee^{n\phi_+(x)\sp_3}\left(\sp_2 +\Boh(n^{-1/3}) \right)\ee^{-n\phi_+(x)\sp_3},\quad n\to\infty,
$$
which is valid uniformly for $x\in D_n$ and uniformly in the parameters $u,\tad$ as required. After some straightforward calculations, we thus arrive at
\begin{multline*}
\int_{D_n}\left(\frac{\msf m'(x)}{\msf m(x)}-\frac{\psi'(x)}{\psi(x)}\right) \frac{\sigma_n(x)}{1+\ee^{u+n^{2/3}Q(x)}}\left[\bm\Delta_n(x)^{-1}\bm \Phi_{n,+}(x)^{-1} \sp_3\bm \Phi_{n,+}(x) \bm\Delta_n(x)\right]_{21}\dd x \\
= \int_{D_n}\left(\frac{\msf m'(x)}{\msf m(x)}-\frac{\psi'(x)}{\psi(x)}\right) \frac{\ii\sigma_n(x)\ee^{-2n\phi_+(x)}}{1+\ee^{u+n^{2/3}Q(x)}} \left[ 1+\frac{\chi_0(x)}{\sigma_n(x)^2}+\Boh(n^{-1/3}) \right]\dd x,\quad n\to\infty,
\end{multline*}
with, as always, uniform error term in $x\in D_n$, $u,\tad$. Each of the terms $(\msf m'/\msf m-\psi'/\psi)$ and  $\chi_0/\sigma_n^2$ is bounded on $(-\varepsilon_0,\varepsilon_1)$, so to bound the integral above it is enough to estimate
$$
\int_{-\varepsilon_0}^{\varepsilon_1} \frac{\sigma_n(x)\ee^{-n\phi_+(x)}}{1+\ee^{u+n^{2/3}Q(x)}}\dd x=\int_{-\varepsilon_0}^{0} \frac{\sigma_n(x)\ee^{-n\phi_+(x)}}{1+\ee^{u+n^{2/3}Q(x)}}\dd x+\int_{0}^{\varepsilon_1} \frac{\sigma_n(x)\ee^{-n\phi(x)}}{1+\ee^{u+n^{2/3}Q(x)}}\dd x.
$$
For the integral over $(-\varepsilon_0,0)$, we know that $\re\phi_+=0$, and then using Tonelli's Theorem to integrate first in $u$ we obtain that
$$
\int_{\sad}^{+\infty} \left|\int_{-\varepsilon_0}^{0} \frac{\sigma_n(x)\ee^{-n\phi_+(x)}}{1+\ee^{u+n^{2/3}Q(x)}}\dd x\right|\dd u\leq 
\int_{-\varepsilon_0}^0 \frac{1}{1+\ee^{\sad+n^{2/3}Q(x)}}\dd x.
$$
Changing variables $v=n^{2/3}Q$ in this last integral, it then follows that the right-hand side above is $\Boh(n^{-2/3})$ uniformly for $\sad\geq -\sad_0$ and $\tad_0\leq \tad\leq 1/\tad_0$.

Finally, for the integral over $(0,\varepsilon_1)$ we now have that $\phi\geq 0$ in this interval, and it is independent of $u$, so once again interchanging order of integration we obtain
$$
0\leq \int_\sad^\infty \int_{0}^{\varepsilon_1} \frac{\sigma_n(x)\ee^{-n\phi(x)}}{1+\ee^{u+n^{2/3}Q(x)}}\dd x \dd u=\int_0^{\varepsilon_1}\frac{\ee^{-n\phi(x)}}{1+\ee^{\sad+n^{2/3}Q(x)}}\dd x\leq \int_0^{\varepsilon_1}\ee^{-n\phi(x)}\dd x,
$$
and now changing variables $v=n\phi$ (which is well defined in this interval because of the local behavior \eqref{eq:phi_local_beh}) we see that the integral on the right-most side is $\Boh(n^{-1})$. This completes the proof.
\end{proof}

To conclude, it remains to prove Proposition~\ref{prop:finalestKnQ}.

\begin{proof}[Proof of Proposition~\ref{prop:finalestKnQ}]
Recalling \eqref{eq:intKnIsplit} and \eqref{eq:splitI1J1J2}, the result is an immediate consequence of Lemmas~\ref{lem:I1inftyI1aI2a}--\ref{lem:estJ2}.
\end{proof}

\appendix

\section{Laplace-type integrals}

For smooth enough functions $f,g$, consider the Laplace-type integral
$$
F(t)\deff \int_0^a g(x) \log(1+\ee^{-y-tf(x)})\dd x,\quad a\in (0,+\infty].
$$
We are interested in the asymptotic behavior of $F(t)$ as $t\to +\infty$ while $y$ also scales with $t$. The first result we need is a version of the usual Watson's Lemma.

\begin{lemma}\label{lem:Watson_app}
Fix $f(x)\equiv x$. Suppose also that $g\in L^1(0,a)$, and for some $\delta>0$ and $\kappa\in [0,1)$ it is of the form
$$
g(x)=\frac{\wt g(x)}{x^\kappa},\quad  0<x<\delta,
$$
where $\wt g$ is $C^\infty$ in a neighborhood of the origin. Assume that
\begin{equation}\label{laplacescaling}
y\geq -M,
\end{equation}
for some $M>0$ fixed. Then $F(t)$ admits an expansion of the form
$$
F(t)=\sum_{k=0}^N \frac{1}{t^{k+1-\kappa}}\wt g^{(k)}(0)F_{k-\kappa}(y) +\Boh(t^{-N-2+\kappa}),\quad t\to \infty,
$$
for any $N\geq 1$, where the error term is uniform for $y$ satisfying \eqref{laplacescaling}, and the coefficients $F_\beta(y)$ are given by
\begin{equation}\label{deff:coeff_FK_app}
F_\beta(y)\deff \int_0^\infty v^\beta \log(1+\ee^{-y-v})\dd v,\quad \beta>-1.
\end{equation}
For $y\geq 0$ and $k\in \Z, k\geq 1$, $F_k$ admits the alternative expression in terms of the polylog function $\Li_s$,
$$
F_k(y)=-k\Gamma(k)   \Li_{2+k}(-\ee^{-y}).
$$
\end{lemma}

\begin{proof}
The proof follows along the lines of the proof of the classical Watson's Lemma for Laplace transforms, see for instance \cite{MillerBookAsymp}, so we go over it without much detail. 

First off, we fix $\delta>0$ for which $\wt g$ admits a Taylor expansion of order $N$ in a neighborhood of the interval $(-\delta,\delta)$, and write for some $\eta>0$,
$$
\int_0^a g(x) \log(1+\ee^{-y-tx})\dd x = \int_0^\delta g(x) \log(1+\ee^{-y-tx})\dd x+\Boh(\ee^{-t \eta}), \quad t\to\infty,
$$
where the error term is obtained using that $g\in L^1(0,\infty)$ and that the function $x\mapsto \log (1+\ee^{-y-tx})$ is increasing. The value $\eta>0$ is independent of $y$, and the error is uniform for $y$ satisfying \eqref{laplacescaling}. For the integral on the right-hand side, we expand $\wt g$ in Taylor series and obtain
$$
\int_0^\delta g(x) \log(1+\ee^{-y-tx})\dd x=\sum_{k=0}^N \frac{\wt g^{(k)}(0)}{k!}\int_0^\delta u^{k-\kappa} \log(1+\ee^{-y-tu})\dd u + R_{N+1}(t),
$$
where the remainder satisfies
$$
|R_{N+1}(t)|\leq \frac{1}{(N+1)!}\sup_{0\leq x\leq \delta}|\wt g^{(N+1)}(x)| \int_0^\delta u^{N+1-\kappa} \log(1+\ee^{-y-tu})\dd u
$$
Performing the change of variables $u=tv$, we see that for $F_\beta(y)$ as defined,
$$
\int_0^\delta u^\beta \log(1+\ee^{-y-tu})\dd u=\frac{F_{\beta}(y)}{t^{\beta+1}}-\frac{1}{t^{\beta+1}}\int_{t\delta}^{\infty} v^{\beta}\log(1+\ee^{-y-v}).
$$
Using \eqref{laplacescaling} we see that the last integral is again $\Boh(\ee^{-\eta t})$ uniformly for $y$. Also, it is immediate that under the restriction \eqref{laplacescaling} each coefficient $F_\beta(y)$ is bounded in $y$, and this concludes the proof of the asymptotic formula. The alternative representation of $F_k(y)$ in terms of polylogs follows by integration by parts.
\end{proof}

Next, we move to more general exponents.

\begin{prop}\label{prop:LaplaceLogIntApp}
Suppose that $f$ is $C^\infty$ in a neighborhood of the origin, with a unique global minimum on $[0,a]$ at $x=0$ with $f(0)=0$, $f'(0)>0$, and that $g$ is as in Lemma~\ref{lem:Watson_app}. In addition, suppose that $y$ satisfies \eqref{laplacescaling} for some $M>0$ fixed. Then $F(t)$ admits an expansion of the form
$$
F(t)=\sum_{k=0}^N \frac{1}{t^{k+1-\kappa}}\wh g^{(k)}(0)F_{k-\kappa}(y) +\Boh(t^{-N-2+\kappa}),\quad t\to \infty,
$$
for any $N\geq 1$, where the error term is uniform for $y$ satisfying \eqref{laplacescaling}, $F_k$'s are as in 
Lemma~\ref{lem:Watson_app} and the function $\wh g$ is $C^\infty$ and satisfies the identity
$$
\wh g(f(x))=f(x)^\kappa g(x),\quad |x| \text{ sufficiently small}.
$$
\end{prop}

\begin{proof}
Because $f$ is assumed to be $C^\infty$ and to have a unique global minimum at $x=0$, by the Inverse Function Theorem it admits an inverse $f^{-1}$ in a neighborhood of the origin. With the change of variables
$$
u=f(x),\quad \text{and setting} \quad u_0\deff f(\delta), \quad \wh g(x)\deff \frac{x^\kappa \wt g(f^{-1}(x))}{(f^{-1}(x))^\kappa}=x^\kappa g(f^{-1}(x)), 
$$
for $\delta>0$ sufficiently small the integral is turned into
\begin{align*}
F(t)& =\int_0^\delta g(x)\log(1+\ee^{-y-tf(x)})\dd x+\int_\delta^a g(x)\log(1+\ee^{-y-tf(x)})\dd x \\
& = \int_0^{u_0} \frac{1}{x^\kappa}\wh g(x)\log(1+\ee^{-y-tu})\dd x +\Boh(\ee^{-\eta t}),
\end{align*}
for some $\eta>0$. Applying Lemma~\ref{lem:Watson_app} to the remaining integral concludes the proof.
\end{proof}

Next we move to a different scaling regime on $y$.

\begin{lemma}\label{lem:Watson_app_2}
Let $f$ and $g$ be as in Proposition~\ref{prop:LaplaceLogIntApp}. For any fixed constants $M>0$ and $\alpha\in [0,1)$, assume that
\begin{equation}\label{laplacescaling2}
-Mt^\alpha\leq  y\leq -\frac{1}{M}.
\end{equation}
Then $F(t)$ admits an expansion of the form
\begin{equation}\label{eq:laplaceform2app}
F(t)=\sum_{k=0}^N \frac{1}{t^{k+1-\kappa}}\wh g^{(k)}(0)\varphi_{k-\kappa}(y) +\Boh(t^{-N-2+\kappa}|y|^{N+3-\kappa}),\quad t\to \infty,
\end{equation}
for any $N\geq 1$, where the error term is uniform for $y$ satisfying \eqref{laplacescaling2}, and the coefficients $\varphi_k(y)$ are of the form
\begin{align*}\label{deff:coeff_FK_app2}
\varphi_\beta(y)& \deff \varphi^{(1)}_\beta(y)+\varphi^{(2)}_\beta(y)+\varphi^{(3)}_\beta(y),\quad \text{with}\\
\varphi_\beta^{(1)}(y) & \deff \frac{1}{(\beta+1)(\beta+2)}|y|^{\beta+2},\\
\varphi_\beta^{(2)}(y) &\deff \int_0^\infty \left(|y-u|^\beta+|y+u|^\beta\right)\log(1+\ee^{-u})\dd u,\\
\varphi_\beta^{(3)}(y) &\deff -\int_{|y|}^\infty |y+u|^\beta\log(1+\ee^{-u})\dd u,
\end{align*}
and the function $\wh g$ satisfies the identity
$$
\wh g(f(x))=f(x)^\kappa g(x),\quad |x| \text{ sufficiently small}.
$$
\end{lemma}

To illustrate the previous result, assume $\kappa=0$, write $y=-Mt^{\alpha}$ and send $t\to \infty$. When $\alpha=0$ the three functions $\varphi_k^{(j)}$ are all $\Boh(1)$. When $0<\alpha<1$, the function $\varphi_k^{(3)}$ decays exponentially fast, whereas $\varphi^{(1)}_k=\Boh(t^{(k+2)\alpha})$ while $\varphi^{(1)}_k=\Boh(t^{k\alpha})$. In either case, the error term is $\Boh(t^{-(1-\alpha)N+3\alpha-2})$, which is in fact decaying only when $\alpha<(N+2)/(N+3)$. In other words, for a given $\alpha$ the asymptotic formula \eqref{eq:laplaceform2app} only becomes truly useful (that is, with a small error) when $N$ is chosen sufficiently large.

\begin{proof}
We prove the result for $f(x)=x$, the general case follows along the same lines of the proof of Proposition~\ref{prop:LaplaceLogIntApp}.

Fix $\delta>0$ so that $g$ is $C^\infty$ on $[0,\delta]$. By replacing $g$ with $g\chi_{[0,a]}$ we may assume that $a=+\infty$. We can then bound
$$
\left|\int_\delta^\infty g(u)\log(1+\ee^{-y-tu})\dd u \right|\leq \|g\|_{L^1(0,\infty)}\ee^{-\eta t^\alpha-t \delta}=\Boh(\ee^{-\delta t/2}),
$$
as we are assuming that $\alpha<1$. Proceeding as in the proof of Proposition~\ref{lem:Watson_app}, we obtain
\begin{equation}\label{eq:appeq1}
\int_0^\delta g(x) \log(1+\ee^{-y-tx})\dd x=\sum_{k=0}^N \frac{\wt g^{(k)}(0)}{k!}\int_0^\delta u^{k-\kappa} \log(1+\ee^{-y-tu})\dd u + R_{N+1}+\Boh(\ee^{-\delta t/2}),
\end{equation}
with error term satisfying
\begin{equation}\label{eq:appeq2}
|R_{N+1}|\leq \frac{1}{(N+1)!}\|\wt g^{(k+1)}\|_{L^\infty (0,\delta)}\int_0^\delta u^{N+1-\kappa} \log(1+\ee^{-y-tu})\dd u.
\end{equation}
The integrals appearing in the sum and in the bound for the error above are of the same form, and now we study them. From \eqref{laplacescaling2} we see that by taking $t$ large enough we can always assume that the point $u_t\deff -y/t$ where the exponent $y+ut$ changes sign belongs to the interval $(0,\delta)$. We split the integrals on $(0,\delta)$ into the integrals over $(0,u_t)$ and $(u_t,\delta)$. For the first one, we then express it as
\begin{align*}
\int_0^{u_t}u^\beta \log(1+\ee^{-y-tu})\dd u & =-\int_0^{u_t}(y+tu)u^\beta\dd u+\int_0^{u_t}u^\beta \log(1+\ee^{y+t u})\dd u \\
& = \frac{1}{(\beta+1)(\beta+2)}\frac{|y|^{\beta+2}}{t^{\beta+1}}+\frac{1}{t^{\beta+1}}\int_0^{|y|}|v+y|^\beta\log(1+\ee^{-v})\dd v \\
& 
\begin{multlined}
= \frac{1}{t^{\beta+1}}\varphi_\beta^{(1)}(y)+\frac{1}{t^{\beta+1}}\int_{0}^{y+t\delta}|y+v|^\beta\log(1+\ee^{-v})\dd v 
\\ -\frac{1}{t^{\beta+1}}\int_{|y|}^{y+t\delta}|y+v|^\beta\log(1+\ee^{-v})\dd v,
\end{multlined}
\end{align*}
where for the second equality we changed variables $v=-y-tu$. Likewise,
$$
\int_{u_t}^\delta u^\beta \log(1+\ee^{-y-tu})\dd u =\frac{1}{t^{\beta+1}}\int_0^{y+t\delta}|y-v|^\beta \log(1+\ee^{-v})\dd v,
$$
and combining them
\begin{multline*}
\int_{0}^\delta u^\beta \log(1+\ee^{-y-tu})\dd u=\frac{1}{t^{\beta+1}}\varphi_\beta^{(1)}(y) +\\
\frac{1}{t^{\beta+1}} \left(\int_0^{y+t\delta}(|y-v|^\beta+|y+v|^\beta)\log(1+\ee^{-v})\dd v-\int_{|y|}^{y+t\delta}|v+y|^\beta\log(1+\ee^{-v})\dd v\right) 
\end{multline*}
Using that $y+t\delta \to +\infty$ and the bounds \eqref{laplacescaling2}, it is straightforward to see that the first integral on the right-hand side is exponentially close to $\varphi_k^{(2)}$ while the second integral is exponentially close to $\varphi_k^{(3)}$. The result then follows using these relations back into \eqref{eq:appeq1}--\eqref{eq:appeq2}.
\end{proof}

\bibliographystyle{abbrv}  
\bibliography{bibliography}

\end{document}